\documentclass[times,final]{elsarticle}

\usepackage{jcomp}
\usepackage{framed,multirow}

\usepackage{amssymb}
\usepackage{latexsym}

\usepackage{url}
\usepackage{xcolor}
\definecolor{newcolor}{rgb}{.8,.349,.1}

\journal{ }

\usepackage{xcolor}
\usepackage{url}
\usepackage{hyperref}
\usepackage{mathtools}

\usepackage{slashbox}
\usepackage{stfloats}
\usepackage{amsmath}
\usepackage{mathbbold}
\usepackage{amssymb}
\usepackage{enumitem}
\usepackage{amsthm}
\usepackage{graphicx}
\usepackage[version=4]{mhchem}
\usepackage{algorithm}
\usepackage{algorithmic}
\usepackage{cleveref}

\newtheorem{proposition}{Proposition}
\newtheorem{remark}{Remark}
\newtheorem{theorem}{Theorem}
\newtheorem{definition}{Definition}
\newtheorem{lemma}{Lemma}

\def\dd{\text{d}}

\begin{document}


\begin{frontmatter}

\title{Stochastic filtering for multiscale stochastic reaction networks based on hybrid approximations}%
\tnotetext[tnote1]{This work was funded by the Swiss National Science Foundation under grant number 182653. This work will be presented in part at the 59th IEEE Conference on Decision and Control, Jeju Island, Republic of Korea, 2020.}

\author[1]{Zhou {Fang}}
\author[1]{Ankit {Gupta}}
\author[1]{Mustafa {Khammash}\corref{cor1}}
\cortext[cor1]{Corresponding author: }
\ead{mustafa.khammash@bsse.ethz.ch}

\address[1]{Department of Biosystems Science and Engineering, ETH Zurich, Mattenstrasse 26, 4058 Basel, Switzerland.}


\begin{abstract}
In the past few decades, the development of fluorescent technologies and microscopic techniques has greatly improved scientists' ability to observe real-time single-cell activities.
In this paper, we consider the filtering problem associate with these advanced technologies, i.e., how to estimate latent dynamic states of an intracellular multiscale stochastic reaction network from time-course measurements of fluorescent reporters.
A good solution to this problem can further improve scientists' ability to extract information about intracellular systems from time-course experiments.

A straightforward approach to this filtering problem is to use a particle filter where particles are generated by simulation of the full model and weighted according to observations.
However, the exact simulation of the full dynamic model usually takes an impractical amount of computational time and prevents this type of particle filters from being used for real-time applications, such as transcription regulation networks. 
Inspired by the recent development of hybrid approximations to multiscale chemical reaction networks, we approach the filtering problem in an alternative way. 
We first prove that accurate solutions to the filtering problem can be constructed by solving the filtering problem for a reduced model that represents the dynamics as a hybrid process. 
The model reduction is based on exploiting the time-scale separations in the original network and, therefore, can greatly reduce the computational effort required to simulate the dynamics.
As a result, we are able to develop efficient particle filters to solve the filtering problem for the original model by applying particle filters to the reduced model. 
We illustrate the accuracy and the computational efficiency of our approach using several numerical examples.
\end{abstract}


\end{frontmatter}



\section{Introduction}

In the past few decades, scientists' ability to look into the dynamic behaviors of a living cell has been greatly improved by the fast development of fluorescent technologies \cite{zhang2002creating} and advances in microscopic techniques \cite{stephens2003light,vonesch2006colored}. 
Despite this big success, light intensity signals observed in a microscope can only report the dynamics of a small number of species in a cell, such as fluorescent proteins and mRNAs, and, therefore, leave other dynamic states of interest, e.g., gene (on/off) state, transcription factor abundance, or enzyme levels, indirectly observable.
Consequently, it is important to establish efficient stochastic filters to estimate latent states of intracellular biochemical reaction systems from these partial observations.

The filtering theory is a topic of active research in both control and statistics communities, which dates back to early works by Wiener \cite{wiener1949extrapolation} and Kalman \cite{kalman1960new}.
It is acknowledged that apart from some particular systems, e.g., linear systems with Gaussian noise and Markov chains with finite states, most nonlinear systems result in infinite-dimensional filters \cite{bain2008fundamentals}; therefore, obtaining numerical solutions to the filtering problem for nonlinear systems is a difficult task.
A powerful algorithm, called the particle filter (or sequential Monte Carlo method), was introduced in \cite{gordon1993novel}, which represent posterior distributions using a population of samples (also called particles) generated by importance sampling and resampling (see also \cite{doucet2009tutorial}).
Thanks to the ever-increasing computational power, a particle filter can solve the filtering problem efficiently in nonlinear non-Gaussian scenarios, and its convergence to the exact solution of the filtering problem is guaranteed under some mild conditions \cite{crisan2001particle,crisan2002survey,del2000branching,chopin2004central,hu2008basic,hu2011general,crisan2014particle,le2004stability}.

An intracellular reaction system is better modeled by a continuous-time Markov chain rather than an ordinary differential equation due to the presence of low copy number species, and it usually results in an infinite-dimensional filter which is not easy to calculate in a straightforward fashion.
In this scenario, the particle filter is a good candidate for solving the filtering problem.
The recent work in \cite{rathinam2020state} is good evidence of this point, where a particle filter was established based on exact partial state observations and shown to work well for many biological processes.
Despite the success of particle filters in solving the filtering problem for intracellular systems, the required heavy computational effort to simulate the system in the sampling step usually prevents it from being used for real-time applications.
Especially for multi-scale reaction systems where reactions fire at different timescales, the computational complexity of exact simulation methods (e.g., Gillespie stochastic simulation algorithm \cite{gillespie1976general,gillespie1977exact} and the next reaction method \cite{gibson2000efficient}) is proportional to the rate of the fastest reaction, and, hence, the algorithm takes a very long time to output an estimate.
A typical way to mitigate this type of problems is to replace the underlying system by a more tractable approximation so that a computationally efficient filter can be constructed based on the reduced model while the filter's accuracy is still preserved (see \cite[Section 2]{calzolari2006approximation}).
Following this idea, a method was proposed to obtain computationally efficient filters for biological systems by approximating the underlying Markov jump process by a diffusion process \cite{sun2008extended,chuang2013robust,liu2012state,calderazzo2019filtering}.
This method was shown to be efficient for broad classes of biochemical reaction systems; however, its drawback is also clear --- the diffusion process strategy loses its validity in low copy number scenarios.
Unfortunately, multi-scale reaction systems, which are common in systems biology \cite{kang2013separation}, usually involve species with low copy numbers; this fact precludes the usage of the diffusion approximation to these multi-scale systems and requires researchers to build new solutions. 

In this paper, we propose another strategy to obtain efficient particle filters based on hybrid approximations of multi-scale stochastic reaction networks.
The hybrid approximation technique treats the firing of fast reactions as continuous processes (see \cite{kang2013separation}) or removes them by quasi-stationary approximations \cite{cao2005slow,weinan2007nested}, as a consequence of which the computational complexity can be greatly reduced.
A systematic approach to constructing such an approximation was provided in \cite{kang2013separation}, and strategies for efficient simulation of these  approximations were introduced in \cite{crudu2009hybrid,hepp2015adaptive,duncan2016hybrid}.
Such a reduced model is called a hybrid approximation as some species exhibit jumping processes due to their low copy numbers, while others evolve continuously.
Therefore, this model is also called a piecewise deterministic Markov process as it is Markovian and governed by a deterministic differential equation between two neighboring jumping points.
Besides its application in running simulations, the hybrid approximation method can also be used to do sensitivity analysis for multi-scale stochastic reaction networks \cite{gupta2019sensitivity,gupta2014sensitivity}, which implies that this method has potential for broad applications. 
In this paper, we prove that the solution to the filtering problem for the original system can be constructed by solving the filtering problem for the hybrid approximation if some mild conditions are satisfied.
Notably, while the result seems intuitive, the mathematical analysis is non-trivial because the filtering problem focuses on the conditional probability rather than the {unconditional probability (i.e., the probability of the full system)}, and the convergence of the latter does not necessarily imply the convergence of the former. 
The obtained result then enables us to develop efficient particle filters to infer system parameters (reaction constants) and latent dynamic states of the original model instantaneously by applying particle filters to the reduced model. 

Compared with our previous work \cite{fang2020stochastic}, where system parameters are known and only discrete-time observations are considered, here we additionally consider the parameter uncertainty and the case of continuous-time observations.
We show that in the presence of parameter uncertainties, the proposed filter can efficiently solve the filtering problem in both the continuous-time observation case and the discrete-time observation case. 
Also, several numerical examples of gene expression networks are presented to illustrate our approach.

It is worth noting that the idea of applying time-scale separation techniques to developing computationally manageable particle filters has already been proposed in \cite{park2008problem,park2010dimensional,park2011particle} for diffusion type stochastic models, and its efficiency is proven in the literature \cite{imkeller2013dimensional}.
Compared with these works, our paper considers a different type of underlying models, the Markov jump process, and provides a much-simplified proof for the convergence of the approximate filters (based on the framework by \cite{calzolari2006approximation}).
All these references and our paper show the efficiency of applying the time-scale separation technique to the filtering problem for multi-scale systems. 

The rest of this paper is organized as follows.
In \Cref{Sec. preliminary}, we first introduce the filtering problem for a chemical reaction network system and, also, review the hybrid model reduction technique and the theory of particle filters. 
Then, our main results are presented in \Cref{Sec. main result}, stating that the filtering problem for a multi-scale reaction system can be accurately and computationally efficiently solved by a particle filter to its hybrid approximation.
Also, we provide the workflow of our approach in \Cref{Sec. main result} to guide practitioners to solve their specific problems.
To improve the readability of this paper, we put proofs of our main results in the {appendices}. 
In \Cref{Sec. numerical examples}, two numerical examples are presented to illustrate our approach.
Finally, \Cref{Sec conclusion} concludes this paper.

\section{Preliminary} {\label{Sec. preliminary}}

\subsection{Notations}
In this paper, we denote the natural filtered probability space by $( \Omega,  \mathcal F, \{\mathcal F_{t}\}_{t\geq 0}, \mathbb P)$, where $\Omega$ is the sample space, $\mathcal F$ is the $\sigma$-algebra, $\{\mathcal F_{t}\}_{t\geq 0}$ is the filtration, and $\mathbb P$ is the natural probability.
Also, we term $\mathbb N_{>0}$ as the set of positive integers, $\mathbb R^{n}$ with $n$ being a positive integer as the space of $n$-dimensional real vectors, $\|\cdot\|$ as the Euclidean norm, $|\cdot|$ as the absolute value notation, and $s\wedge t$ (for any $s,t \in \mathbb{R}$) as $\min(s,t)$. 
For any positive integer $n$ and any $t>0$, we term $\mathbb D_{ \mathbb R^{ n}}[0,t]$ as the Skorokhod space that consists of all $\mathbb R^n$ valued cadlag functions on $[0,t]$. {Here, cadlag functions refer to the functions that are right continuous and have a left limit.}

\subsection{Chemical reaction networks and their filtering problems}{\label{subsection chemical reaction network}}

In this paper, we consider an intracellular system undergoing reactions
\begin{equation*}
v_{1,j} S_1 + \dots + v_{n,j} S_n \ce{ ->[$k_{j}$]} v'_{1,j} S_1 + \dots + v'_{n,j} S_n, \quad j=1,\dots, r,
\end{equation*}
where $S_i$ ($i=1,\dots,n$) are distinguished species in the systems, $v_{i,j}$ and $v'_{i,j}$ are non-negative integers, called stoichiometric coefficients, $k_{j}$ are the reaction constants, and $r$ is the number of reactions.
Also, we name the linear combination of species (e.g, $v_{1,j} S_1 + \dots + v_{n,j} S_n$) as a complex.
Let $X(t)=\left(X_1(t),X_2(t), \dots, X_n(t)\right)^{\top}$ be the numbers of molecules of these species at time $t$, then the system's dynamics following mass-action kinetics can be expressed as \cite{anderson2011continuous}
\begin{equation*}
X(t)= X(0)+\sum_{j=1}^{r} (v'_{\cdot j}-v_{\cdot j}) R_j\left( \int_0^t \lambda_j (X(s))\dd s \right)
\end{equation*}
where $R_{j}(t)$ ($j=1,\dots,r$) are mutually independent unit rate Poisson processes, $\lambda_j (x)\triangleq k_{j} \prod_{i=1}^n \frac{x_i!}{\left(x_i-v_{i,j}\right)!}\mathbbold 1_{\{x_i\geq v_{i,j}\}}$ are propensities with $\mathbbold{1}_{\{\cdot\}}$ the indicator function,
and the initial condition $X(0)$ has a particular known distribution.

In a practical biochemical reaction system, different species can vary a lot in abundance, and rate constants can also vary over several orders of magnitude. 
{To emphasize this phenomenon, we reformulate the dynamical equation as follows using the normalized quantities of the aforementioned variables together with their scales. 
Following \cite{kang2013separation}, we first choose a large number $N$ (viewed as the scaling factor) and then normalize all the quantities by a power of $N$.
Specifically, we term the variable $\alpha_i$ (for $i=1,\dots, n$) satisfying $X^{N}_i(t)\triangleq N^{-\alpha_i}X_{i}(t)= O(1)$ as the magnitude of species $S_{i}$, the variable $\beta_j$ (for $j=1,\dots,r$) satisfying $k'_j\triangleq k_j N^{-\beta_j}= O(1)$ as the magnitude of the reaction constant $k_{j}$, the variable $\gamma$ as the time scale of interest, and $X^{N,\gamma}(t) \triangleq X^{N}(tN^{\gamma})$ as the normalized state in the time scale of interest.
The choice of the scaling factor $N$ is problem-specific; a suitable $N$ should clearly separate the scales of species numbers and rate constants. 
Though such a scaling factor might not always exists, it often does for multiscale reaction systems.
Also, in this paper, we only consider problems in a finite time window where the scales of the species and rate constants do not change over time. 
}
Finally, in the normalized coordinate, the dynamical equation can be expressed as
\begin{align}{\label{eq. scaling stochastic dynamics}}
X^{N,\gamma}(t) =& X^{N,\gamma}(0)   +\sum_{j=1}^{r} \Lambda^{N} (v'_{\cdot j}-v_{\cdot j}) R_j\left( \int_0^{t} N^{\gamma+\tilde \rho_j} \lambda^N_{j}(X^{N,\gamma}(s)) \dd s \right) 
\end{align}
where $\Lambda^{N}\triangleq \text{diag}(N^{-\alpha_1},\dots,N^{-\alpha_s})$,  $\tilde \rho_j\triangleq\beta_j + \sum_{i=1}^{n}v_{ij} \alpha_i$ is the time scales of reactions, and 
$\lambda^{N}_j (x)\triangleq k'_j \prod_{i=1}^{n} \prod_{\ell=0}^{v_{i,j}-1}  (x_{i} - \ell N^{-\alpha_i}) \mathbbold 1_{\left\{x_i\geq N^{-\alpha_i}v_{i,j}\right\}} $ is the normalized propensity and has the constant order.
Usually, the precise values of the initial condition $ X^{N}(0) $ and reaction constants $\mathcal K\triangleq \left(k'_1,\dots, k'_r\right)$ are unknown to researchers due to the variability of these parameters from cell to cell.
Therefore, in this paper, we consider these parameters to be subject to some probability distributions. 
For fixed $\mathcal K$, the infinitesimal generator of \eqref{eq. scaling stochastic dynamics}, denoted as $\mathcal L_{\mathcal K}^{N,\gamma}$, satisfies\footnote{Note that the terms $\lambda^{N}_j (x)$ depend on the choice of $\mathcal K$.}
\begin{equation*}
\mathcal L^{N,\gamma}_{\mathcal K} f(x)= 
\sum_{j=0}^{r} \lambda^{N}_j (x) N^{\gamma+\tilde \rho_{j}} \left[f\left(x+\Lambda^{N} (v'_{\cdot j}-v_{\cdot j})\right) - f(x)\right] 
\end{equation*}
for any bounded continuous function $f$ on $\mathbb R^{n}$.

We assume that some species in the system are fluorescent reporters, and $m$ channels of light intensity signals of these reporters can be observed from a microscopic platform (e.g., the one in \cite{rullan2018optogenetic}).
Mainly, two types of observations, continuous-time observations and discrete-time observations, are present in practical experiments depending on hardware properties.
We denote the continuous-time observations by $\dot Y^{N,\gamma}_{c}(\cdot)$ and the discrete-time observations by $Y^{N,\gamma}_d(\cdot)$ and, moreover, assume them to satisfy 
\begin{align}
Y^{N,\gamma}_{c}(t)&= \int_0^t h\left(X^{N,\gamma}(s)\right) \dd s  + B(t) \qquad &&\forall t\in[0,+\infty) \label{eq. continuous time observation}\\
Y^{N,\gamma}_{d}(t_i)&= h\left(X^{N,\gamma}(t_i)\right)  + W(t_i) \qquad && \forall i\in \mathbb N_{>0}, \label{eq. discrete time observation}
\end{align}
where $h$ is an $m$-dimensional bounded Lipschitz continuous function indicating the relation between the observation and the reaction process,  $B(t)$ is an $m$-vector of independent standard Brownian motions, $\{t_i\}_{i\in\mathbb N_{>0}}$ is a strictly increasing sequence of time points at which the discrete observation comes, and  $\{W_{\ell}(t_i)\}_{i\in\mathbb N_{>0}}$ are mutually independent $m$-variate Gaussian random variables whose covariance matrices are equal to the identity matrix.
{In \eqref{eq. continuous time observation} and \eqref{eq. discrete time observation}, the derivative $\dot Y^{N,\gamma}_{c}(t)$ $\left(=  h\left(X^{N,\gamma}(t)\right)   + \dot B(t)\right)$ and $Y^{N,\gamma}_{d}(t_i)$ model the light signals measured by a microscope in the continuous-time setting and the discrete-time setting, respectively.
	Since $\dot B(t)$ is mathematically tricky, we write the continuous-time observations using the integral form in \eqref{eq. continuous time observation} and use the accumulative quantity $Y^{N,\gamma}_{c}(t)$ later for inference.}
For these observations, we also assume that these Brownian motions and Gaussian random variables are independent of the Poisson processes $R_{j} (\cdot)$ ($j=1,\dots,r$) and system parameters $\mathcal K$ and $X^{N}(0)$; in other words, the noise in observations is independent of the underlying reaction system.
Also, we term $\mathcal Y^{N,\gamma}_{c,t}$ as the filtration generated by the continuous observation $\{Y^{N,\gamma}_{c}(s)\}_{0<s\leq t}$ and $\mathcal Y^{N,\gamma}_{d,t_i}$ as the filtration generated by the discrete observation $\{Y^{N,\gamma}_{d}(t_j)\}_{0<j\leq i}$.

The goal of the filtering problem is to infer latent states of the underlying system, e.g., $\phi(\mathcal K, X^{N,\gamma}(t))$ with $\phi$ a known measurable function from $\phi:\mathbb R^{r}\times \mathbb R^{n}$ to $ \mathbb R$, based on the collected observations. 
Specifically, it requires one to compute the conditional expectation $\pi^{N,\gamma}_{c,t}(\phi)\triangleq \mathbb E_{\mathbb P}\left[\phi\left(\mathcal K, X^{N,\gamma}(t)\right) \left | \mathcal Y^{N,\gamma}_{c,t}\right.\right]$ based on continuous-time observations, or the conditional expectation $\pi^{N,\gamma}_{d,t_i}(\phi)\triangleq \mathbb E_{\mathbb P}\left[\phi\left(\mathcal K, X^{N,\gamma}(t_i)\right) \left | \mathcal Y^{N,\gamma}_{d,t_i}\right.\right]$ based on discrete-time observations.
Also, we call these conditional expectations as true filters of the chemical reaction network system.
For a well-behaved function $\phi$, the conditional expectation $\pi^{N,\gamma}_{c,t}(\phi)$ satisfies the Kushner-Stratonovich equation \cite{bain2008fundamentals}
\begin{align*}
	\pi^{N,\gamma}_{c,t}(\phi)=&\pi^{N,\gamma}_{c,0}(\phi)+\int_0^t \pi^{N,\gamma}_{c,s}\left(\mathcal L^{N,\gamma} \phi\right) \dd s 
	+ \int_0^t \left[\pi^{N,\gamma}_{c,s}\left(\phi h^{\top} \right)-\pi^{N,\gamma}_{c,s}\left(h^{\top} \right)\pi^{N,\gamma}_{c,s}\left(\phi \right)\right]
	\left(\dd Y^{N,\gamma}_c-\pi^{N,\gamma}_{c,s}\left(h \right) \dd s \right)
\end{align*}
where $\mathcal L^{N,\gamma} \phi (\kappa,x)\triangleq \mathcal L^{N,\gamma}_{\kappa} \phi_{\kappa}( x) $ with $\phi_{\kappa}( \cdot)=\phi(\kappa, \cdot)$.
The equation is usually infinite-dimensional and has no explicit solution. 
In the discrete-time scenario, the conditional distribution satisfies a recursive expression \cite[Proposition 10.6]{bain2008fundamentals}, which also has no explicit solution in most cases.
As a result, we intend to find efficient algorithms that can quickly and accurately solve these filtering problems.

Before introducing our main results, we review the setup of hybrid approximations of multi-scale reaction systems and the theory of particle filters.

\subsection{Hybrid approximations of multi-scale chemical reaction networks} {\label{subsec hybrid approximation}}

\subsubsection{Hybrid approximations at the first timescale}

{By fixing $\alpha_i$ ($i=1,\dots,n$) and $\beta_j$ ($j=1,\dots,r$) and changing the scaling factor $N$, we obtain a one-parameter family of models $X^{N,\gamma}(\cdot)$.
To derive reduced models, we first assume the initial conditions of these models satisfy}
\begin{equation}{\label{eq. assumption initial conditons}}
	``\lim_{N\to \infty} X^{N}(0) \text{~exists $\mathbb{ P}$-almost surely, and its $i$-th entry is positive if $\alpha_i>0$.}"
\end{equation}
{which can be simply achieved by setting their initial conditions to be the same.}

Following the notations in \cite{kang2013separation}, we term $\Gamma_i^{+}\triangleq \{ j | v'_{i,j}>v_{i,j} \}$ and $\Gamma_i^{-}\triangleq \{ j | v'_{i,j}<v_{i,j} \}$. 
Then, the constant $\alpha_i-\max_{j\in \Gamma_i^{+}\cup \Gamma_i^{-}} \left(\beta_j+v_{\cdot j}^{\top} \alpha \right)$ is the timescale of the $i$-th species, i.e., the minimum $\gamma$ such that $X^{N,\gamma}_i (t)$ starts to evolve at rate of $O(1)$.
Moreover, we term $\gamma_1\triangleq \min_i
\left( \alpha_i-\max_{j\in \Gamma_i^{+}\cup \Gamma_i^{-}} \tilde \rho_j \right)$ as the parameter of the fastest timescale of the system and
$D^{\tilde \alpha} \triangleq \text{diag}( \mathbbold 1_{\{\alpha_1=\tilde \alpha\}}, \dots,  \mathbbold 1_{\{\alpha_n=\tilde \alpha\}})$ as a diagonal matrix indicating whether a species is at the scale of $N^{\tilde \alpha}$.

By neglecting all slow reactions ($\tilde \rho_j+\gamma_1<0$) and approximate fast reactions ($\tilde \rho_j+\gamma_1>0$) by a continuous process, one can arrive at a simplified dynamic model as follows.
\begin{align}{\label{eq. reduced model at the first time scale}}
X^{\gamma_{1}}(t)=& \lim_{N\to \infty} X^{N}(0) + \sum_{j : \gamma_1+\tilde \rho_j>0} \int_0^t \lambda'_j( X^{\gamma_1}(s))  D^{\gamma_1+\tilde \rho_j}\left(v'_{\cdot j}-v_{\cdot j}\right)\dd s  
 +  \sum_{j : \gamma_1+\tilde \rho_j=0 }
R_{j}\left(\int_0^t \lambda'_j( X^{\gamma_1}(s)) \dd s \right) D^0\left(v'_{\cdot j}-v_{\cdot j}\right) 
\end{align}
where $\lambda'_j(x)=\lim_{N\to+\infty} \lambda^N_j (x)$.
For any fixed $\mathcal K$, the infinitesimal generator of this process, denoted as $\mathcal L^{\gamma_1}_{\mathcal K}$, satisfies\footnote{Note that the propensities $\lambda'_j (x)$ depend on the choice of $\mathcal K$.}
\begin{align*}
	\mathcal L^{\gamma_1}_{\mathcal K} f(x)= 
	&\sum_{j : \gamma_1+\tilde \rho_j=0 } \lambda'_j (x) \left[f\left(x+D^0(v'_{\cdot j}-v_{\cdot j})\right) - f(x)\right]   
	+  	\sum_{j : \gamma_1+\tilde \rho_j>0 } \lambda'_j (x) \left[D^{\gamma_1+\tilde \rho_j}\left(v'_{\cdot j}-v_{\cdot j}\right)\right]^{\top}\nabla f(x)
\end{align*}
for any bounded continuously differentiable function $f$ on $\mathbb R$.
If we further assume that both $X^{N,\gamma_1}(\cdot)$ and $X^{\gamma_1}(\cdot)$ are non-explosive, i.e., 
 \begin{equation}{\label{eq. assumption infinite explosion time gamma1}}
 \lim_{c\to\infty}\tau^{N,\gamma_1}_{c}=\infty  \quad \text{ and } \quad
 \lim_{c\to\infty} \tau^{\gamma_1}_{c}=\infty   \quad  \text{$\mathbb{ P}$-a.s.}
 \end{equation}
 where $\tau^{N,\gamma_1}_{c}\triangleq \inf \{t \big| \| X^{N,\gamma_1}(t) \|\geq c \}$ and $\tau^{\gamma_1}_{c}\triangleq \inf \{t \big| \| X^{\gamma_1}(t) \|\geq c \}$, 
then $X^{N,\gamma_1}(\cdot)$ converges to $X^{\gamma_1}(\cdot)$ in distribution on any finite interval.
\begin{proposition}[Adapted from \cite{kang2013separation}]{\label{prop Kang kurtz gamma1}}
	If conditions \eqref{eq. assumption initial conditons} and \eqref{eq. assumption infinite explosion time gamma1} hold, then on any finite interval $[0,T]$, there holds $\left(\mathcal K, X^{N,\gamma_1}(\cdot )\right)\Rightarrow \left(\mathcal K,X^{\gamma_{1}}(\cdot)\right)$ in the sense of the Skorokhod topology. 
\end{proposition}
\begin{proof}
	The proof follows easily from \cite[Theorem 4.1]{kang2013separation}
\end{proof}

Note that the convergence result does not always hold on the infinite time interval;
one intuitive explanation is that the limit model may preclude bi-stability or some other phenomenon that the full model has (see \cite[Section VI]{gillespie2000chemical}).

\begin{remark}{\label{remark computation efficiency at the first time scale}}	
	The computational complexity to simulate $X^{\gamma_1}(\cdot)$ can be greatly lower than the complexity to simulate $X^{N, \gamma_1}(\cdot)$, because the former avoids the exact simulation of fast reactions ($\gamma+\beta_{j}+v^{\top}_{\cdot j} \alpha >0$), which consume a lot of computational resources to update the system at a rate proportional to  $N^{\gamma_1+\tilde \rho_j}$.
\end{remark}

\subsubsection{Hybrid model at the second timescale}
The second level of model reduction corresponds to the balance of complexes. 
Similar to the first timescale case, for any $\theta \in \mathbb{R}^{n}$, we term $\Gamma_\theta^{+}\triangleq \{j: \theta^{\top}(v'_{\cdot i}-v_{\cdot i})>0\}$ 
and $\Gamma_{\theta}^-\triangleq \{j: \theta^{\top}(v'_{\cdot i}-v_{\cdot i})<0\}$. 
The timescale of the complex $\sum_{i=1}^{n} \theta_i S_i$ is denoted by $\gamma_{\theta} \triangleq  \max_{i: \theta_i>0} \alpha_i -\max_{j\in \Gamma^-_{\theta} \cap \Gamma^+_{\theta}} \left(\beta_j+v_{\cdot j}^{\top} \alpha \right)$.
To avoid the degeneracy in the limit, we only consider those cases satisfying
\begin{align}{\label{eq. condition for the second time scale 1}}
	\max_{j\in \Gamma^-_{\theta} } \left(\beta_{j}+v_{\cdot j}^{\top} \alpha \right)= 	\max_{j\in \Gamma^+_{\theta} } \left(\beta_{j}+v_{\cdot j}^{\top} \alpha \right), 
	\quad \text{or} \quad
	\gamma \leq \gamma_{\theta}, \quad \forall \theta \in \mathbb{R}^{n}_{\geq 0},
\end{align}
which prevents the concentration of the complex from exploding or diminishing to zero as the system scale grows \cite{kang2013separation}.
Notably, the above condition always holds for  $\gamma=\gamma_1$ (see \cite{kang2013separation}).
If there exists a constant, $\gamma_2$, such that
\begin{equation}{\label{eq. assumption for the second time scale 2}}
	\gamma_2\triangleq \inf_{\theta\in \mathbb R^n_{\geq 0}} \{\gamma_{\theta}| \gamma_\theta >\gamma_1 \} >\gamma_1, \text{~and~} \gamma_2 \leq \sup\{\gamma | \eqref{eq. condition for the second time scale 1} \text{ holds} \},
\end{equation}
then we term  $\gamma_2$ as the second timescale for the system.

Let $e_{i}$ be a unit vector of $\mathbb{R}^{n}$ with $i$-th element being 1 and the rest being 0.
We denote $\mathbb L_1$ as the space spanned by $\mathbb S_1\triangleq \left\{e_{i} \big| \exists j \text{ s.t. } e_{j}^{\top} D^{r_1+\tilde \rho_j}(v'_{\cdot j}-v_{\cdot j})\neq 0\right\}$
and $\mathbb{L}_2$ as the space spanned by $\mathbb S_2=\left\{\theta \in \mathbb{R}^{n}_{\geq 0}\big| \theta ^{\top} D^{r_1+\tilde \rho_j} (v'_{\cdot j}-v_{\cdot j})= 0 \right\}$. 
Note that the subspaces $\mathbb L_1$ and $\mathbb L_2$ are not necessarily orthogonal, because their intersection can contain non-zero vectors.
At the second timescale, one can easily see that the species correspond to $\mathbb S_1$ fluctuate dramatically due to the reactions of the first timescale (i.e., $D^{\gamma_1+\tilde \rho_j}$ being non-zero), however, these reactions do not influence the slow complex corresponding to $\mathbb S_2$. 

We further denote 
the projection operator onto $\mathbb{L}_2$ by $\Pi_2$, and the identity operator by $I$.
For any fixed $\mathcal K\in\mathbb R^{r}_{>0}$ and $x_2\in\mathbb L_2$, a generator $\mathcal L^{\gamma_1}_{\mathcal K, x_2}$ with state space $\{x| x=(I-\Pi_2) \tilde x, ~ \Pi_2 \tilde x=x_2,~\tilde x\in\mathbb R^n  \}$ is defined by 
\begin{equation*}
	\mathcal L^{\gamma_1}_{\mathcal K, x_2} f(x_1)=  \mathcal L^{\gamma_1}_{\mathcal K} f(x_1+x_2)
\end{equation*}
for any bounded continuously differentiable function $f$ on $\mathbb R^{n}$.
Also, we assume that for any positive $x_2\in\mathbb L_2$ and almost every $\mathcal K$ with respect to its probability measure, there holds the condition that 
\begin{equation}{\label{eq. derivative of V1}}
	\text{  $\mathcal L^{\gamma_1}_{\mathcal K, x_2}$ has a unique stationary distribution $\bar V_1^{\mathcal K, x_2}$.}
\end{equation}
This assumption requires the fast varying species to be stable and, therefore, makes it possible to conclude their dynamic effects using a random measure. 
We further term a random measure $V^{\mathcal K}_1 (\dd x_1,\dd s)\triangleq \bar V_1^{\mathcal K, X^{\gamma_2}(s)}(\dd x_1) \dd s$.
Then, at the second timescale, a reduced model on the state space $\mathbb L_2$ can be derived as follows.
\begin{align}{\label{eq. reduced models at the second time scale}}
	X^{\gamma_2}(t)
	=&\lim_{N\to\infty} \Pi_2 X^{N}(0) 
    + \Pi_2\sum_{j : \gamma_2  +\tilde \rho_j>0} \left(\int_{\mathbb L_1 \times [0,t]} \lambda'_j\left( X^{\gamma_2}(s)+x_1\right) V^{\mathcal K}_1(\dd x_1, \dd s) \right) D^{\gamma_2  +\tilde \rho_j}(v'_{\cdot j}-v_{\cdot j}) \\
	&+ \Pi_2\sum_{j : \gamma_2  +\tilde \rho_j=0 } R_{j} \left(\int_{\mathbb{L}_1 \times[0,t]} \lambda'_j\left( X^{\gamma_2}(s)+x_1\right) V^{\mathcal K}_1(\dd x_1, \dd s) \right)
	D^{0}(v'_{\cdot j}-v_{\cdot j}). \notag
\end{align}

\begin{remark}{\label{remark computation efficiency at the second time scale}}
	Compared with the original model \eqref{eq. scaling stochastic dynamics}, the reduced model \eqref{eq. reduced models at the second time scale} not only replaces some fast reactions ($\gamma_2+\tilde \rho_j >0$ and $D^{\gamma_2  +\tilde \rho_j}\neq \mathbbold{0}_{n\times n}$\footnote{Here, $\mathbbold{0}_n$ means an $n\times n$ zero matrix.}) with a continuous process but also removes some fast reactions  ($\gamma_2+\tilde \rho_j >0$ and $D^{\gamma_2  +\tilde \rho_j}= \mathbbold{0}_{n\times n}$) from the model.
	As a result, the reduced model \eqref{eq. reduced models at the second time scale} can further accelerate the simulation of the target system.
\end{remark}

For the validity of the hybrid model at the second time scale (see \ref{eq. reduced models at the second time scale}), we need to introduce some additional assumptions. 
The first one is an ergodicity assumption on the fast-varying subsystem.
For a fixed $\mathcal K$, we define a random measure on $\mathbb{L}_1 \times [0,\infty)$ by 
\begin{equation*}
	V_1^{N,\gamma_2,\mathcal K}(A\times [0,t])=\int_0^t \mathbbold  1_{A}\left((I-\Pi_2)X^{N,\gamma_2}(s)\right) \dd s
\end{equation*}
where $\mathbbold  1_{A}$ is the indicator function. Also, we require that
\begin{equation}{\label{eq. assumption egordicity}}
	V_{1}^{N,\gamma_2,\mathcal K} \Rightarrow V^{\mathcal K}_1 ~ \text{almost surely.}
\end{equation}
The above convergence of random measures means that for almost every fixed $\mathcal K$, any bounded continuous function $f(\cdot)$, and all $t>0$, there holds
\begin{equation*}
	\int_{\mathcal L_1 \times [0,t]} f(x) V^{N,\gamma_2,\mathcal K}_1(\dd x \times \dd s) 
	\Rightarrow
	\int_{\mathcal L_1 \times [0,t]} f(x) V^{\mathcal K}_1(\dd x, \dd s).
\end{equation*}
Since the set of continuous functions with compact support is a subset of bounded functions, the above condition also suggests the convergence of these random measures in distribution (see \cite[Theorem 4.5]{kallenberg1974lectures}).
The other assumptions are technical and can be stated as follows.
Let $\tau^{N,\gamma_2}_{c}\triangleq \inf \{t \big| |\Pi_2 X^{N,\gamma_2}(t) |\geq c \}$ and 
$$l_{c}(y)\triangleq \sup \left\{\sum_{j\in \{\Pi_2 D^{\gamma_2+\tilde \rho_j} (v'_{\cdot j}-v_{\cdot j})\neq 0\}} \lambda^{N}_{j}(x) ~ \Bigg | ~|\Pi_2 x| \leq c, x-\Pi_2x=y \right\}.  $$
We assume that for each $c$ there exists a  function $\Psi_c : \mathbb R_{\geq 0} \to \mathbb R_{\geq 0}$ satisfying $\lim_{r\to \infty} r^{-1} \Psi_c(r)=\infty$ such that
\begin{align}{\label{eq. assumption technical 1}}
	&\left\{
	\int_{\mathbb{L}_1 \times[0,t\wedge \tau_c^{N,\gamma_{2}}]} \Psi_c(l_{c}(y))V_1^{N,\gamma_2,\mathcal K} (\dd y\times \dd s)
	\right\} \text{is stochastically bounded.} 
\end{align}
and
\begin{align}\label{eq. assumption technical 2}
	&\sum_{j} \left| N^{\gamma_2+\tilde \rho_j} \Lambda^{N} (I-D^{\gamma_1+\tilde \rho_j}-D^{\gamma_2+\tilde \rho_j})(v'_{\cdot j}-v_{\cdot j})\right|
     \times \int_{\mathbb{L}_1 \times[0,t\wedge \tau_c^{N,\gamma_{2}}]}
	\lambda^{N}_{j}\left(\Pi_2 X^{N,\gamma_2}(s)+y\right) V_{1}^{N,\gamma_2,\mathcal K}(\dd y \times \dd s) \to 0. 
\end{align}
Finally, under the above assumptions and non-explosivity of $X^{N,\gamma_2}$ and $X^{\gamma_2}$, i.e.,
\begin{equation}{\label{eq. assumption infinite explosion time gamma2}}
	\lim_{c\to\infty} \tau^{N,\gamma_2}_{c}=\infty, \quad \text{ and } \quad \lim_{c\to\infty}  \tau^{\gamma_2}_{c}=\infty, \quad \text{a.s.}
\end{equation}
where $\tau^{\gamma_2}_{c}\triangleq \inf \{t \big| \| X^{\gamma_2}(t) \|\geq c \}$, the process  $X^{N,\gamma_2}$ converges in distribution to $X^{\gamma_2}$ on any finite time interval.
The details are listed in the following proposition.

\begin{proposition}[Adapted from \cite{kang2013separation}]{\label{prop kang kurtz gamma2}}
	If conditions \eqref{eq. assumption initial conditons}, \eqref{eq. assumption for the second time scale 2}, \eqref{eq. derivative of V1}, \eqref{eq. assumption egordicity}, \eqref{eq. assumption technical 1}, \eqref{eq. assumption technical 2}, and \eqref{eq. assumption infinite explosion time gamma2} are satisfied, then on any finite time interval $[0,T]$, there holds  $\left(\mathcal K, \Pi_2 X^{N,\gamma_2}(\cdot) \right)\Rightarrow \left(\mathcal K, X^{\gamma_2}(\cdot) \right)$ in the sense of the Skorokhod topology. 
\end{proposition}
\begin{proof}
	The proof follows easily from \cite[Theorem 4.4 and Theorem 5.1]{kang2013separation}.
\end{proof}

Among all these assumptions in the above result, \eqref{eq. assumption egordicity}, \eqref{eq. assumption technical 1}, and \eqref{eq. assumption technical 2} are the most demanding ones.
Verifying them usually requires one to construct proper Lyapunov functions for the fast dynamics (see \cite[Remark 4.19]{gupta2014sensitivity}) and the full dynamics (see \cite[Lemma 5.3]{kang2013separation}).
Moreover, since system parameters are random (but fixed over time) in our problem, the  construction of Lyapunov functions should be robust to the choice of system parameters so that the criteria can work for almost every $\mathcal K$.
The method proposed in \cite{gupta2014scalable} where the vector norm serves as the Lyapunov function can be a candidate for solving such a problem.

\subsubsection{Hybrid approximations at a higher timescale}
The key to generating a hybrid approximation of a multiscale system is to assume quasi-equilibrium for the fast sub-network \cite[Section 2.3]{gupta2014sensitivity}.
Therefore, following this rule, one can construct a hierarchy of reduced models to depict dynamic behaviors at different timescales.
For our main results, we assume ourselves in the situation of Proposition \ref{prop Kang kurtz gamma1} and Proposition \ref{prop kang kurtz gamma2} only, i.e., the first and second timescales.
However, these results can be straightforwardly extended to a higher timescale using the same proof scheme shown in the {appendices}.

\subsubsection{Observations for reduced models}

We can also define filtering problems for the reduced models \eqref{eq. reduced model at the first time scale} and \eqref{eq. reduced models at the second time scale} as follows.
We first define artificial readouts for these reduced models by 
\begin{align*}
& Y^{\gamma_\ell}_{c}(t)= \int_0^t h\left(X^{\gamma_\ell}(s)\right) \dd s  + B(t),
\quad\text{and}
\quad
 Y^{\gamma_\ell}_{d}(t_i)= h\left(X^{\gamma_\ell}(t_i)\right)  + W(t_i), 
&& \ell=1,2
\end{align*}
and term 
$\mathcal Y^{\gamma_\ell}_{c,t}$ and $\mathcal Y^{\gamma_\ell}_{d,t_i}$ ($\ell=1,2$)
as the filtrations generated by $\{Y^{\gamma_\ell}_{c}(s)\}_{0<s\leq t}$ and $\{Y^{\gamma_\ell}_{d}(t_j)\}_{0<j\leq i}$,
respectively.
Then, for reduced models \eqref{eq. reduced model at the first time scale} and \eqref{eq. reduced models at the second time scale}, the filtering problem requires one to compute the conditional expectations 
\begin{equation*}
\pi^{\gamma_\ell}_{c,t}(\phi)\triangleq \mathbb E_{\mathbb P}\left[\phi\left(\mathcal K, X^{\gamma_\ell}(t)\right) \left | \mathcal Y^{\gamma_\ell}_{c,t}\right.\right]
\quad \text{and} \quad
\pi^{\gamma_\ell}_{d,t_i}(\phi)\triangleq \mathbb E_{\mathbb P}\left[\phi\left(\mathcal K, X^{\gamma_\ell}(t_i)\right) \left | \mathcal Y^{\gamma_\ell}_{d,t_i}\right.\right]
\end{equation*}
for $\ell=1,2$ based on the corresponding continuous-time and discrete-time observations, respectively.

Notably, since the reduced model does not exist in reality, one can never collect these ``imaginary" readouts $Y^{\gamma_1}_{c}(\cdot)$, $Y^{\gamma_1}_{d}(\cdot)$, $Y^{\gamma_2}_{c}(\cdot)$, or $Y^{\gamma_2}_{d}(\cdot)$ in real experiments.
Therefore, solving the filtering problems for the reduced models has no much practical meaning.
The value of these artificial filtering problems lies in helping us construct an efficient particle filter to solve the filtering problems for the original model, which will be discussed in-depth in \Cref{Sec. main result}.

\subsection{Particle filters}{\label{subsection observations}}
The particle filter, also known as the sequential Monte Carlo method, is a protocol used to infer latent dynamic states and system parameters by generating samples that mimic dynamic behaviors of the underlying systems. 
In the sequel, we separately review continuous-time particle filters and discrete-time ones, including their derivations and specific algorithms.

\subsubsection{Continuous-time particle filters}
By denoting a Girsanov random variable $$Z^{N,\gamma}_{c}  (t)\triangleq \exp\left( \int_{0}^{t} h^{\top}(X^{N,\gamma}(s)) \dd Y^{N,\gamma}_{c}(s) - \frac{1}{2} \int_0^{t} \|h(X^{N,\gamma}(s))\|^2 \dd s \right),$$ whose reciprocal is a martingale under $\mathbb P$ \cite[Lemma 3.9]{bain2008fundamentals}, and a reference probability $\left.\frac{\dd \mathbb P_c ^{N,\gamma}}{\dd \mathbb P}\right|_{\mathcal F_{t}}\triangleq \left(Z^{N,\gamma}_c(t)\right)^{-1}$, one can arrive at the Kallianpur-Striebel formula \cite[Proposition 3.16]{bain2008fundamentals} (also see \cite{kallianpur1968estimation})
\begin{equation}{\label{eq. Kallianpur-Striebel formula for the full model continuous observation}}
	\pi_{c,t}^{N,\gamma}(\phi) \triangleq 
	{\mathbb E_{\mathbb P^{N,\gamma}_{c}} \left[ Z^{N,\gamma}_{c}(t) \phi\left(\mathcal K, X^{N,\gamma}(t)\right) \left| \mathcal Y^{N,\gamma}_{c,t} \right.\right]}
	~\Bigg/~{\mathbb E_{\mathbb P^{N,\gamma}_{c}} \left[ Z^{N,\gamma}_{c}(t) \left| \mathcal Y^{N,\gamma}_{c,t} \right.\right]}
	\qquad
	\mathbb P\text{-a.s}~\text{and}~\mathbb P^{N,\gamma}_c\text{-a.s},
\end{equation}
which transforms the filtering problem under the natural probability into the filtering problem under the reference probability.
Since under the reference probability, the observations are independent of the underlying system and become an $m$-vector of independent standard Brownian motions \cite[Proposition 3.13]{bain2008fundamentals}, the computation of conditional expectations under the reference probability is much simpler than the computation under the original probability $\mathbb P$.
We can also construct reference probabilities for reduced models, which orthogonalizes\footnote{Here, ``orthogonalize" means making two random variables independent of each other.} the observations and the underlying systems and provide Kallianpur-Striebel formulas for them.
The details are shown in \ref{Sec change of measure methods}.

Based on \eqref{eq. Kallianpur-Striebel formula for the full model continuous observation}, a continuous-time sequential importance resampling (SIR) particle filter can be constructed as \Cref{alg continuous time regularized particle filters}, which generates samples under the reference probability and uses Monte Carlo method to solve the filtering problem.
Specifically, in \Cref{alg continuous time regularized particle filters}, the weights $w_j(t)$ (for $j=1,\dots,M$) mimic normalized $Z^{N,\gamma}_c(t)$, and particles $(\kappa_j(t), x_j(t))$ mimic the state $(\mathcal K, X^{N,\gamma}(t))$; thus, the empirical sum of $\phi(\kappa_j(t), x_j(t))$ with respect to weights $w_j(t)$, approximates the true filter by the Kallianpur-Striebel formula \eqref{eq. Kallianpur-Striebel formula for the full model continuous observation}.
Notably, resampling is executed in each iteration to delete non-significant particles and mitigate the long-term sample impoverishment at the cost of adding additional noise at the current step \cite{doucet2009tutorial}.
The residual resampling scheme is usually preferred in the resampling step, as it generally outperforms other schemes, e.g., multinomial resampling scheme, in the sense of asymptotic variances \cite{chopin2004central}.

\begin{algorithm}
	\caption{The continuous-time SIR particle filter \cite[Chapter 9]{bain2008fundamentals}}
	\label{alg continuous time regularized particle filters}
	
	\begin{algorithmic}[1]
		\STATE  Input continuous-time observations $Y(\cdot)$, the model of the underlying system, and a joint distribution of the initial states and reaction constants; set $\delta$ to be a positive constant indicating the length of sub-time intervals.
		\STATE  Initialization: Sample $M$ particles $(\bar \kappa_1(0), \bar x_1(0)),\dots, (\bar \kappa_M(0), \bar x_M(0))$ from the input distribution, and set $i=1$.
		\WHILE{$i\delta$ does not exceed the terminal time of the observations}
		\STATE  \textit{Sampling:} Simulate 
		$x_j(\cdot)$ ($j=1,\dots,M$) in the time interval $((i-1)\delta, i\delta]$ according to the underlying model with parameters $\bar \kappa_j((i-1)\delta) $ and initial conditions $\bar x_j ((i-1)\delta)$;
		Set $\kappa_j(t)=\bar \kappa_j ((i-1)\delta)$ for $t\in((i-1)\delta, i\delta] $;
		Calculate weights $w_j(t)\propto z_j(t) \triangleq \exp\left( \int_{(i-1)\delta}^t h^{\top} \left( x_j(s) \right) \dd Y(s) - \frac{1}{2}\int_{(i-1)\delta}^t \| h(x_j(s))\|^2 \dd s \right).$
		\STATE  \textit{Output the filter:} 
		$\bar \pi_{M,c,t} (\phi)  = {\sum_{j=1}^{M} w_j(t) \phi(\kappa_j, x_j(t))}$ for $t\in ((i-1)\delta, i\delta]$.
		\STATE  \textit{Resampling:} Resample
		$\{ w_{j}(i\delta), (\kappa_j(i\delta), x_j(i\delta))\}$ to obtain M equally weighted samples $\{{1}/{M}, (\bar \kappa_j(i\delta), \bar x_j(i\delta))\}$; 
		\ENDWHILE
	\end{algorithmic}
\end{algorithm}

\subsubsection{Discrete-time particle filters}
For discrete-time observations, we can also define a random variable 
$Z^{N,\gamma}_d(t_i)\triangleq \prod_{j=1}^{i} H\left(X^{N,\gamma}(t_j),Y^{N,\gamma}_d(t_j)\right)$, in which $H(x,y)= \exp\left(h^{\top}(x) y-\frac{1}{2}\|h(x)\|^2\right)$, and whose reciprocal is a martingale under $\mathbb P$.\footnote{The martingale property can be easily checked by calculating the characteristic function of the random variable.} 
Then, a reference probability, $ \mathbb P_d ^{N,\gamma}$, can be constructed by $\left.\frac{\dd \mathbb P_d ^{N,\gamma}}{\dd \mathbb P}\right|_{\mathcal F_{t}}\triangleq \left(Z^{N,\gamma}_c(t)\right)^{-1}$, under which the underlying system and the discrete-time observation are independent of each other, and $\{Y^{N,\gamma}_d(t_i)\}_{i\in \mathbb N_{>0}}$ are mutually independent $m$-variate Gaussian random variables with covariance matrices being the identity matrix.
Using the Kallianpur-Striebel formula \cite{kallianpur1968estimation}, we can express the filter like the following
\begin{equation}{\label{eq. Kallianpur-Striebel formula for the full model discrete observation}}
\pi_{d,t_i}^{N,\gamma}(\phi) \triangleq 
{\mathbb E_{\mathbb P^{N,\gamma}_{d}} \left[ Z^{N,\gamma}_{d}(t) \phi\left(\mathcal K, X^{N,\gamma}(t_i)\right) \left| \mathcal Y^{N,\gamma}_{d,t_i} \right.\right]}
~\Bigg/~{\mathbb E_{\mathbb P^{N,\gamma}_{d}} \left[ Z^{N,\gamma}_{d}(t_i) \left| \mathcal Y^{N,\gamma}_{d,t_i} \right.\right]} \qquad \mathbb P\text{-a.s}~\text{and}~\mathbb P^{N,\gamma}_d\text{-a.s},
\end{equation}
which transforms the original filtering problem into the filtering problem under the reference probability.
For the reduced models, we can also construct such reference probabilities and Kallianpur-Striebel formulas, which is discussed in detail in \ref{Sec change of measure methods}.

Based on \eqref{eq. Kallianpur-Striebel formula for the full model discrete observation}, a discrete-time SIR particle filter can be constructed as \Cref{alg discrete time regularized particle filters} to solve the filtering problem numerically.
\Cref{alg discrete time regularized particle filters} can be viewed as a discrete-time analog of \Cref{alg continuous time regularized particle filters}, where the only difference is substituting the discrete-time arguments in \Cref{alg discrete time regularized particle filters} for the continuous ones in \Cref{alg continuous time regularized particle filters}.

\begin{algorithm}
	\caption{The discrete-time SIR particle filter \cite{gordon1993novel,doucet2009tutorial}}
	\label{alg discrete time regularized particle filters}
	
	\begin{algorithmic}[1]
		\STATE  Input observations $\{Y(t_i)\}_{i\in\mathbb N_{>0}}$, a dynamical model, and a initial distribution;
		\STATE  \textit{Initialization:} The same as the initialization step in \Cref{alg continuous time regularized particle filters}; Also, $t_0=0$.
		\WHILE{$t_i$ does not exceed the terminal time of the observations}
		\STATE  \textit{Sampling:} simulate $x_j(\cdot)$ ($j=1,\dots,M$) from time $t_{i-1}$ to $t_i$ according to the underlying model with parameters $\bar \kappa_j(t_{i-1}) $ and initial conditions $\bar x_j (t_{i-1})$;  Set $\kappa_j (t_{i})=\bar \kappa_j (t_{i-1})$; Calculate weights $w_j(t_i)\propto z_j(t_i)\triangleq  H(x_{j}(t_i), Y(t_i))$.
		\STATE  \textit{Output the filter:}  $\bar \pi_{M,d,t_i} (\phi)  = {\sum_{j=1}^{M} w_j(t_i) \phi(\kappa_j, x_j(t_i))}$.
		\STATE  \textit{Resampling:} 
		The same as the resampling in \Cref{alg continuous time regularized particle filters} (change $i\delta$ to $t_i$);
		\ENDWHILE
	\end{algorithmic}
\end{algorithm}

\section{Main results}\label{Sec. main result}

The goal of this paper is to construct computationally efficient algorithms to solve filtering problems for multi-scale reaction network systems, i.e., calculating $\pi^{N,\gamma}_{c,t}(\phi)$ and $\pi^{N,\gamma}_{d,t_i}(\phi)$. 
A straightforward idea to solve this problem is to use the particle filters {that utilize the full dynamic model \eqref{eq. scaling stochastic dynamics}}. However, the computational cost of simulating \eqref{eq. scaling stochastic dynamics} is extremely expensive (see \Cref{remark computation efficiency at the first time scale}), and, therefore, these particle filters are computationally inefficient.
Consequently, we need to figure out a smarter way to approach this problem. 

We first demonstrate that solutions to the filtering problems for original models can be constructed by the solutions to the filtering problems for the reduced model.
Note that $\pi^{\gamma_1}_{c,t}(\phi)$, $\pi^{\gamma_1}_{d,t_i}(\phi)$, $\pi^{\gamma_2}_{c,t}(\phi)$ and $\pi^{\gamma_2}_{d,t_i}(\phi)$ are, respectively, $\mathcal Y^{\gamma_1}_{c,t}$, $\mathcal Y^{\gamma_1}_{d,t_i}$, $\mathcal Y^{\gamma_2}_{c,t}$, and $\mathcal Y^{\gamma_2}_{d,t_i}$ measurable;
hence, there exist measurable functions $\hat f^{\gamma_1}_{\phi,c,t}(\cdot)$, $\hat f^{\gamma_1}_{\phi, d,t_i}(\cdot)$, $\hat f^{\gamma_2}_{\phi,c,t}(\cdot)$, and $\hat f^{\gamma_2}_{\phi,d,t_i}(\cdot)$ such that 
\begin{align*}
	&\pi^{\gamma_\ell}_{c,t}(\phi)= \hat f^{\gamma_\ell}_{\phi,c,t}(Y^{\gamma_\ell}_{c,0:t})
	&\text{and}&& \pi^{\gamma_\ell}_{d,t_i}(\phi)= \hat f^{\gamma_\ell}_{\phi,d,t_i}(Y^{\gamma_\ell}_{d,1:i}), 
	&& \mathbb P\text{-a.s.},
	&& \forall \ell\in\{1,2\},
\end{align*}
where $Y^{\gamma_\ell}_{c,0:t}$ ($\ell=1,2$) is the trajectory of $Y^{\gamma_{\ell}}_{c}(\cdot)$ from time 0 to $t$, and $Y^{\gamma_\ell}_{d,1:i}\triangleq \left( Y^{\gamma_\ell}_d(t_1),\dots, Y^{\gamma_\ell}_d(t_i) \right)$ for $\ell=1,2$.
Similarly, there also exist measurable functions $\hat f^{N, \gamma}_{\phi,c,t}(\cdot)$ and $\hat f^{N, \gamma}_{\phi, d,t_i}(\cdot)$ (for $\gamma\in\mathbb R$) such that
\begin{align*}
&\pi^{N, \gamma}_{c,t}(\phi)= \hat f^{N, \gamma}_{\phi,c,t}\left(Y^{N, \gamma}_{c,0:t}\right)
&\text{and}&& \pi^{N, \gamma}_{d,t_i}(\phi)= \hat f^{N, \gamma}_{\phi,d,t_i}\left(Y^{N, \gamma}_{d,1:i}\right), 
&& \mathbb P\text{-a.s.},
\end{align*}
where $Y^{N,\gamma}_{c,0:t}$ is the trajectory of $Y^{N,\gamma}_{c}(\cdot)$ from time 0 to $t$, and $Y^{N,\gamma}_{d,1:i}\triangleq \left( Y^{N,\gamma}_d(t_1),\dots, Y^{N,\gamma}_d(t_i) \right)$.
In this paper, we name these functions filtering maps, as they map the observations to the solutions of filtering problems.
Based on these functions, we define artificial filters like the following
\begin{align*}
	&\tilde \pi^{N,\gamma_\ell}_{c,t} (\phi) \triangleq  \hat f^{\gamma_\ell}_{\phi,c,t}\left(Y^{N, \gamma_\ell}_{c,0:t}\right)
	&& \text{and}
	&& \tilde \pi^{N,\gamma_\ell}_{d,t_i} (\phi) \triangleq  \hat f^{\gamma_\ell}_{\phi,d,t_i}\left(Y^{N, \gamma_\ell}_{d,1:i}\right) 
	&& \forall \ell\in\{1,2\},
\end{align*}
which is constructed by plugging the observations of the full model into the filter maps of the reduced models. 
We can show that under some mild conditions, these artificial filters can be arbitrarily close to the exact filters of the original model as $N$ goes to infinity.
The rigorous statement of this result is presented as follows.

\begin{theorem}{\label{thm convergence of artificial filters}}
	~
	\begin{enumerate}
		\item Assume that \eqref{eq. assumption initial conditons} and \eqref{eq. assumption infinite explosion time gamma1} are satisfied. Then, for and any bounded continuous function $\phi$, there hold
		\begin{equation*}
			\hat f^{N,\gamma_1}_{\phi,c,t}\left(Y^{N, \gamma_1}_{c,0:t}\right)-\hat f^{\gamma_1}_{\phi,c,t}\left(Y^{N, \gamma_1}_{c,0:t}\right)
			\stackrel{\mathbb P}{\to} 0 \qquad \text{as} \quad N\to \infty,
			\qquad \forall t>0,
		\end{equation*}
		and 
		\begin{equation*}
			\hat f^{N,\gamma_1}_{\phi,d,t_i}\left(Y^{N, \gamma_1}_{d,1:i}\right)-\hat f^{\gamma_1}_{\phi,d,t_i}\left(Y^{N, \gamma_1}_{d,1:i}\right)
			\stackrel{\mathbb P}{\to} 0 \qquad \text{as} \quad N\to \infty,
			\qquad \forall i\in\mathbb N_{>0}.
		\end{equation*}
		\item Assume that conditions \eqref{eq. assumption initial conditons}, \eqref{eq. assumption for the second time scale 2}, \eqref{eq. derivative of V1}, \eqref{eq. assumption egordicity}, \eqref{eq. assumption technical 1}, \eqref{eq. assumption technical 2}, and \eqref{eq. assumption infinite explosion time gamma2} are satisfied, and, moreover, $h(x)=h(\Pi_2 x)$ $\forall x\in\mathbb R^{n}$.
		Then, for any bounded continuous function $\phi$ such that $\phi(\kappa,x)=\phi(\kappa, \Pi_2 x)$ $\forall (\kappa,x)\in\mathbb R^{r}\times\mathbb R^{n}$,
		there hold
		\begin{equation*}
		  \hat f^{N,\gamma_2}_{\phi,c,t}\left(Y^{N, \gamma_2}_{c,0:t}\right)-\hat f^{\gamma_2}_{\phi,c,t}\left(Y^{N, \gamma_2}_{c,0:t}\right)
		  \stackrel{\mathbb P}{\to} 0 \qquad \text{as} \quad N\to \infty,
		   \qquad \forall t>0,
		\end{equation*}
		and
		\begin{equation*}
		\hat f^{N,\gamma_2}_{\phi,d,t_i}\left(Y^{N, \gamma_2}_{d,1:i}\right)-\hat f^{\gamma_2}_{\phi,d,t_i}\left(Y^{N, \gamma_2}_{d,1:i}\right)
		\stackrel{\mathbb P}{\to} 0 \qquad \text{as} \quad N\to \infty,
		\qquad \forall i\in\mathbb N_{>0}.
		\end{equation*}
 	\end{enumerate}
\end{theorem}
\begin{proof}
	The proof is shown in \ref{Sec the convergence of theoretical filters}.
\end{proof}

The above theorem indicates that the filtering problem for the original model can be transformed to the problem of computing artificial filters $\tilde \pi^{N,\gamma_1}_{c,t} (\phi) $, $\tilde \pi^{N,\gamma_1}_{d,t_i} (\phi)$, $\tilde \pi^{N,\gamma_2}_{c,t} (\phi)$, and $\tilde \pi^{N,\gamma_2}_{d,t_i}(\phi)$.
Note that these artificial filters plug observations of the original models into filter maps of the reduced models. 
Therefore, one idea of calculating them is to construct particle filters where the observations of original models and the dynamics of the hybrid approximations are inserted. 
The corresponding particle filters are defined as follows.

\begin{definition}{\label{def particle filters}}
	~
	\begin{itemize}
		\item Let $\bar \pi^{N,\gamma_\ell}_{M,c,t}(\phi)$ ($\ell=1,2$) be the output of the particle filter \Cref{alg continuous time regularized particle filters} where $Y^{N,\gamma_\ell}_{c}(\cdot)$, the dynamic model of $X^{\gamma_\ell}(\cdot)$, and the joint distribution of $\left(\mathcal K, \lim_{N\to\infty}X^{N}(0)\right)$  (for $\ell=1$) or $\left(\mathcal K, \lim_{N\to\infty}\Pi_2X^{N}(0)\right)$ (for $\ell=2$) are inserted.
		\item Let $\bar \pi^{N,\gamma_\ell}_{M,d,t_i}(\phi)$ ($\ell=1,2$) be the output of the particle filter \Cref{alg discrete time regularized particle filters} where $Y^{N,\gamma_\ell}_{d}(\cdot)$, the dynamic model of $X^{\gamma_\ell}(\cdot)$, and the joint distribution of $\left(\mathcal K, \lim_{N\to\infty}X^{N}(0)\right)$  (for $\ell=1$) or $\left(\mathcal K, \lim_{N\to\infty}\Pi_2X^{N}(0)\right)$ (for $\ell=2$) are inserted.
	\end{itemize}
\end{definition}

Under some mild conditions, the above-defined particle filters converge to the corresponding artificial filters (see \Cref{thm convergence of particle filters}), and, therefore, can accurately approximate the true solutions to the filtering problems for the original model (see \Cref{thm main results}).

\begin{theorem}{\label{thm convergence of particle filters}}
	Let the SIR particle filters defined in \Cref{def particle filters} utilize the residual resampling scheme in the resampling step.
	If the process $X^{\gamma_{\ell}}(\cdot)$ ($\ell\in\{1,2\}$) is almost surely nonexplosive,
	then for any bounded continuous function $\phi(\kappa, x)$,
	there hold
	\begin{align*}
	&\mathbb E_{\mathbb P}\left[ \left|\bar \pi^{N,\gamma_\ell}_{M,c,t}(\phi)-
	\hat f^{\gamma_{\ell}}_{\phi,c,t}\left(Y^{N,\gamma_{\ell}}_{c,0:t} \right)\right| \right]
	\leq \frac{ \tilde c_t}{\sqrt M} \| \phi \|_{\infty}
	&& (\forall t>0)
	&&\text{and}
	&&\mathbb E_{\mathbb P}\left[ \left|\bar \pi^{N,\gamma_\ell}_{M,d,t_i}(\phi)-
	\hat f^{\gamma_{\ell}}_{\phi,d,t_i}\left(Y^{N,\gamma_{\ell}}_{d,1:i} \right)\right| \right]
	\leq \frac{ \tilde d_{t_i}}{\sqrt M} \| \phi \|_{\infty}
	&& \forall (i\in\mathbb N_{>0})
	\end{align*}
	where $\tilde c_t$ and $\tilde d_{t_i}$ are time dependent scalars.
\end{theorem}
\begin{proof}
	The convergence of discrete-time particle filters follows immediately from \cite[Theorem 1]{chopin2004central}.
	The convergence of continuous-time particle filters is shown in \ref{section convergence of particle filters}.
\end{proof}

\begin{theorem}{\label{thm main results}}
	Let the SIR particle filters defined in \Cref{def particle filters} utilize the residual resampling scheme in the resampling step.
	\begin{enumerate}
		\item  If \eqref{eq. assumption initial conditons} and \eqref{eq. assumption infinite explosion time gamma1} are satisfied, then for any bounded continuous function $\phi(\kappa, x)$, 
		there hold
		\begin{align*}
			&\lim_{N\to\infty} \lim_{M\to\infty} \mathbb P 
			\left(  
			   \left|\bar \pi^{N,\gamma_1}_{M,c,t}(\phi) - \pi^{N,\gamma_1}_{c,t}(\phi)\right| \leq \delta 
			\right) =1, 
			&& \forall t>0, ~ \forall \delta>0,\\
			&\lim_{N\to\infty} \lim_{M\to\infty} \mathbb P 
			\left(  
			\left|\bar \pi^{N,\gamma_1}_{M,d,t_i}(\phi) - \pi^{N,\gamma_1}_{d,t_i}(\phi)\right| \leq \delta 
			\right) =1 ,
			&& \forall i\in \mathbb N_{>0}, ~ \forall \delta>0.
		\end{align*}
		\item Assume that conditions \eqref{eq. assumption initial conditons}, \eqref{eq. assumption for the second time scale 2}, \eqref{eq. derivative of V1}, \eqref{eq. assumption egordicity}, \eqref{eq. assumption technical 1}, \eqref{eq. assumption technical 2}, and \eqref{eq. assumption infinite explosion time gamma2} are satisfied, and, moreover, $h(x)=h(\Pi_2 x)$ for all $x\in\mathbb R^{n}$.
		Then, for any bounded continuous function $\phi(\kappa, x)$ such that
		$\phi(\kappa, x)=\phi(\kappa, \Pi_2 x)$ for all $(\kappa,x)\in\mathbb R^{r}\times \mathbb R^{n}$, there hold
		\begin{align*}
		&\lim_{N\to\infty} \lim_{M\to\infty} \mathbb P 
		\left(  
		\left|\bar \pi^{N,\gamma_2}_{M,c,t}(\phi) - \pi^{N,\gamma_2}_{c,t}(\phi)\right| \leq \delta 
		\right) =1, 
		&& \forall t>0, ~ \forall \delta>0,\\
		&\lim_{N\to\infty} \lim_{M\to\infty} \mathbb P 
		\left(  
		\left|\bar \pi^{N,\gamma_2}_{M,d,t_i}(\phi) - \pi^{N,\gamma_2}_{d,t_i}(\phi)\right| \leq \delta 
		\right) =1 ,
		&& \forall i\in \mathbb N_{>0}, ~ \forall \delta>0.
		\end{align*}
	\end{enumerate}
\end{theorem}
\begin{proof}
	It follows immediately from \Cref{thm convergence of artificial filters} and \Cref{thm convergence of particle filters}.
\end{proof}

\Cref{thm main results} states that errors between the proposed particle filters (see \Cref{def particle filters}) and the corresponding solutions to the target filtering problems are very likely to be small, provided a large scaling factor $N$ and a large particle population size $M$.
Therefore, the proposed particle filters can solve the filtering problems of the original model accurately.
Apart from the accuracy, these SIR particle filters consume much less computational resources than the ones {constructed by the full dynamical model}, as the former ones only require the simulation of a reduced model in the sampling step, which is computationally much cheaper than the simulation of the full model (see \Cref{remark computation efficiency at the first time scale} and \Cref{remark computation efficiency at the second time scale}).
Both facts suggest that the particle filters proposed in \Cref{def particle filters} based on hybrid approximations are efficient in solving the target filtering problems.

Compared with the model reduction technique via time-scale separations (\Cref{prop Kang kurtz gamma1} and \Cref{prop kang kurtz gamma2}), our main result for the filtering problem (\Cref{thm main results}) requires two more additional conditions, first that the function $\phi(\cdot)$ should be bounded, and second that functions $h(\cdot)$ and $\phi(\cdot)$ should be irrelevant to the fast subnetwork (for the second time scale only).
The boundedness of $\phi(\cdot)$ is a result of the boundedness condition required for the convergence of filters (see \Cref{thm convergence of artificial filters} and \Cref{thm convergence of particle filters}).
This condition is a very common requirement in filtering theory, as it guarantees the true filter to be well-defined and simplifies analyses. 
However, it is not generally satisfied for biological applications, where concentrations of most target species have no theoretical upper bounds apart from some exceptions, e.g., gene copies.
To tackle this problem, we can truncate the target quantity by a large number beyond which the conditional probability is comparatively low, and the tail event contributes little to the conditional expectation.
Then, an estimate of the truncated quantity generated by a particle filter proposed in \Cref{def particle filters} can provide an accurate approximation to the conditional expectation of the target quantity.

The other additional condition that $h(\cdot)$ and $\phi(\cdot)$ are irrelevant to the fast subnetwork does not place many obstacles to applying the proposed particle filters because of the following reasons.
First, most fluorescent reporters are large protein molecules which typically have slower dynamics than smaller molecules, such as enzymes and amino acids.
Consequently, the function $h(\cdot)$, which represents the abundance  of fluorescent reporters, is usually irrelevant to the fast fluctuating subnetwork.
Meanwhile, researchers are inclined to estimate the state of large molecules, e.g., proteins or DNAs, rather than small molecules, as a result of which the function $\phi(\cdot)$ is also independent of the fast subnetwork in most cases.

We conjecture that both of the additional conditions in \Cref{thm main results} can be relaxed, and, therefore, the proposed method can cover a much broader class of multi-scale reaction systems.
In the proof (see \ref{Sec the convergence of theoretical filters}), we frequently use the portmanteau lemma to show the convergence of expectations, which leads to the boundedness condition of $\phi$.
Therefore, by considering the convergence rate of the hybrid approximation (see \cite{enciso2019constant}), it could be possible to extend the portmanteau lemma to work for unbounded functions and, thus, relax the boundedness condition of $\phi$.
On the other hand, the fast subnetwork state is very much shaped by the state of the slow subnetwork (see \eqref{eq. assumption egordicity}).
Therefore, it is reasonable to think that we can instantaneously get the knowledge of the fast subnetwork by inferring the slow one, which can potentially relax the constraint that $h(\cdot)$ and $\phi(\cdot)$ are irrelevant to the fast subnetwork.  
However, to rigorously justify all these points, we need to have more elaborate analyses of the problem, because of which we leave it for future work.

We end this section with a simple algorithm that summarizes the procedures to estimate latent states of a multi-scale reaction system using the proposed SIR particle filters. 

\begin{algorithm}
	\caption{The procedure to estimate latent states of a multi-scale reaction network system using SIR particle filters based on hybrid approximations}
	\label{alg work-flow}
	\begin{algorithmic}[1]
		\STATE Describe the multi-scale reaction system in terms of scaling parameters $\alpha_i$-s and $\beta_j$-s and choose the time scale of interest  (see \Cref{subsection chemical reaction network}). Preferably, the chosen time scale is $\gamma_1$ or $\gamma_2$.
		\STATE Construct a hybrid approximation of the target system using \eqref{eq. reduced model at the first time scale} or \eqref{eq. reduced models at the second time scale}.
		\STATE  Check whether the quantity of interest, $\phi(\mathcal K, X^{N,\gamma}(t))$, has a theoretical bound. If not, truncate the function by a large number beyond which the system is not likely to reach.
		\STATE  Based on the type of the observation and the chosen timescale, pick an appropriate particle filter in \Cref{def particle filters} to solve the filtering problem. The performance of this filter is guaranteed by \Cref{thm main results}.
	\end{algorithmic}
\end{algorithm}

\section{Numerical examples}{\label{Sec. numerical examples}}
In this section, we illustrate our approach using two numerical examples, a simple gene expression model and a transcriptional regulation network.
{Such gene expression models are prevalent in systems and synthetic biology and captures key biophysical phenomena observed in lab experiments.}
Since the burst kinetics in gene expression networks are highly gene specific  \cite{suter2011mammalian}, the identifiability is usually guaranteed in these systems. 
Meanwhile, a hybrid approximation based approach has already been successfully applied to the parameter inference problem for gene expression systems \cite{herbach2017inferring}. Thus, it is natural to think that our approach that also utilizes hybrid approximations should be applicable to the filtering problem for these gene models. 

All algorithms mentioned in this section are implemented on Matlab and executed on a high performance computing cluster with {15-core}, 2.25GHz AMD EPYC 7742 processors.
All code applied to perform the analysis is available  at ``github.com/ZhouFang92/Particle-filters-for-multi-scale-CRNs-v.1".


\subsection{A simple gene expression model}\label{sec simple gene expression model}

\begin{figure}[h!]
	\centering
	\includegraphics[width= 0.7 \textwidth]{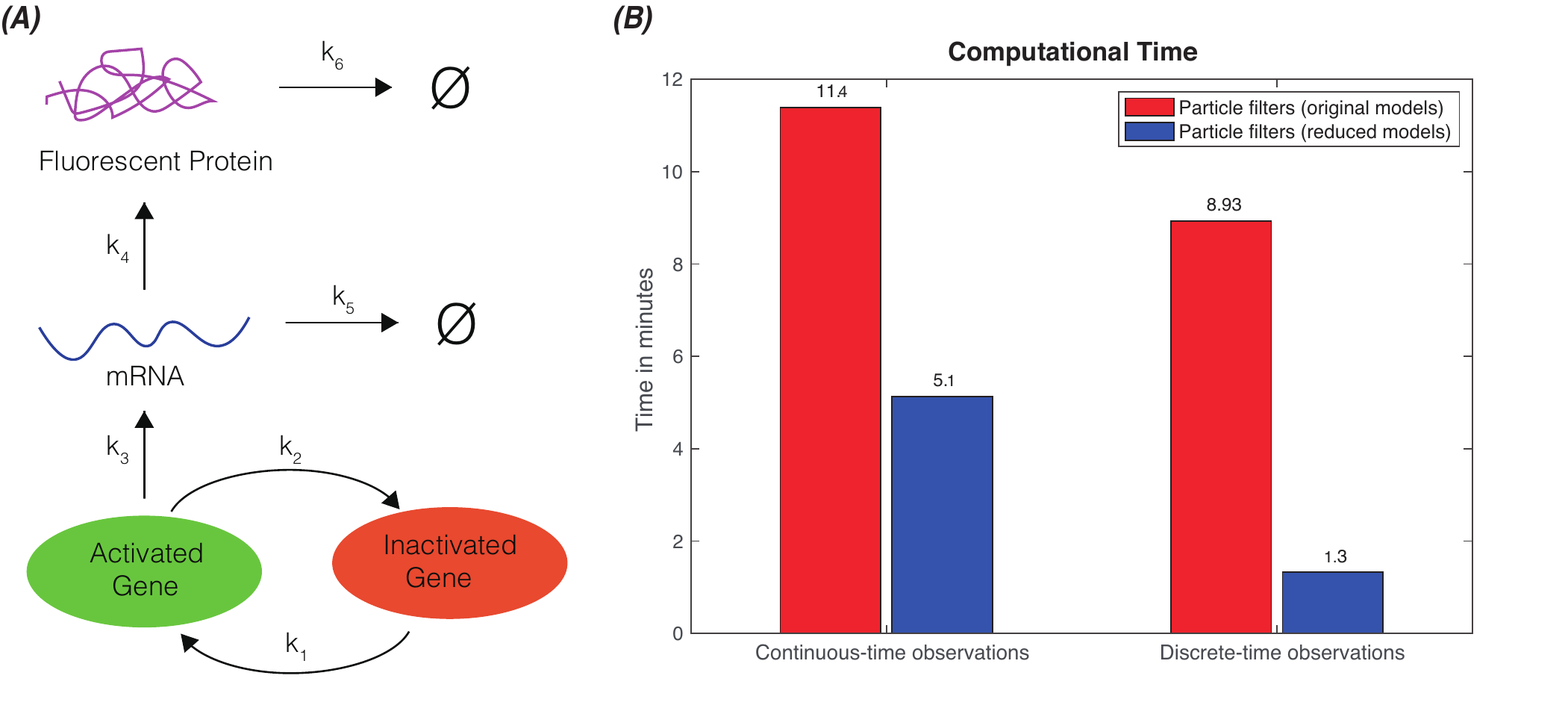}
	\caption{
		(A) shows the biological circuit of the considered network. (B) shows the computational time consumed by different particle filters. {The filters were run on a 2.25GHz AMD processors with 15 cores.}
	}
	\label{fig. ex1_cpu_time}
\end{figure}

We first consider a gene expression model (see {\Cref{fig. ex1_cpu_time}} (A)) {that contains the basic patterns in the central dogma of molecular biology.
This model is also known as a telegraph model, which has been widely used in the study of gene transcriptional dynamics; also, it can capture important biophysical phenomena (e.g., transcriptional bursting \cite{suter2011mammalian,rullan2018optogenetic}).} 
The model involves four species: a gene having on and off states (denoted respectively as $S_2$ and $S_1$), the mRNA it transcribes (denoted as $S_3$), and a fluorescent protein it expresses (denoted as $S_4$); besides, it has six reactions: 
\begin{align*}
S_1 &\ce{->[$k_1$]} S_2 && \text{gene activation},
&S_2 &\ce{->[$k_2$]} S_1 && \text{gene deactivation}, \\
S_2 &\ce{->[$k_3$]} S_3 + S_2 && \text{gene transcription},
&S_3	&\ce{ ->[$k_{4}$]}  S_3+S_4  && \text{translation},\\
S_3 &\ce{->[{$k_5$}]} \emptyset && \text{mRNA degradation},
& S_4 &\ce{ ->[$k_{6}$]} \emptyset && \text{protein degradation.}
\end{align*}
For this system, we take $N=100$, $\alpha_4=1$, and $\alpha_i=0$ for $1\leq i \leq 3$, i.e., the cellular system consists of hundreds of fluorescent protein molecules but very few copies of other molecules. 
The values of reaction constants and their scaling exponents are shown in {\Cref{table reaction constants}}.
With this setting, we can easily calculate that $\gamma_1=0$, and the second time scale does not exist.
For initial conditions, we assume $X^N_{1}(0)$ to have a binary distribution with mean $1/3$, $X^{N}_2(0)$ to satisfy $X^{N}_2(0)=1- X^{N}_1(0)$, $X^{N}_3(0)$ to have a Poisson distribution with mean 2, and $X^{N}_4(0)$ to also have a Poisson distribution with mean 2.
Also, we assume that all reaction constants and initial conditions except $X_2(0)$  are independent of each other.

\begin{table}[h]
	\centering
	\begin{tabular}{|cr|cr|cr|cr|}
		\hline
		\multicolumn{2}{c}{Reaction constants {(minute$^{-1}$)}} & \multicolumn{2}{c}{Exponents} &  \multicolumn{2}{c}{Scaled rates {(minute$^{-1}$)}} &  \multicolumn{2}{c}{Reaction scale}  \\
		\hline
		$k_1$ & $\mathcal U(1 \times 10^{-2}, 2 \times 10^{-2})$  & $\beta_1$  &  0&   $k'_1$ & $\mathcal U (0.01, 0.02)$ &  $\beta_{1}+ \alpha_1$ & 0\\
		$k_2$ &  $\mathcal U(7\times 10^{-3}, 1 \times 10^{-2} )$ & $\beta_2$ &  $0$&  $k'_2$& $\mathcal U (0.007,0.01)$ & $\beta_{2} + \alpha_2$ & 0 \\
		$k_3$ &  $\mathcal U (7\times 10^{-1}, 9\times 10^{-1})$&  $\beta_3$&  0&  $k'_3$& $\mathcal U (0.7, 0.9)$ & $\beta_{3} + \alpha_2$ & 0 \\
		$k_4$ & $\mathcal U (3\times 10^{+1}, 4\times 10^{+1})$ & $\beta_4$ &  1 &  $k'_4$& $\mathcal U (0.3, 0.4)$& $\beta_{4} + \alpha_3$ & 1 \\
		$k_5$ &  $\mathcal U (1\times 10^{-1},3\times 10^{-1})$ &  $\beta_5$&  $0$&  $k'_5$& $\mathcal U (0.1,0.3)$ & $\beta_{5} + \alpha_3$& 0\\
		$k_6$ &  $\mathcal U ( 3\times 10^{-1},4 \times 10^{-1})$&  $\beta_6$&  0&  $k'_6$& $\mathcal U (0.3, 0.4)$& $\beta_{6} + \alpha_4$ & 1 \\
		\hline
	\end{tabular}
	\caption{Scaling exponents for reaction rates of the gene expression model: {$\mathcal U$ is the notation for the uniform distribution.}}
	\label{table reaction constants}
\end{table}

By \eqref{eq. scaling stochastic dynamics}, we can easily derive the full dynamical model at the fastest time scale $\gamma_1$ as follows.
	\begin{align*}
		X^{N,\gamma_1}_1(t)=& X^{N,\gamma_1}_1(0)- R_1 \left( k'_1  \int_0^t X^{N,\gamma_1}_{1} (s) \dd s \right)  + R_{2}\left(k'_{2} \int_0^t X^{N,\gamma_1}_2(s) \dd s\right)\\
		X^{N,\gamma_1}_2(t)=& X^{N,\gamma_1}_2(0)+ R_1 \left( k'_1  \int_0^t X^{N,\gamma_1}_{1} (s) \dd s \right)
		- R_{2}\left(k'_{2} \int_0^t X^{N,\gamma_1}_2(s) \dd s\right)  \\
		X^{N,\gamma_1}_3(t)=&X^{N,\gamma_1}_3(0)+R_{3} \left(k'_3\int_0^t X^{N,\gamma_1}_2(s) \dd s\right)  -R_5\left(k'_5\int_0^t X^{N,\gamma_1}_{3}(s) \dd s\right) \\ 
		X^{N,\gamma_1}_4(t)=& X^{N,\gamma_1}_4(0) + N^{-1}R_4\left(k'_4  N \int_0^t X^{N,\gamma_1}_3(s) \dd s\right)  - N^{-1}R_6\left(k'_6 N  \int_0^t X^{N,\gamma_1}_4(s) \dd s\right). 
\end{align*}Moreover, by \eqref{eq. reduced model at the first time scale}, the reduced dynamical model satisfies
\begin{align*}
X^{\gamma_1}_1(t)=& \lim_{N\to \infty}X^{N,\gamma_1}_1(0)- R_1 \left( k'_1  \int_0^t X^{\gamma_1}_{1} (s) \dd s \right)+ R_{2}\left(k'_{2} \int_0^t X^{\gamma_1}_2(s) \dd s\right)\\
X^{\gamma_1}_2(t)=& \lim_{N\to \infty} X^{N,\gamma_1}_2(0)+ R_1 \left( k'_1  \int_0^t X^{\gamma_1}_{1} (s) \dd s \right)  - R_{2}\left(k'_{2} \int_0^t X^{\gamma_1}_2(s) \dd s\right)  \\
X^{\gamma_1}_3(t)=&\lim_{N\to \infty} X^{N,\gamma_1}_3(0) 
+R_{3} \left(k'_3\int_0^t X^{\gamma_1}_2(s) \dd s\right) -R_5\left(k'_5\int_0^t X^{\gamma_1}_{3}(s) \dd s\right) \\
X^{\gamma_1}_4(t)=& \lim_{N\to \infty} X^{N,\gamma_1}_4(0) 
+   \int_0^t k'_4 X^{N,\gamma_1}_3(s) \dd s - \int_0^t k'_6  X^{N,\gamma_1}_4(s) \dd s.
\end{align*}
In this example, we also assume that light intensity signals can be observed continuously or discretely from a microscope, satisfying \eqref{eq. continuous time observation} or \eqref{eq. discrete time observation}
where $h(x)= \left(10 \times x_4 \right) \wedge 10^{3}$ with $10^3$ being the measurement range, $B(t)$ is a Brownian motion, $t_i=2i$ (i.e., discrete-time observations come every two minutes), and $\{W(t_i)\}_{i\in\mathbb N_{>0}}$ is a sequence of mutually independent standard Gaussian random variables.

\begin{figure}[h!]
	\centering
	\includegraphics[width= 0.9 \textwidth]{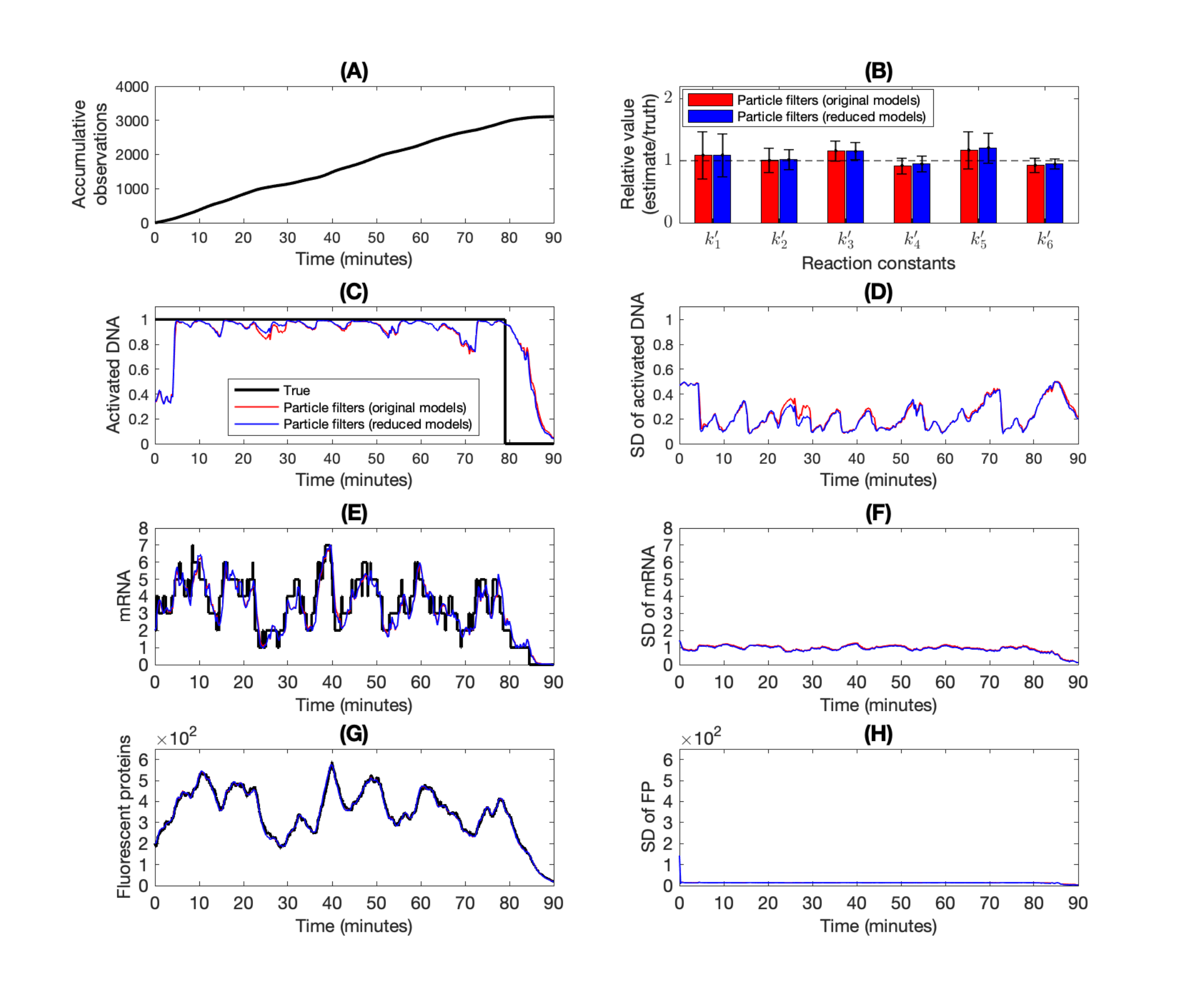}
	\caption{Simulation results for the simple gene expression model with continuous-time observations: {Panel (A) shows the accumulative observation signal $Y^{N,\gamma}_{c}(t)$, which is the integral of the time-course measurement from time 0 to time $t$. Panel (B) shows the performance of both filters in estimating reaction constants, where error bars represent 95\% confidence intervals. In this panel, we draw the relative values of these estimates $\left(\frac{\text{the estimate}}{\text{the true parameter}}\right)$ so that all the true parameters are rescaled to 1.}
		The rest of the panels compare performances of different particle filters in estimating dynamic states, where black lines are the true values of the underlying system, red lines are the estimates by the particle filter using the full dynamic model, and the blue lines are the estimates by the particle filter using the reduced dynamic model.
		Specifically, (C), (E), and (G) show the estimates of the conditional means of Activated DNAs, mRNAs, and fluorescent proteins, respectively; 
		(D), (F), and (H) present the estimates of the conditional standard deviations of these quantities.
	}
	\label{fig ex1_continuous}
\end{figure}

In this numerical example, we first randomly chose a set of system parameters, simulated the system for 90 minutes, and generated observations for both continuous-time and discrete-time scenarios. 
Then, based on these observations, we used particle filters that applies the hybrid approximations to infer dynamic states and reaction constants. 
Throughout this experiment, we set particle population to be {100,000}.
Meanwhile, we took the particle filter that applies the full dynamic model as a benchmark. 

\begin{figure}[h!]
	\centering
	\includegraphics[width= 0.9 \textwidth]{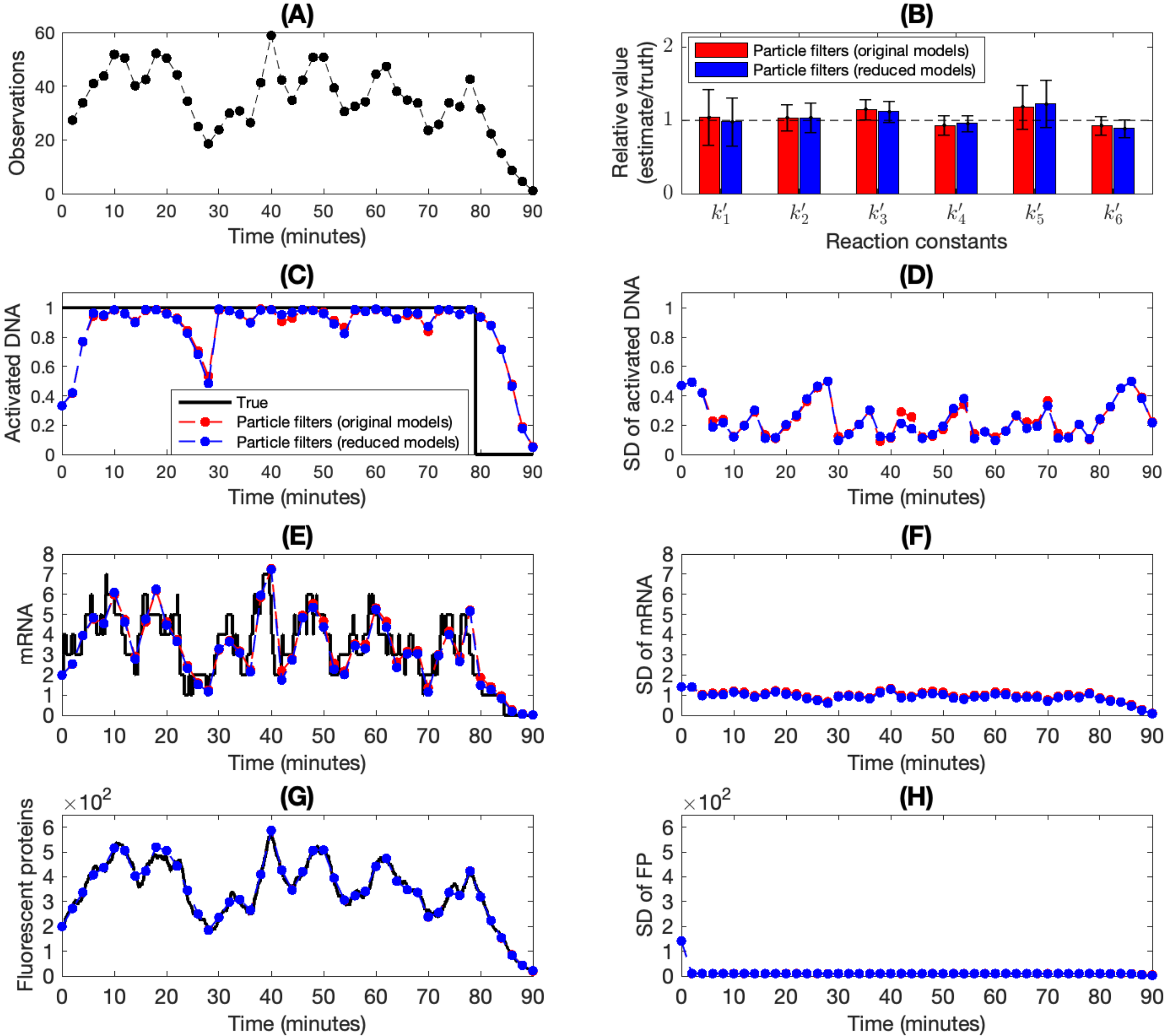}
	\caption{Simulation results for the simple gene expression model with discrete-time observations.
		The meaning of each panel is the same as the corresponding one in {\Cref{fig ex1_continuous}, except that panel (A) draws the raw data of the observation instead of the accumulative observation signal.}
	}
	\label{fig ex1_discrete}
\end{figure}

\begin{figure}[h!]
	\centering
	\includegraphics[width= 0.70 \textwidth]{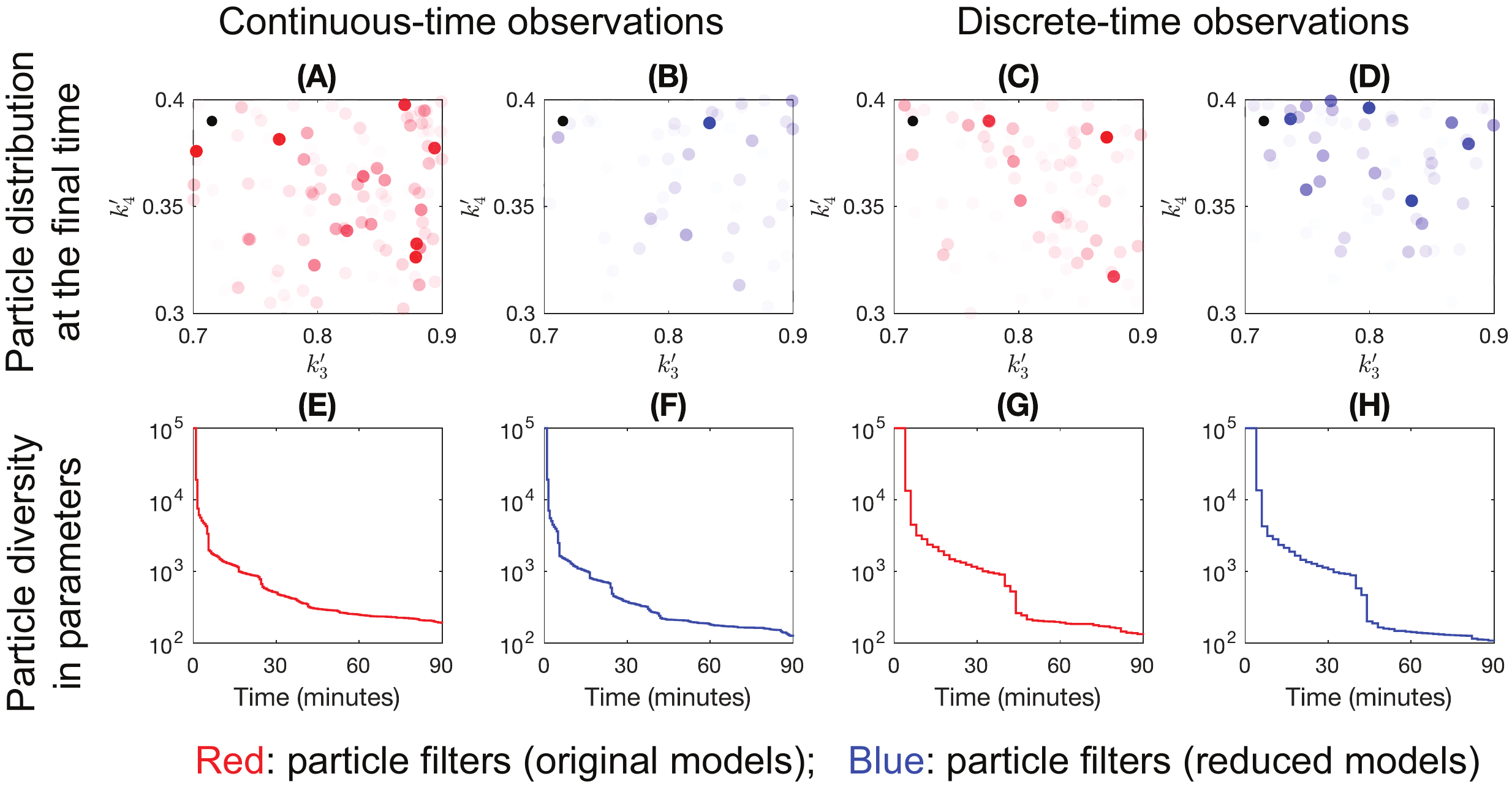}
	\caption{
	{Particle distributions in the simple gene expression example. 
			The first two columns show the results for the continuous-time observation case, and the last two columns show the results for the discrete-time observation case. 
			The first row shows particle distributions in $k'_3$--$k'_4$ plane at the final time, where the black dot shows the true value, and the colored dots show the particles from the filters.
			The transparency of a colored dot indicates the number of particles located in this specific site; the more transparent it is, the fewer particles are located. 
			The second row shows the time evolution of the particle diversity in parameters (i.e., the number of distinguishable particles for estimating model parameters). 
			This figure tells that for both observation types and both particle filters, the particle diversity in parameters decays dramatically over time. Finally, hundreds of distinguishable particles are left, and their distributions are not very consistent.}
	}
	\label{fig ex1_particle_distribution}
\end{figure}

Numerical simulation results are presented in {\Cref{fig. ex1_cpu_time}, \Cref{fig ex1_continuous}, and \Cref{fig ex1_discrete}}.
From {\Cref{fig. ex1_cpu_time}}, we see that particle filters applying reduced models is much faster than particle filters applying full dynamical models.
{Moreover, both kinds of filters provide very consistent estimates to the dynamical states in both continuous-time and discrete-time scenarios (see \Cref{fig ex1_continuous} and \Cref{fig ex1_discrete});
their relative $L_1$ distances in estimating different parts of the system are summarized in \Cref{table ex1 L2 distance}.}
\begin{table}[h]
	\centering
	\begin{tabular}{|c|cc|cc|cc|}
		\hline
		\multirow{2}{*}{}            & \multicolumn{2}{c|}{Activated DNA}                                                                 & \multicolumn{2}{c|}{mRNA}                                                                & \multicolumn{2}{c|}{Protein}                                                             \\ \cline{2-7} 
		& \multicolumn{1}{c|}{Mean} & \begin{tabular}[c]{@{}c@{}}Standard\\ deviation\end{tabular} & \multicolumn{1}{c|}{Mean} & \begin{tabular}[c]{@{}c@{}}Standard\\ deviation\end{tabular} & \multicolumn{1}{c|}{Mean} & \begin{tabular}[c]{@{}c@{}}Standard\\ deviation\end{tabular} \\ \hline
		Continuous-time observations & \multicolumn{1}{c|}{1.21\%}     &               6.17\%                                               & \multicolumn{1}{c|}{2.00\%}     &           4.65\%                                                   & \multicolumn{1}{c|}{0.25\%}     &                   4.91\%                                           \\ \hline
		Discrete-time observations   & \multicolumn{1}{l|}{1.16\%}     &                   6.44\%                                           & \multicolumn{1}{l|}{3.76\%}     &                      7.12\%                                        & \multicolumn{1}{l|}{0.07\%}     &        1.55\%                                                      \\ \hline
	\end{tabular}
    \caption{{Relative $L_1$ distance between both kinds of particle filters in estimating the simple gene expression system. The relative $L_1$ distance between two signals $f_1(t)$ and $f_2(t)$ is defined by $\frac{\|f_1-f_2\|_{L_1}}{\|f_1\|_{L_1}}$.  The results show that these two kinds of particle filters provide very consistent estimate to the dynamical states with relative errors no greater than 8\%.} }
    \label{table ex1 L2 distance}
\end{table}
{Besides, we also observe from \Cref{fig ex1_continuous} and \Cref{fig ex1_discrete} that both particle filters accurately track the trajectories of the mRNA and protein, and they also successfully identify the change of the gene state after a short time period.
The relatively slow convergence of the estimate to the gene state is attributed to the long time delay between the DNA dynamics and the observation processes, which makes the observation less informative about the DNA dynamics. 
Particularly, in the beginning, the filters are not certain about the gene state until they identify the growth of mRNA copies at about the sixth minute, which suggests that the gene has already been activated.
This delay also occurs at the final time when the gene is deactivated --- the filters can not be sure of the gene state until they identify a sharp decrease of the mRNA copies at a time point when the gene has been shut down for a while.  
}

{Moreover, both filters also provide consistent and accurate estimates to model parameters (see \Cref{fig ex1_continuous}(B) and \Cref{fig ex1_discrete}(B)).
Despite this success, in these filters, we still observe a well-known phenomenon, called sample degeneracy, that the particle diversity in parameters drops quickly over time (see the second row of \Cref{fig ex1_particle_distribution}). 
This degeneracy usually has adverse effects on parameter inference. 
Fortunately, in this case study, there are still hundreds of distinguishable particles left at the final time.
Therefore, sample degeneracy only slightly affects the estimate to the posterior of model parameters (see the first row of \Cref{fig ex1_particle_distribution}) and does not affect their mean estimates very much (see \Cref{fig ex1_continuous}(B) and \Cref{fig ex1_discrete}(B)).
We leave the problem of how to combat sample degeneracy for future work. 
}

{Finally, we investigate the performance of our filter when the observation noise is non-Gaussian. 
In real experiments, though the scale of observation noise can be reliably quantified,  identifying its probability distribution can be difficult. 
When the Gaussian-noise assumption is violated, the weight updating step which utilizes the Gaussian likelihood might be problematic and introduce extra errors to the estimate.
To consider this factor, we replaced Gaussian noise with several other types of noise (shown in \Cref{table ex1 other noise}) and then tested our filter (which applies the reduced model and updates weights by the Gaussian likelihood) in these settings.
In these numerical experiments, we restricted ourselves to the discrete-time observation case and took the particle filter applying the full model and updating weights by the exact distribution of observation noise as a benchmark. 
Numerical results are presented in \Cref{table ex1 other noise}, which show the relative $L_1$ distance between these two filters.}
\begin{table}[h]
	\centering
	\begin{tabular}{|l|cc|cc|cc|}
		\hline
		\multirow{2}{*}{}            & \multicolumn{2}{c|}{Activated DNA}                                                                 & \multicolumn{2}{c|}{mRNA}                                                                & \multicolumn{2}{c|}{Protein}                                                             \\ \cline{2-7} 
		& \multicolumn{1}{c|}{Mean} & \begin{tabular}[c]{@{}c@{}}Standard\\ deviation\end{tabular} & \multicolumn{1}{c|}{Mean} & \begin{tabular}[c]{@{}c@{}}Standard\\ deviation\end{tabular} & \multicolumn{1}{c|}{Mean} & \begin{tabular}[c]{@{}c@{}}Standard\\ deviation\end{tabular} \\ \hline
		T distribution (with 4 degrees of freedom) & \multicolumn{1}{c|}{1.48\%}     &             8.07\%                                               & \multicolumn{1}{c|}{4.64\%}     &           10.06\%                                                   & \multicolumn{1}{c|}{0.33\%}     &                   24.11\%                                           \\ \hline
		Laplace distribution ($\text{Laplace}(0,1)$) & \multicolumn{1}{c|}{1.33\%}     &               7.90\%                                               & \multicolumn{1}{c|}{4.38\%}     &          8.41\%                                                   & \multicolumn{1}{c|}{0.33\%}     &                   23.88\%                                           \\ \hline
		Log-normal distribution ($\text{Lognormal}(0,0.25)$) & \multicolumn{1}{c|}{1.32\%}     &               5.75\%                                               & \multicolumn{1}{c|}{1.90\%}     &           3.76\%                                                   & \multicolumn{1}{c|}{3.13\%}     &                   42.21\%                                           \\ \hline
	\end{tabular}
	\caption{
		{ Performance of our filter when Gaussian-noise assumption is violated (for the simple gene expression model).			
		This table shows the relative $L_1$ distance between our filter and the benchmark, when the observation noise has non-Gaussian distributions shown in the first column. 
		In these experiments, our filter utilizes the reduced model and updates weights as if the noise is Gaussian; in contrast, the benchmark filter utilizes the original model and updates weights according to the exact distribution of the observation noise.}}
	\label{table ex1 other noise}
\end{table}
{Note that the observation noise directly affects our belief about on the fluorescent proteins --- the stronger the noise intensity is, the larger conditional variance will be.
The shape of the observation noise also matters. 
Therefore, it is not surprising to see that the non-Gaussian noise mainly affects the performance of our filter in estimating the conditional standard deviation of fluorescent proteins (see \Cref{table ex1 other noise}).
However, our filter still provides accurate estimate for the mean dynamics of proteins with relative errors no greater than 4\%. 	
Moreover, our filter is still consistent with the benchmark in estimating the DNA and mRNAs with relative errors no greater than 11\%.
Notably, this consistency between our filter and the benchmark holds for both white noise (the t-distribution and Laplace distribution) and colored noise (the log-normal distribution). 
To conclude, our method is still accurate for this gene expression model when the Gaussian-noise assumption is violated. 
}

{In summary, our particle filter applying the reduced model is both accurate and computationally efficient in solving filtering problems for this multiscale gene expression model.
Moreover, when the Gaussian-noise assumption is violated, our filter can still provide reliable estimates to dynamical states regardless of whether the observation noise is white or colored.
}






\subsection{A transcription regulation network}\label{section transcription regulation network}

\begin{figure}[h!]
	\centering
	\includegraphics[width= 0.8 \textwidth]{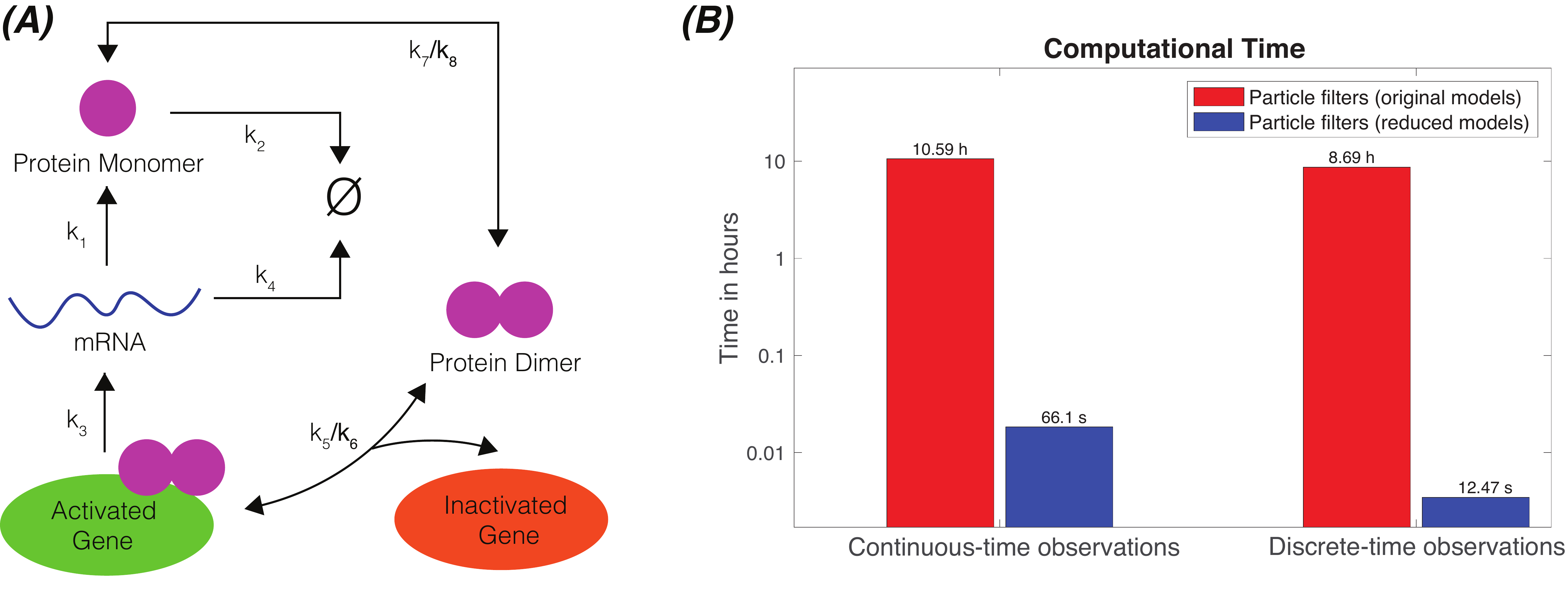}
	\caption{
		(A) shows the biological circuit of the considered network. (B) shows the computational time consumed by different particle filters. Notably, y-axis tick values grow exponentially rather than linearly.
		{The filters were run on a 2.25GHz AMD processors with 15 cores.}
	}
	\label{fig. ex2_cpu_time}
\end{figure}

We then consider a transcription regulation network (called Goutsias' model \cite{goutsias2005quasiequilibrium}) where the gene product affects its own gene activation process (see {\Cref{fig. ex2_cpu_time}} (A)).
{Such an autoregulation mechanism (where the gene product auto-regulates its own gene expression) is very common in living cells, which occurs over 40\% of known \textit{E. coli}'s transcription factors \cite{shen2002network}. Moreover, it has various biological functions, including speeding up the response time of gene expression \cite{rosenfeld2002negative}, inducing genetic oscillators \cite{kruse2005oscillations}, and achieving adaptation to periodic external stimuli \cite{fang2019adaptation}.}
The model considered in this subsection involves five species: the protein monomer (denoted by $S_1$), the dimer (also the transcription factor, denoted by $S_2$), the mRNA ($S_3$), the unbound DNA ($S_4$), and the bound DNA (also the activated state, denoted by $S_5$); moreover, it has eight reactions:
\begin{align*}
S_3 &\ce{->[$k_1$]} S_1+S_3 && \text{translation},
&S_1 &\ce{->[$k_2$]} \emptyset && \text{protein degradation}, \\
S_5 &\ce{->[$k_3$]} S_5 + S_3 && \text{gene transcription},
&S_3	&\ce{ ->[$k_{4}$]}  \emptyset && \text{mRNA degradation},\\
S_2+ S_4 &\ce{->[{$k_5$}]} S_5 && \text{DNA activation},
& S_5 &\ce{ ->[$k_{6}$]} S_2+ S_4 && \text{DNA inactivation},\\
2S_1 &\ce{->[{$k_7$}]} S_2 && \text{dimerization},
& S_2 &\ce{ ->[$k_{8}$]} 2S_1 && \text{dedimerization}.
\end{align*}

\begin{table}[h]
	\centering
	\begin{tabular}{|cr|cr|cr|lr|}
		\hline
		\multicolumn{2}{c}{Reaction constants {(minute$^{-1}$)}} & \multicolumn{2}{c}{Exponents} &  \multicolumn{2}{c}{Scaled rates {(minute$^{-1}$)}} &  \multicolumn{2}{c}{Reaction scale}  \\
		\hline
		$k_1$ & $\mathcal U(4 \times 10^{-1}, 5 \times 10^{-1})$  & $\beta_1$  & $ 0$&   $k'_1$ & $\mathcal U (0.4, 0.5)$ & $\beta_1 + \alpha_3$ & $0$\\
		$k_2$ &  $\mathcal U(6\times 10^{-3}, 9 \times 10^{-3} )$ & $\beta_2$ &  $-1$&  $k'_2$& $\mathcal U (0.6,0.9)$ & $\beta_2 + \alpha_1$ & $0$ \\
		$k_3$ &  $\mathcal U (7\times 10^{-3}, 9\times 10^{-3})$&  $\beta_3$&  $-1$&  $k'_3$& $\mathcal U (0.7, 0.9)$ & $\beta_3+\alpha_5$ & $-1$ \\
		$k_4$ & $\mathcal U (3\times 10^{-3}, 5\times 10^{-3})$ & $\beta_4$ &  $-1$ &  $k'_4$& $\mathcal U (0.3, 0.5)$& $\beta_4+\alpha_3$ & $-1$ \\
		$k_5$ &  $\mathcal U (1\times 10^{-2},3\times 10^{-2})$ &  $\beta_5$&  $-1$&  $k'_5$& $\mathcal U (1,3)$ & $\beta_5+\alpha_2+\alpha_4$& 0\\
		$k_6$ &  $\mathcal U ( 3\times 10^{-1}, 5 \times 10^{-1})$&  $\beta_6$&  $0$&  $k'_6$& $\mathcal U (0.3, 0.5)$& $\beta_6+\alpha_5$ & 0 \\
		$k_7$ &  $\mathcal U ( 6\times 10^{-2}, 9 \times 10^{-2})$&  $\beta_7$&  $-1$&  $k'_7$& $\mathcal U (6, 9)$& $\beta_7+2\alpha_1$ & 1 \\		
		$k_{8}$ &  $\mathcal U ( 4\times 10^{-1}, 6 \times 10^{-1})$&  $\beta_{8}$&  $0$&  $k'_{8}$& $\mathcal U (0.4, 0.6)$& $\beta_8+\alpha_2$ & 1 \\	
		\hline
	\end{tabular}
	\caption{Scaling exponents for reaction rates of the gene expression model: {$\mathcal U$ is the notation for the uniform distribution}}
	\label{table reaction constants of transcription regulation network}
\end{table}

We take $N=100$, $\alpha_1=\alpha_2=1$ and $\alpha_i=0$ for $3\leq i \leq 5$;
in other words, the cellular system consists of hundreds of protein molecules but very few copies of other molecules. 
The values of reaction constants and their scaling exponents are shown in {\Cref{table reaction constants of transcription regulation network}}.
With this setting, we can easily check that $\gamma_1=0$, $\gamma_2=1$, $\mathbb L_1=\text{span}\{e_1,e_2,e_4,e_5\}$,
and $\mathbb L_2=\text{span}\{e_1+2 e_2, e_3, e_4+e_5\}$, where $e_i$ is the unit vector with the $i$-th entry being one and the rest being zero.
For initial conditions, we assume $X^{N}_1(0)$ to have a Poisson distribution with mean 1, $X^{N}_2(0)$ to have a Poisson distribution with mean 10, $X^{N}_3(0)$ to have a Poisson distribution with mean 2, $X^N_{4}(0)$ to have a binary distribution with mean $1/10$, and $X^{N}_5(0)$ to satisfy $X^{N}_4(0)=1- X^{N}_1(0)$.
Also, we assume that all reaction constants and initial conditions except $X_5(0)$ are independent of each other.

In this case, we are interested in system dynamics at the second time scale.
By \eqref{eq. scaling stochastic dynamics}, the full dynamics of the transcription regulation network at the second time scale can be written as 
\begin{align*}
		X^{N,\gamma_2}_1 (t) =&  X^{N,\gamma_2}_1 (0) +  \frac{1}{N} R_1 \left(  k'_1 N  \int_0^t X^{N,\gamma_2}_3 (s) \dd s \right) - \frac{1}{N} R_2 \left(  k'_2 N  \int_0^t X^{N,\gamma_2}_1 (s) \dd s \right)
		- \frac{1}{N} R_5\left( k'_5 N \int_0^t X^{N,\gamma_2}_2 (s) X^{N,\gamma_2}_4(s) \dd s\right) \\
		&+ \frac{1}{N} R_6\left( k'_6 N \int_0^t X^{N,\gamma_2}_5 (s) \dd s\right) 
		 -  \frac{2}{N} R_7 \left(  k'_7 N^2 \int_0^t X^{N,\gamma_2}_1 (s) \left(X^{N,\gamma_2}_1 (s)-\frac{1}{N}\right)\dd s \right) 
		  +  \frac{2}{N} R_8 \left(  k'_8 N^2  \int_0^t X^{N,\gamma_2}_2 (s) \dd s \right)  \\
		X^{N,\gamma_2}_2 (t) =& X^{N,\gamma_2}_2 (0) + \frac{1}{N} R_7 \left(  k'_7 N^2  \int_0^t X^{N,\gamma_2}_1 (s) \left(X^{N,\gamma_2}_1 (s)-\frac{1}{N}\right) \dd s \right) 
		-  \frac{1}{N} R_8 \left(  k'_8 N^2  \int_0^t X^{N,\gamma_2}_2 (s) \dd s \right) \\
		X^{N,\gamma_2}_3 (t) =& X^{N,\gamma_2}_3 (0) + R_3 \left(k'_3 \int_0^t X^{N,\gamma_2}_5 (s) \dd s \right)- R_4 \left(k'_4 \int_0^t X^{N,\gamma_2}_3 (s) \dd s \right)\\
		X^{N,\gamma_2}_4 (t) =& X^{N,\gamma_2}_4 (0) - R_5\left( k'_5 N \int_0^t X^{N,\gamma_2}_2 (s) X^{N,\gamma_2}_4(s) \dd s\right)
		+ R_6\left( k'_6 N \int_0^t X^{N,\gamma_2}_5 (s) \dd s\right)
		\\
		X^{N,\gamma_2}_5 (t) =& X^{N,\gamma_2}_5 (0)+ R_5\left( k'_5 N \int_0^t X^{N,\gamma_2}_2 (s) X^{N,\gamma_2}_4(s) \dd s\right)
		- R_6\left( k'_6 N \int_0^t X^{N,\gamma_2}_5 (s) \dd s\right).
	\end{align*}
Then, we derive the hybrid model at the second time scale.
By \eqref{eq. reduced model at the first time scale}, the reduced model at the first time scale satisfies
	\begin{align*}
		X^{\gamma_1}_1(t) &=  \lim_{N\to\infty }X^{N}_1(0) 
		-  2  k'_7  \int_0^t \left(X^{\gamma_1}_1 (s)\right)^2 \dd s 
		+  2 k'_8  \int_0^t X^{\gamma_1}_2 (s) \dd s  \\
		X^{\gamma_1}_2(t) &=  \lim_{N\to\infty }X^{N}_2(0)  +   k'_7  \int_0^t \left(X^{\gamma_1}_1 (s)\right)^2 \dd s 
		-  k'_8  \int_0^t X^{\gamma_1}_2 (s) \dd s  \\
		X^{\gamma_1}_3(t) &=  \lim_{N\to\infty }X^{N}_3(0) \\
		X^{\gamma_1}_4(t) &=  \lim_{N\to\infty }X^{N}_4(0) - R_5\left( k'_5 \int_0^t X^{\gamma_1}_2 (s) X^{\gamma_1}_4(s) \dd s\right)
		+ R_6\left( k'_6 \int_0^t X^{\gamma_1}_5 (s) \dd s\right) \\
		X^{\gamma_1}_5(t) &=  \lim_{N\to\infty }X^{N}_5(0) + R_5\left( k'_5 \int_0^t X^{\gamma_1}_2 (s) X^{\gamma_1}_4(s) \dd s\right)
		- R_6\left( k'_6 \int_0^t X^{\gamma_1}_5 (s) \dd s\right).
	\end{align*}
Also, it is easy to check that the operator $\mathcal L^{N,\gamma_1}_{K, y}$ for $y\in \mathbb L_2$ admits a unique stationary distribution 
\begin{align} \label{eq. quasi-stationary distrbution}
	\bar V_1^{\mathcal K, y}(x_1)=  
	\left\{
	\begin{array}{cc}
		C^{q}_{2y_4} \left(\frac{k'_6}{k'_5\psi(y)+k'_6}\right)^{q+y_4} \left(\frac{k'_5\psi(y)}{k'_5\psi(y)+k'_6}\right)^{y_4-q}  & x_1=( 9y_1-2\psi(y), \psi(y)-y_2 ,0, q, - q) \\
		0 & \text{otherwise}
	\end{array}
	\right.
\end{align}
where $x_1\in \{x| x=(I-\Pi_2) \tilde x, ~ \Pi_2 \tilde x=y,~\tilde x\in\mathbb R  \}$, $-y_4 \leq q \leq y_4$, and 
\begin{equation*}
	\psi(y)=\frac{y_1+2y_2}{2}+\frac{k'_8-\sqrt{(k'_8)^2+8k'_7k'_8(y_1+2y_2)}}{8k'_7}.
\end{equation*}
Therefore, by \eqref{eq. reduced models at the second time scale}, the reduced model at the second time scale satisfies
	\begin{align*}
		X^{\gamma_2}_1(t) =& \lim_{N\to\infty }\frac{X^{N}_1(0) + 2X^{N}_2(0) }{5}+ \frac{k'_1}{5}  \int_0^t X^{\gamma_2}_3 (s) \dd s  - \frac{k'_2}{5}  \int_0^t 5X^{\gamma_2}_1 (s) - 2\psi\left(X^{\gamma_2}(s)\right)  \dd s  \\
		X^{\gamma_2}_2(t) = &\lim_{N\to\infty }\frac{2X^{N}_1(0) + 4X^{N}_2(0) }{5}+ \frac{2k'_1}{5}  \int_0^t X^{\gamma_2}_3 (s) \dd s  - \frac{2k'_2}{5}  \int_0^t 5X^{\gamma_2}_1 (s) - 2\psi\left(X^{\gamma_2}(s)\right)  \dd s  \\
		X^{\gamma_2}_3(t) =& \lim_{N\to\infty }X^{N}_3(0) + R_3 \left(k'_3 \int_0^t \frac{2k'_5\psi(X^{\gamma_2}(s)) X^{\gamma_2}_4(s) }{k'_5\psi(X^{\gamma_2}(s))+k'_6} \dd s \right)- R_4 \left(k'_4 \int_0^t X^{\gamma_2}_3 (s) \dd s \right) \\
		X^{\gamma_2}_4(t) =&\lim_{N\to\infty }\frac{X^{N}_4(0) + X^{N}_5(0) }{2} \\
		X^{\gamma_2}_5(t) =& \lim_{N\to\infty }\frac{X^{N}_4(0) + X^{N}_5(0) }{2}. 
	\end{align*}

Let $S_1$ and $S_2$ be both fluorescent reports, and light intensity signals satisfy  \eqref{eq. continuous time observation} or \eqref{eq. discrete time observation}
where $h(x)=  \left(x_1+2x_2\right) \wedge 10^{3}$ with $10^3$ being the measurement range, $B(t)$ is a Brownian motion, $t_i=2i$ (i.e., discrete-time observations come every 200 seconds or 3.33 minutes), and $\{W(t_i)\}_{i\in\mathbb N_{>0}}$ is a sequence of mutually independent standard Gaussian random variables.

In this example, we also randomly chose a set of system parameters, simulate the system for 2.5 hours (9000 seconds), and generate observations for both continuous-time and discrete-time scenarios.
Then, we use SIR particle filters that apply the reduced model at the second time scale to infer non-fast fluctuating dynamic states and reaction constants $\mathcal K$.
Throughout this numerical example, we set the population of particles to be {10,000}.
Finally, we take the particle filter that applies the full dynamic model as a benchmark. 

Numerical simulation results are presented in {\Cref{fig. ex2_cpu_time}, \Cref{fig ex2_continuous}, and \Cref{fig ex2_discrete}}.
From {\Cref{fig. ex2_cpu_time}}, we can observe that particle filters applying the reduced dynamics consumes far less computational time than particle filters applying the full dynamic model.
Meanwhile, we can learn from {\Cref{fig ex2_continuous} and \Cref{fig ex2_discrete}} that both particle filters can follow the trend of true dynamic states, and their trajectories almost merge together {(whose relative $L_1$ distances are presented in \Cref{table ex2 L2 distance})}.
{These observations imply that these filters perform quite the same in estimating hidden dynamic states.} 
\begin{table}[h]
	\centering
	\begin{tabular}{|c|cc|cc|cc|}
		\hline
		\multirow{2}{*}{}            & \multicolumn{2}{c|}{Activated DNA}                                                                 & \multicolumn{2}{c|}{mRNA}                                                                & \multicolumn{2}{c|}{Total mass of proteins}                                                             \\ \cline{2-7} 
		& \multicolumn{1}{c|}{Mean} & \begin{tabular}[c]{@{}c@{}}Standard\\ deviation\end{tabular} & \multicolumn{1}{c|}{Mean} & \begin{tabular}[c]{@{}c@{}}Standard\\ deviation\end{tabular} & \multicolumn{1}{c|}{Mean} & \begin{tabular}[c]{@{}c@{}}Standard\\ deviation\end{tabular} \\ \hline
		Continuous-time observations & \multicolumn{1}{c|}{0.20\%}     &               6.14\%                                               & \multicolumn{1}{c|}{1.55\%}     &           1.19\%                                                   & \multicolumn{1}{c|}{0.06\%}     &                    1.42\%                                           \\ \hline
		Discrete-time observations   & \multicolumn{1}{l|}{0.30\%}     &                   8.32\%                                           & \multicolumn{1}{l|}{2.44\%}     &                      2.01\%                                        & \multicolumn{1}{l|}{0.80\%}     &        1.39\%                                                      \\ \hline
	\end{tabular}
	\caption{{Relative $L_1$ distance between both kinds of particle filters in estimating the transcription regulation network. The relative $L_1$ distance between two signals $f_1(t)$ and $f_2(t)$ is defined by $\frac{\|f_1-f_2\|_{L_1}}{\|f_1\|_{L_1}}$.  The results show that these two kinds of particle filters provide very consistent estimate to the dynamical states with relative errors no greater than 9\%.} }
	\label{table ex2 L2 distance}
\end{table}
{Moreover, from \Cref{fig ex2_continuous}(B), and \Cref{fig ex2_discrete}(B), we can observe that both filters also provide consistent and accurate estimates to model parameters.
However, sample degeneracy still occurs in this case, which makes the estimated posterior of model parameters not very consistent (see \Cref{fig ex2_particle_distribution}).  
}

{Finally, we investigate the performance of our filters when the observation noise is non-Gaussian. Similar to the previous case study, we restricted ourselves to the discrete-time observation case and replaced the Gaussian noise with three other types of noise (shown in the first column of  \Cref{table ex2 other noise}).
In these numerical experiments, we compared our filter that uses the reduced model and updates weights by the Gaussian likelihood with a benchmark filter, which uses the full dynamical model and updates weight by the exact distribution of observation noise. 
Their relative $L_1$ distances are presented in \Cref{table ex2 other noise}.
From the result, we observe that in this case, our method still provides a relatively accurate estimate to the DNA state and mRNAs with relative  $L_1$ errors less than  16\%.
The non-Gaussian noise mainly affects the estimation of the conditional standard deviation of  fluorescent proteins; however, our filter still provides accurate estimates to the conditional mean of proteins with relative errors no greater than 4\%. 
To conclude, when the Gaussian noise assumption does not hold,  our filter is still reliable for this transcription regulation network, and this result holds for both white noise (the t-distribution and Laplace distribution) and colored noise (the log-normal distribution).}

In summary, our particle filter applying the reduced model is both accurate and computationally efficient in solving filtering problems for this transcription regulation model.
{Moreover, when the Gaussian-noise assumption does not hold, our filter can still provide reliable estimates to dynamical states no matter whether the observation noise is white or colored.
}

\begin{figure}[h!]
	\centering
	\includegraphics[width= 0.9 \textwidth]{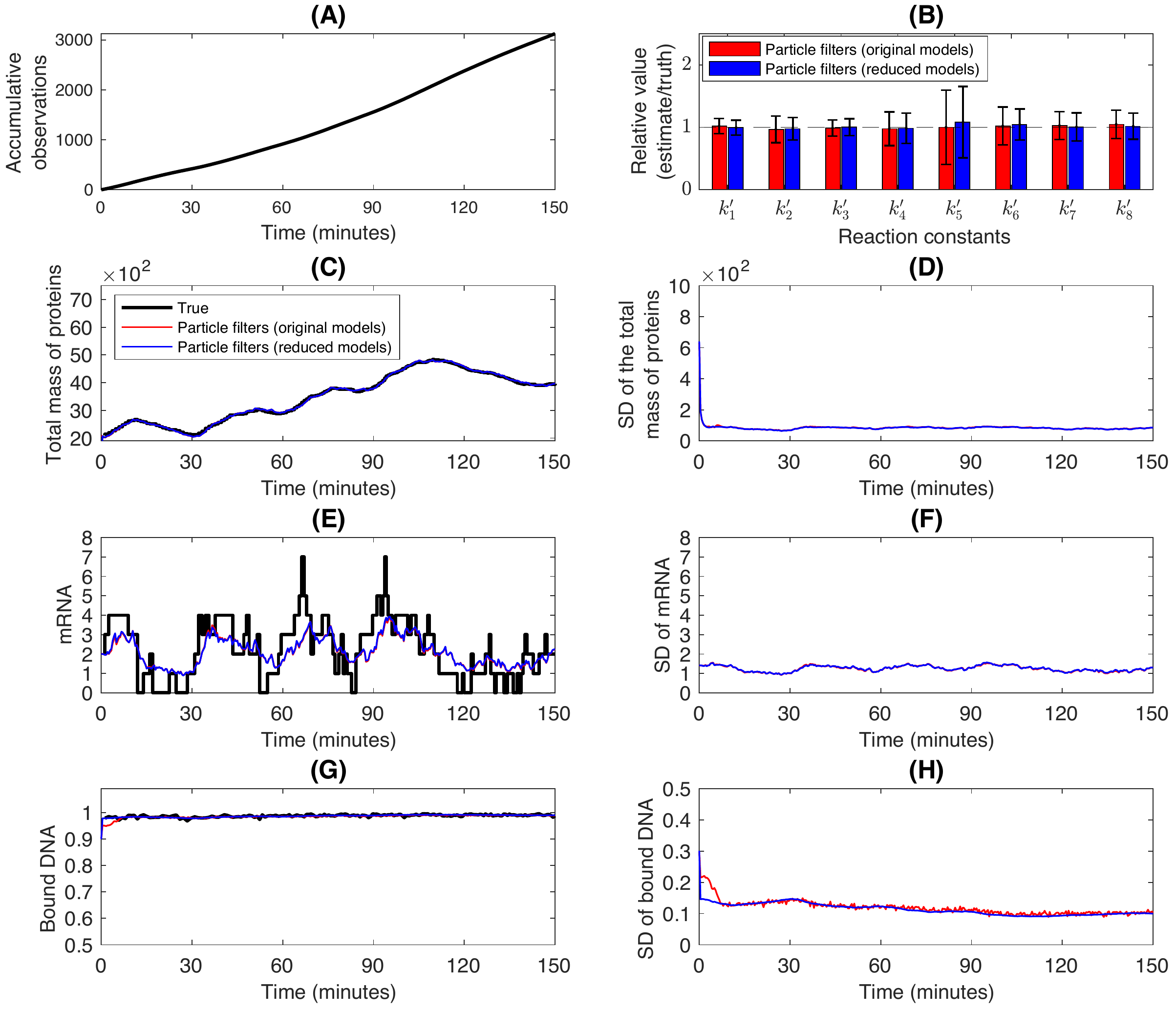}
	\caption{Simulation results for the transcription regulation network with continuous-time observations: {Panel (A) shows the accumulative observation signal $Y^{N,\gamma}_{c}(t)$, which is the integral of the time-course measurement from time 0 to time $t$. Panel (B) shows the performance of both filters in estimating reaction constants, where error bars represent 95\% confidence intervals. Here we draw the relative values of these estimates $\left(\frac{\text{the estimate}}{\text{the true parameter}}\right)$ so that all the true parameters are rescaled to 1.}
	The rest of the panels compare performances of different particle filters in estimating dynamic states, where black lines are the true values of the underlying system (the panel (G) shows 1 minute moving average of the true value), red lines are the estimates by the particle filter using the full dynamic model, and the blue lines are the estimates by the particle filter using the reduced dynamic model.
	Specifically, (C), (E), and (G) show the estimates of the conditional means of the total mass of proteins ($S_1+2S_2$), mRNAs, and the activated Genes, respectively; 
	(D), (F), and (H) present the estimates of the conditional standard deviations of these quantities.
	}
	\label{fig ex2_continuous}
\end{figure}

\begin{figure}[h!]
	\centering
	\includegraphics[width= 0.9  \textwidth]{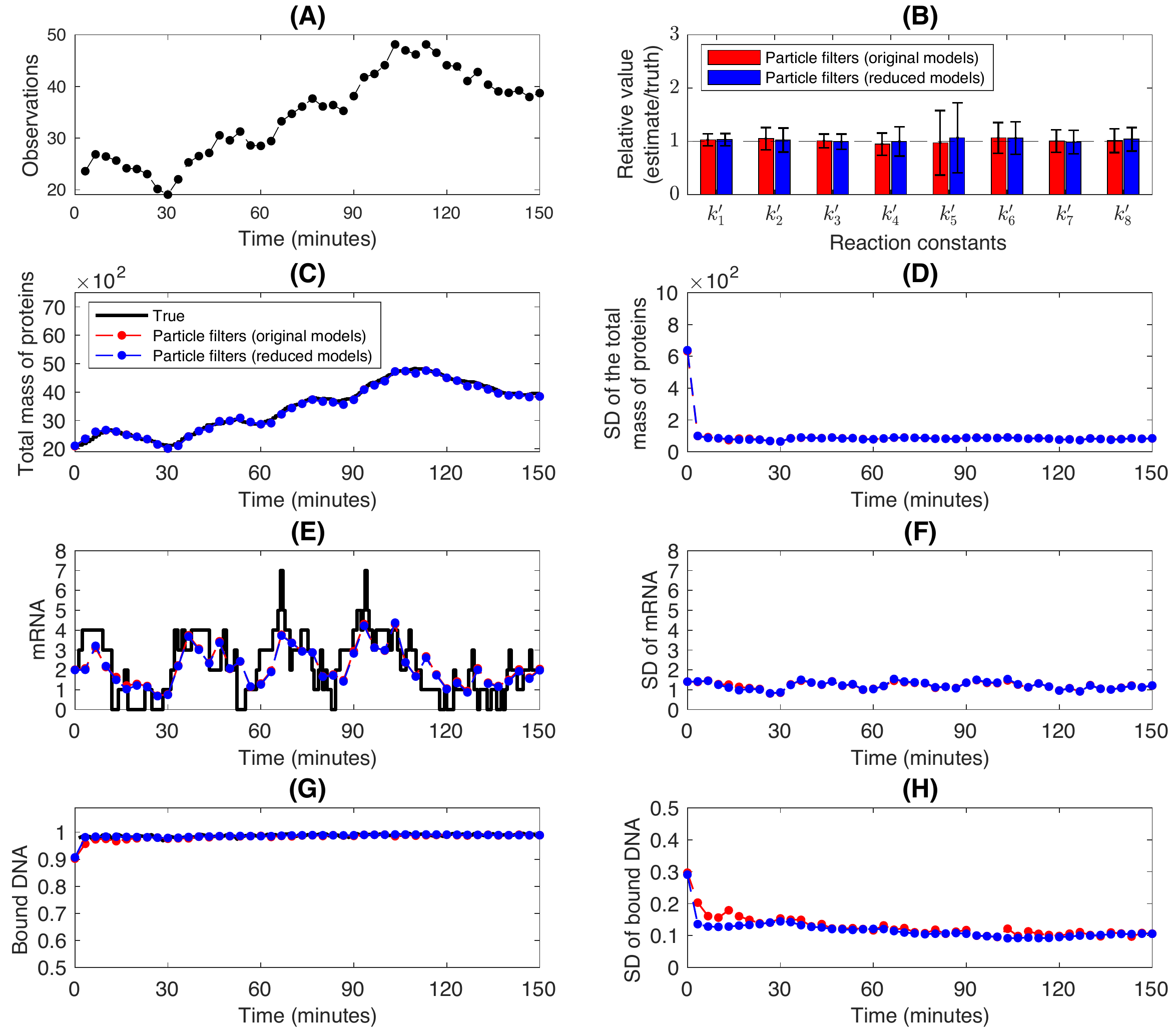}
	\caption{Simulation results for the transcription regulation network with discrete-time observations. The meaning of each panel is the same as the corresponding one in {\Cref{fig ex2_continuous}, except that panel (A) draws the raw data of the observation instead of the accumulative observation signal.}
	}
	\label{fig ex2_discrete}
\end{figure}

\begin{figure}
	\centering
\includegraphics[width= 0.70 \textwidth]{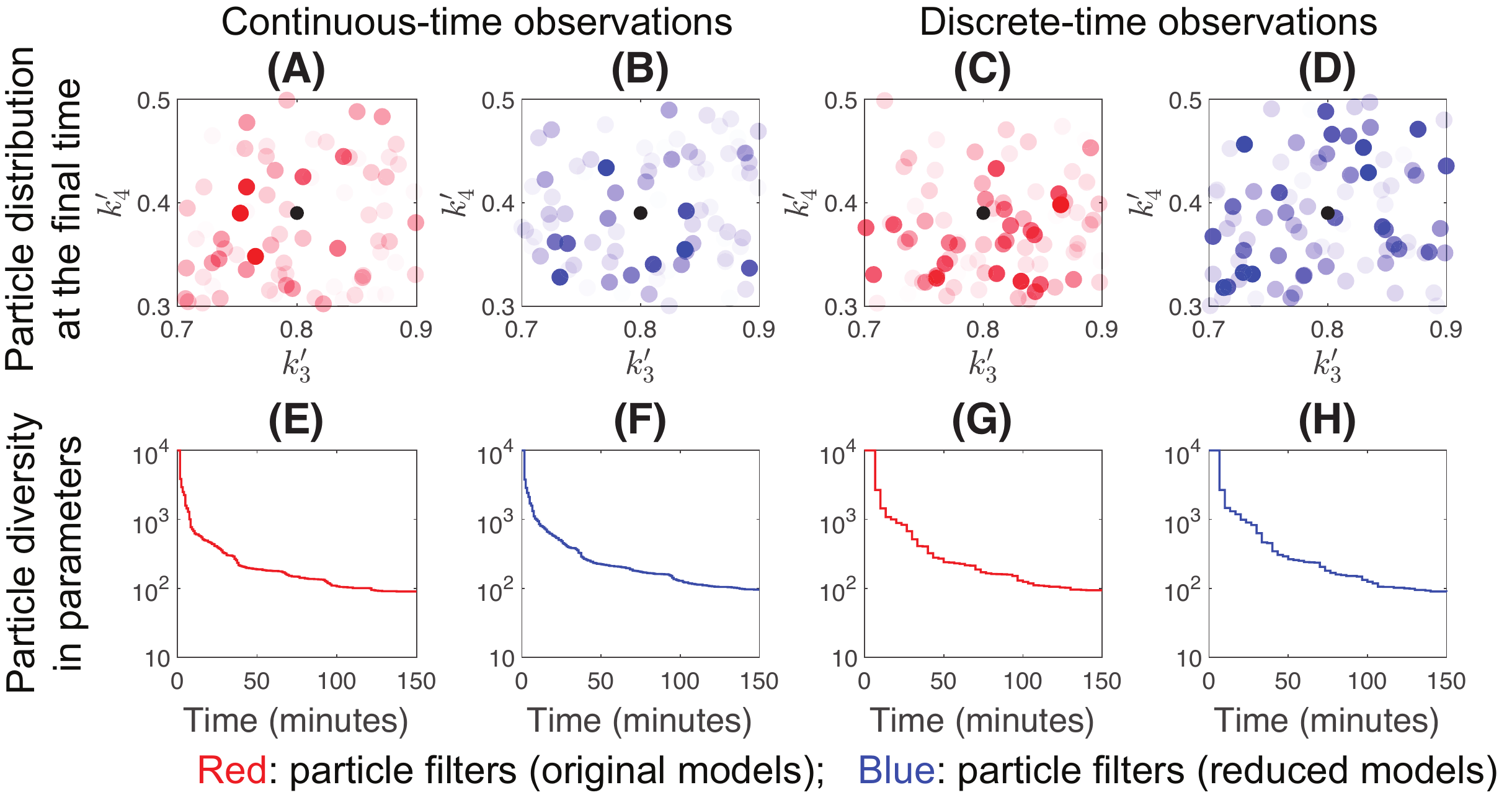}
\caption{
	{Particle distributions in the transcription regulation network example. 
		The first two columns show the results for the continuous-time observation case, and the last two columns show the results for the discrete-time observation case. 
		The first row shows particle distributions in $k'_3$--$k'_4$ plane at the final time, where the black dot shows the true value, and the colored dots show the particles from the filters.
		The transparency of a colored dot indicates the number of particles located in this specific site; the more transparent it is, the fewer particles are located. 
		The second row shows the time evolution of the particle diversity in parameters (i.e., the number of distinguishable particles for estimating model parameters). 
		This figure tells that for both observation types and both particle filters, the particle diversity in parameters decays dramatically over time. Finally, about a hundred distinguishable particles are left, and their distributions are not very consistent. For instance, (C)  has very few particles in the far north, whereas (D) has many particles in that region.}
}
\label{fig ex2_particle_distribution}
\end{figure}

\begin{table}[h]
	\centering
	\begin{tabular}{|l|cc|cc|cc|}
		\hline
		\multirow{2}{*}{}            & \multicolumn{2}{c|}{Bound DNA}                                                                 & \multicolumn{2}{c|}{mRNA}                                                                & \multicolumn{2}{c|}{Total mass of proteins}                                                             \\ \cline{2-7} 
		& \multicolumn{1}{c|}{Mean} & \begin{tabular}[c]{@{}c@{}}Standard\\ deviation\end{tabular} & \multicolumn{1}{c|}{Mean} & \begin{tabular}[c]{@{}c@{}}Standard\\ deviation\end{tabular} & \multicolumn{1}{c|}{Mean} & \begin{tabular}[c]{@{}c@{}}Standard\\ deviation\end{tabular} \\ \hline
		T distribution (with 6 degrees of freedom) & \multicolumn{1}{c|}{0.56\%}     &             16.74\%                                               & \multicolumn{1}{c|}{15.01\%}     &           13.02\%                                                   & \multicolumn{1}{c|}{0.69\%}     &                   20.18\%                                           \\ \hline
		Laplace distribution ($\text{Laplace}(0,1)$) & \multicolumn{1}{c|}{0.41\%}     &               10.93\%                                               & \multicolumn{1}{c|}{10.17\%}     &          5.52\%                                                   & \multicolumn{1}{c|}{0.70\%}     &                   21.17\%                                           \\ \hline
		Log-normal distribution ($\text{Lognormal}(0,0.25)$) & \multicolumn{1}{c|}{0.27\%}     &               9.01\%                                               & \multicolumn{1}{c|}{13.05\%}     &           10.57\%                                                   & \multicolumn{1}{c|}{3.44\%}     &                   47.93\%                                           \\ \hline
	\end{tabular}
	\caption{
		{Performance of our filter when Gaussian-noise assumption is violated (for the transcription regulation model).			
			This table shows the relative $L_1$ distance between our filter and the benchmark, when the observation noise has non-Gaussian distributions shown in the first column. 
			In these experiments, our filter utilizes the reduced model and updates weights as if the noise is Gaussian; in contrast, the benchmark filter utilizes the original model and updates weights according to the exact distribution of the observation noise.}}
	\label{table ex2 other noise}
\end{table}










\section{Conclusion}\label{Sec conclusion}
In this paper, we established efficient particle filters to solve filtering problems for multiscale stochastic reaction network systems by using the time-scale separation technique.
We first showed that the solution of the filtering problem for the original system can be accurately approximated by the solution of the filtering problem for a reduced model that represents the dynamics by a hybrid approximation (see \cref{thm convergence of artificial filters}). 
The reduced model is based on exploiting the time-scale separation in the original network and can greatly reduce the required computational effort to simulate the dynamics.
Consequently, these results enabled us to develop efficient particle filters to solve the original filtering problem by applying the particle filter to the reduced model (see \cref{thm main results}).
Finally, we used two numerical examples to illustrate our approach.
Both examples show that the constructed filter can accurately and computationally efficiently solve the corresponding filtering problems and, therefore, imply that our approach can improve scientists' ability to extract dynamic information about intracellular systems from time-course experiments.

There are a few topics deserving further investigation in future work. 
First, it is worthwhile to further improve the performance of particle filters in estimating system parameters.
{From numerical experiments, we observed that the SIR particle filter is less accurate in estimating the posterior of system parameters due to sample degeneracy.}
A standard method to mitigate this problem is to use regularized particle filters instead of SIR ones \cite{liu2001combined}.
Usually, to guarantee the performance of a regularized particle filter, one needs to show some regularity conditions of the transition kernel of the dynamic model (see \cite[Lemma 2.38]{del2000branching}), which requires researchers to do more elaborate analyses of the underlying system.
{As our first attempt to solve this problem, we provided in a follow-up paper \cite{fang2021convergence} several regularity conditions for SRNs using parameter sensitivity analysis and proved the convergence of regularized particle filters under these conditions.}
Another direction for future research is to utilize adaptive model reduction algorithms (e.g., \cite{hepp2015adaptive}) to speed up the particle filters further. 
It will be very helpful to systems where magnitudes of concentration levels vary over time, such as toggle switches \cite{gardner2000construction} and repressilator \cite{elowitz2000synthetic}. 
Finally, one can also implement our algorithms on a Cyberloop platform \cite{rullan2018optogenetic} and apply them to biological studies, e.g., identifying intracellular reaction mechanisms and guiding cell differentiation.
In this scenario, one need to deal with hundreds of cellular systems under the microscope and, therefore, solve hundreds of filtering problems  simultaneously, which requires the algorithm to be parallelized efficiently.

\appendix

\section{Change of measure methods for reduced models}{\label{Sec change of measure methods}}

In this section, we introduce reference probabilities and Kallianpur-Striebel formulas for the reduced models, which will be helpful in later analysis.

\subsection{Continuous-time observations}
We first consider the scenario where observations are continuous with respect to time. 
For reduced models \eqref{eq. reduced model at the first time scale} and \eqref{eq. reduced models at the second time scale}, we define Girsanov's random variables, $Z^{\gamma_\ell}_{c}(t)$ ($\ell=1,2$), like the following
\begin{align*}
	&Z^{\gamma_\ell}_c(t) \triangleq \exp\left(
	  \int_0^t h^{\top}(X^{\gamma_\ell}(s)) \dd Y^{\gamma_\ell}_c(s) - \frac{1}{2} \int_0^t \| h(X^{\gamma_\ell}(s))  \|^2 \dd s
	\right) 
\end{align*}
whose reciprocals are martingales with respect to $\mathcal F_t$ under $\mathbb P$ \cite[Lemma 3.9]{bain2008fundamentals}.
Then, reference probabilities for reduced models \eqref{eq. reduced model at the first time scale} and \eqref{eq. reduced models at the second time scale} are respectively defined by
\begin{equation*}
	\left.\frac{\dd \mathbb P^{\gamma_1}_c}{\dd \mathbb P}\right|_{\mathcal F_t}
	\triangleq \left(Z^{\gamma_1}_c (t)\right)^{-1}
	\qquad \text{and} \qquad
	\left.\frac{\dd \mathbb P^{\gamma_2}_c}{\dd \mathbb P}\right|_{\mathcal F_t}
	\triangleq \left(Z^{\gamma_2}_c (t)\right)^{-1}.
\end{equation*}
Under the reference probability $\mathcal P^{\gamma_\ell}_c$ ($\ell=1,2$), the observation $Y^{\gamma_\ell}_c(\cdot)$ is independent of the underlying system $(\mathcal K, X^{\gamma_\ell})$ and becomes an m-vector of independent standard Brownian motions \cite[Proposition 3.13]{bain2008fundamentals}.
For these reduced models, Kallianpur-Striebel formulas are expressed as \cite[Proposition 3.16]{bain2008fundamentals} (also see \cite{kallianpur1968estimation})
\begin{align}
	\pi^{\gamma_\ell}_{c,t}= 
	{\mathbb E_{\mathbb P^{\gamma_\ell}_c} \left[ Z^{\gamma_\ell}_c(t) \phi (\mathcal K, X^{\gamma_\ell}(t)) \left| \mathcal Y^{\gamma_\ell}_{c,t} \right. \right]}
	~\Bigg/~
	{\mathbb E_{\mathbb P^{\gamma_\ell}_c} \left[ Z^{\gamma_\ell}_c(t)  \left| \mathcal Y^{\gamma_\ell}_{c,t} \right. \right]} 
	&&
	\mathbb P\text{-a.s} ~\text{and}~\mathbb P^{\gamma_1}_c\text{-a.s},
	\label{eq. Kallianpur-stribel formula for continuous time observations and reduced models}
\end{align}
for $\ell=1,2$ and any bounded measurable function $\phi$.

Note that the terms on the right hand sides of \eqref{eq. Kallianpur-Striebel formula for the full model continuous observation} and \eqref{eq. Kallianpur-stribel formula for continuous time observations and reduced models} are respectively $\mathcal Y^{N,\gamma}_{c,t}$, and $\mathcal Y^{\gamma_\ell}_{c,t}$ measurable.
Therefore, for any bounded measurable function $\phi$, there exist measurable functions $\hat g^{N,\gamma}_{\phi, c,t}(\cdot)$, $\hat g^{\gamma_\ell}_{\phi, c,t}(\cdot)$ satisfying
\begin{align}
    \hat g^{N,\gamma}_{\phi, c,t}(Y^{N,\gamma}_{c,0:t}) &=
    \mathbb E_{\mathbb P^{N,\gamma}_{c}} \left[ Z^{N,\gamma}_{c}(t) \phi\left(\mathcal K, X^{N,\gamma}(t)\right) \left| \mathcal Y^{N,\gamma}_{c,t} \right.\right], 
    \label{eq. definition of g full model continuous time } 	\quad \forall \gamma\in\mathbb R
    \\
    \hat g^{\gamma_\ell}_{\phi, c,t}(Y^{\gamma_\ell}_{c,0:t})
    &=
	\mathbb E_{\mathbb P^{\gamma_\ell}_c} \left[ Z^{\gamma_\ell}_c(t) \phi (\mathcal K, X^{\gamma_\ell}(t)) \left| \mathcal Y^{\gamma_\ell}_{c,t} \right. \right], 	
	\quad \ell=1,2,
	\label{eq. definition of g reduced mode continuous time}
\end{align}
where $Y^{N,\gamma}_{c,0:t}$ is the trajectory of $Y^{N,\gamma}_{c}(\cdot)$ from time 0 to $t$, and  $Y^{\gamma_\ell}_{c,0:t}$ ($\ell=1,2$) is the trajectory of $Y^{\gamma_\ell}_{c}(\cdot)$ from time 0 to $t$.
Moreover, by \eqref{eq. Kallianpur-Striebel formula for the full model continuous observation} and \eqref{eq. Kallianpur-stribel formula for continuous time observations and reduced models}, we can conclude that 
for all bounded measurable function $\phi$, there hold
\begin{align}
	\hat f^{N,\gamma}_{\phi, c,t}(Y^{N,\gamma}_{c,0:t}) 
	&=
	\frac{\hat g^{N,\gamma}_{\phi, c,t}(Y^{N,\gamma}_{c,0:t}) }{\hat g^{N,\gamma}_{1, c,t}(Y^{N,\gamma}_{c,0:t})}, 
	\qquad\mathbb P\text{-a.s},~ \mathbb P^{N,\gamma}_c\text{-a.s}
	\qquad \forall \gamma\in\mathbb R,
	\label{eq. f function Kallianpur-Striebel formula full model continuous time} \\
	\hat f^{\gamma_\ell}_{\phi, c,t}(Y^{\gamma_\ell}_{c,0:t})
	&=
	\frac{\hat g^{\gamma_\ell}_{\phi, c,t}(Y^{\gamma_\ell}_{c,0:t})}{\hat g^{\gamma_\ell}_{1, c,t}(Y^{\gamma_\ell}_{c,0:t})},
	\qquad\mathbb P\text{-a.s},~ \mathbb P^{\gamma_{\ell}}\text{-a.s}
	\qquad ~~ \forall \ell\in\{1,2\}.
	\label{eq. f function Kallianpur-Striebel formula reduced model continuous time}
\end{align}

\subsection{Discrete-time observations}
Then, we consider the case where the observations are discrete with respect to time.
For reduced models \eqref{eq. reduced model at the first time scale} and \eqref{eq. reduced models at the second time scale}, we define random variables
\begin{equation*}
	Z^{\gamma_1}_d(t_i)= \prod_{j=1}^{i} H\left(X^{\gamma_1}(t_j),Y^{\gamma_1}(t_j)\right) \quad \text{~and~} \quad Z^{\gamma_2}_d(t_i)= \prod_{j=1}^{i} H\left(X^{\gamma_2}(t_j),Y^{\gamma_2}(t_j)\right) 
\end{equation*}
where $H(x,y)= \exp\left(h^{\top}(x) y-\frac{1}{2}\|h(x)\|^2\right)$, and whose reciprocals are martingales under $\mathbb P$.\footnote{The martingale property can be easily checked by looking at their characteristic functions.}
Hence, the reference probabilities for reduced models \eqref{eq. reduced model at the first time scale} and \eqref{eq. reduced models at the second time scale} can be respectively defined by
\begin{equation*}
	\left.\frac{\dd \mathbb P^{\gamma_1}_d}{\dd \mathbb P}\right|_{\mathcal F_{t_i}}
	\triangleq \left(Z^{\gamma_1}_d (t_i)\right)^{-1}
	\qquad \text{and} \qquad
	\left.\frac{\dd \mathbb P^{\gamma_2}_d}{\dd \mathbb P}\right|_{\mathcal F_{t_i}}
	\triangleq \left(Z^{\gamma_2}_d (t_i)\right)^{-1}.
\end{equation*}
Under the reference probability $\mathbb P^{\gamma_\ell}_d$ ($\ell=1,2$), the observations $Y^{\gamma_\ell}_d(t_i)$ are independent of the underlying system $(\mathcal K, X^{\gamma_\ell}(\cdot))$ and are mutually independent m-variate Gaussian random variables whose coefficient matrices are the identity matrix.\footnote{It can be easily checked by looking at the joint characteristic function of the observation and the underlying system.}
Then, for these reduced models, Kallianpur-Striebel formulas are expressed as 
\begin{align}
	\pi^{\gamma_\ell}_{d,t_i}(\phi) &=
	{{\mathbb E}_{\mathbb P^{\gamma_\ell}_d} \left[\left. Z^{\gamma_\ell}_d(t_i) \phi\left(\mathcal K, X^{\gamma_\ell}(t_i)\right) \right|  \mathcal Y^{\gamma_\ell}_{d,t_i}  \right]}
	~\Bigg/~
	{ {\mathbb E} _{\mathbb P^{\gamma_\ell}_d} \left[ \left. Z^{\gamma_\ell}_d(t_i) \right|  \mathcal Y^{\gamma_\ell}_{d,t_i}  \right]} \label{eq. Kallianpur-Striebel formula discrete time observations and reduced model}
	&&
	\mathbb P\text{-a.s.}~\text{and}~\mathbb P^{\gamma_\ell}_d\text{-a.s.} 
\end{align}
for $\ell=1,2$ and any bounded measurable function $\phi$.

Note that terms on the right hand sides of \eqref{eq. Kallianpur-Striebel formula for the full model discrete observation} and \eqref{eq. Kallianpur-Striebel formula discrete time observations and reduced model} are respectively $\mathcal Y^{N, \gamma}_{d,t_i}$ and $\mathcal Y^{\gamma_\ell}_{d,t_i} $ measurable. Therefore, for any bounded measurable function $\phi$, there exist measurable functions $\hat g^{N,\gamma}_{\phi, d,t_i}(\cdot)$ and $\hat g^{\gamma_\ell}_{\phi, d,t_i}(\cdot)$ satisfying
\begin{align*}
\hat g^{N,\gamma}_{\phi, d,t_i}(Y^{N,\gamma}_{d,1:i}) &=
\mathbb E_{\mathbb P^{N,\gamma}_{d}} \left[ Z^{N,\gamma}_{d}(t_i) \phi\left(\mathcal K, X^{N,\gamma}(t_i)\right) \left| \mathcal Y^{N,\gamma}_{d,t_i} \right.\right], 
&& \forall\gamma\in\mathbb R,
\\
\hat g^{\gamma_\ell}_{\phi, d,t_i}(Y^{\gamma_\ell}_{d,1:i})
&=
\mathbb E_{\mathbb P^{\gamma_\ell}_d} \left[ Z^{\gamma_\ell}_d(t_i) \phi (\mathcal K, X^{\gamma_\ell}(t_i)) \left| \mathcal Y^{\gamma_\ell}_{d,t_i} \right. \right],  
&& \ell=1,2,
\end{align*}
where $Y^{N,\gamma}_{d,1:i}=\left(Y^{N,\gamma}_{d}(t_1),\dots, Y^{N,\gamma}_d(t_i)\right)$  and  $Y^{\gamma_\ell}_{d,1:i}=\left( Y^{\gamma_\ell}_d (t_1),\dots, Y^{\gamma_\ell}_d (t_i)\right)$ ($\ell=1,2$).
Moreover, by \eqref{eq. Kallianpur-Striebel formula for the full model discrete observation} and \eqref{eq. Kallianpur-Striebel formula discrete time observations and reduced model}, we can conclude that for any bounded measurable function $\phi$, there hold 
\begin{align}
	\hat f^{N,\gamma}_{\phi, d,t_i}(Y^{N,\gamma}_{d,1:i}) 
	&=
	\frac{\hat g^{N,\gamma}_{\phi, d,t_i}(Y^{N,\gamma}_{d,1:i}) }{\hat g^{N,\gamma}_{1, d,t_i}(Y^{N,\gamma}_{d,1:i})}, 
	\qquad \mathbb P\text{-a.s}, ~\mathbb P^{N,\gamma}_d\text{-a.s.}
	\qquad \forall\gamma\in\mathbb R,
	\label{eq. f function Kallianpur-Striebel formula full model discrete time}
	\\
	\hat f^{\gamma_\ell}_{\phi, d,t_i}(Y^{\gamma_\ell}_{d,1:i})
	&=
	\frac{\hat g^{\gamma_\ell}_{\phi, d,t_i}(Y^{\gamma_\ell}_{d,0:i})}{\hat g^{\gamma_\ell}_{1, d,t_i}(Y^{\gamma_\ell}_{d,1:i})},
	\qquad \mathbb P\text{-a.s}, ~\mathbb P^{\gamma_{\ell}}_d\text{-a.s.}
	\qquad~~ \ell=1,2.
	\label{eq. f function Kallianpur-Striebel formula reduced model discrete time}
\end{align}

\section{The proof of \Cref{thm convergence of artificial filters}}{\label{Sec the convergence of theoretical filters}}
In this section, we present the proof of \Cref{thm convergence of artificial filters}. 

\subsection{The proof framework}
To prove the result, we borrow the proof framework proposed in \cite{calzolari2006approximation}, which deals with the convergence of filters via constructing auxiliary probability spaces that have nice properties.
Specifically, the framework establishes a common probability space, on which some random variables are constructed to mimic the system state and observations on the natural probability space.
Then, by showing the convergence of the filters on the common probability space, one can arrive at the convergence of the filters on the natural probability space. 

For continuous observations, the framework requires us to construct random variables $\tilde {\mathcal K}^{N, \gamma_\ell}$, $\tilde {\mathcal K}^{\gamma_\ell}$, $\tilde X^{N,\gamma_\ell}(\cdot)$, $\tilde X^{\gamma_\ell}(\cdot)$, $\tilde Y^{\gamma_\ell}_c(\cdot)$, $\tilde Z^{N,\gamma_\ell}_c(t)$, $\tilde Z^{\gamma_\ell}_c(t)$ on a common probability space $\left(\tilde \Omega^{\gamma_\ell}_t, \tilde{\mathcal F}^{\gamma_\ell}_t, \tilde{ \mathbb Q}^{\gamma_\ell}_t \right)$ for any $t>0$ and $\ell=1,2$ such that
\begin{enumerate}[label=(A.1.\arabic*), itemindent=1em]
	\item $\left( \tilde {\mathcal K}^{N, \gamma_\ell}, \tilde  X^{N,\gamma_\ell}(\cdot), \tilde  Y^{\gamma_\ell}_c(\cdot), \tilde Z^{N,\gamma_\ell}_c(t)\right)\in \mathbb R^{r}_{> 0} \times \mathbb D_{\mathbb R^{n}}[0,t] \times \mathbb D_{\mathbb R^{m}}[0,t] \times \mathbb R$ has the same law as $\left( \mathcal K,  X^{N,\gamma_\ell}(\cdot),  Y^{N,\gamma_\ell}_c(\cdot), Z^{N,\gamma_\ell}_c(t)\right) $ on  $\left( \Omega, {\mathcal F}_{t}, {\mathbb{ P}}^{N,\gamma_\ell}_c \right)$, {\label{eq. assumption 1.1 of the big frame work}}
	\item $\left( \tilde {\mathcal K}^{ \gamma_\ell}, \tilde  X^{\gamma_\ell}(\cdot), \tilde  Y^{\gamma_\ell}_c(\cdot), \tilde Z^{\gamma_\ell}_c(t)\right)\in \mathbb R^{r}_{> 0} \times \mathbb D_{\mathbb R^{n}}[0,t] \times \mathbb D_{\mathbb R^{m}}[0,t] \times \mathbb R$ has the same law as $\left( \mathcal K,  X^{\gamma_\ell}(\cdot),  Y^{\gamma_\ell}_c(\cdot), Z^{\gamma_\ell}_c(t)\right) $ on $\left( \Omega, {\mathcal F}_{t}, {\mathbb{ P}}^{\gamma_\ell}_c \right)$,  {\label{eq. assumption 1.2 of the big frame work}}
	\item $\hat f_{\phi,c,t}^{N,\gamma_\ell}(\tilde Y^{\gamma_\ell}_{c,0:t}) \to  \hat f_{\phi,c,t}^{\gamma_\ell}(\tilde Y^{\gamma_\ell}_{c,0:t})$ in $\tilde{ \mathbb Q}^{\gamma_\ell}_t$-probability where $\tilde Y^{\gamma_\ell}_{c,0:t}$ is the trajectory of $\tilde Y^{\gamma_\ell}_{c}(\cdot)$ from time 0 to t, \label{eq. assumption 1.3 of the big frame work}
	\item $\lim_{N\to \infty}\mathbb E_{\tilde{ \mathbb Q}^{\gamma_\ell}_t}\left[\left|\tilde Z^{N,\gamma_\ell}_c(t)-\tilde Z^{\gamma_\ell}_c(t)\right|\right]=0$. \label{eq. assumption 1.4 of the big frame work}
\end{enumerate}
For discrete-time observations, the framework requires us to construct random variables $\tilde {\mathcal K}^{N, \gamma_\ell}$, $\tilde {\mathcal K}^{\gamma_\ell}$, $\tilde X^{N,\gamma_\ell}(\cdot)$, $\tilde X^{\gamma_\ell}(\cdot)$, $\tilde Y^{\gamma_\ell}_{d}(\cdot)$, $\tilde Z^{N,\gamma_\ell}_d(t_i)$, $\tilde Z^{\gamma_\ell}_d(t_i)$ in a common probability space $\left(\tilde \Omega^{\gamma_\ell}_{t_i}, \tilde{\mathcal F}^{\gamma_\ell}_{t_i}, \tilde{ \mathbb Q}^{\gamma_\ell}_{t_i} \right)$ for any $i\in\mathbb N_{>0}$ and $\ell=1,2$ such that
\begin{enumerate}[label=(A.2.\arabic*), itemindent=1em]
	\item $\left( \tilde {\mathcal K}^{N,\gamma_\ell}, \tilde  X^{N,\gamma_\ell}(\cdot), \tilde  Y^{\gamma_\ell}_{d,1:i}, \tilde Z^{N,\gamma_\ell}_d(t_i)\right)\in \mathbb R^{r}_{\geq 0} \times \mathbb D_{\mathbb R^n}[0,t_i] \times \mathbb R^{m\times i} \times \mathbb R$,
	where $\tilde  Y^{\gamma_\ell}_{d,1:i}=\left(\tilde  Y^{\gamma_\ell}_{d}(t_1),\dots, \tilde  Y^{\gamma_\ell}_{d}(t_i)\right)$,
	 has the same law as the random variables $\left( \mathcal K,  X^{N,\gamma_\ell}(\cdot),  Y^{N,\gamma_\ell}_{d,1:i}, Z^{N,\gamma_\ell}_d(t_i)\right)$ on $\left( \Omega, {\mathcal F}_{t_i}, {\mathbb{ P}}^{N,\gamma_2} \right)$, {\label{eq. assumption 2.1 of the big frame work}}
	\item $\left( \tilde {\mathcal K}^{\gamma_\ell}, \tilde  X^{\gamma_\ell}(\cdot), \tilde  Y^{\gamma_\ell}_{d,1:i}, \tilde Z^{\gamma_\ell}_d(t_i)\right)\in \mathbb R^{r}_{\geq 0} \times \mathbb D_{\mathbb R^n}[0,t_i] \times \mathbb R^{m\times i} \times \mathbb R$ has the same law as $\left( \mathcal K,  X^{\gamma_\ell}(\cdot),  Y^{\gamma_\ell}_{d,1:i}, Z^{\gamma_\ell}_d(t_i)\right)$ on $\left( \Omega, {\mathcal F}_{t_i}, {\mathbb{ P}}^{\gamma_2} \right)$, {\label{eq. assumption 2.2 of the big frame work}}
	\item $\hat f_{\phi,d,t_i}^{N,\gamma_\ell}(\tilde Y^{\gamma_\ell}_{d,1:i}) \to  \hat f_{\phi,d,t_i}^{\gamma_\ell}(\tilde Y^{\gamma_\ell}_{d,1:i})$ in $\tilde{ \mathbb Q}^{\gamma_\ell}_{t_i}$-probability. \label{eq. assumption 2.3 of the big frame work}
	\item $\lim_{N\to \infty}\mathbb E_{\tilde{ \mathbb Q}^{\gamma_\ell}_{t_i}}\left[\left|\tilde Z^{N,\gamma_\ell}_d(t_i)-\tilde Z^{\gamma_\ell}_d(t_i)\right|\right]=0$. \label{eq. assumption 2.4 of the big frame work}
\end{enumerate}
If we succeed in finding the above random variables, then the convergence results in \Cref{thm convergence of artificial filters} are guaranteed by the following theorem.
Consequently, the rest of this section contributes to constructing random variables such that the above conditions are satisfied.
\begin{theorem}[Adapted from \cite{calzolari2006approximation}] \label{theorem framework first time scale}~
	\begin{enumerate}
		\item 	For both $\ell=1,2$, a time point $t>0$, and a particular measurable function $\phi$, if \ref{eq. assumption 1.1 of the big frame work}, \ref{eq. assumption 1.2 of the big frame work}, \ref{eq. assumption 1.3 of the big frame work}, and \ref{eq. assumption 1.4 of the big frame work} are satisfied, then $ f^{N,\gamma_\ell}_{\phi,c,t}\left(Y^{N,\gamma_{\ell}}_{c,0:t}\right)-f^{\gamma_\ell}_{\phi,c,t}\left(Y^{N,\gamma_{\ell}}_{c,0:t}\right)
		\stackrel{\mathbb P}{\to} 0$  as $N$ goes to infinity;
		\item For both $\ell=1,2$, a integer $i>0$, and a particular measurable function $\phi$, if \ref{eq. assumption 2.1 of the big frame work}, \ref{eq. assumption 2.2 of the big frame work}, \ref{eq. assumption 2.3 of the big frame work}, and \ref{eq. assumption 2.4 of the big frame work} are satisfied, then $f^{N,\gamma_\ell}_{\phi,d,t_i}\left(Y^{N,\gamma_{\ell}}_{d,1:i}\right)-f^{\gamma_\ell}_{\phi,d,t_i}\left(Y^{N,\gamma_{\ell}}_{d,1:i}\right) \stackrel{\mathbb P}{\to} 0$  as $N$ goes to infinity;
	\end{enumerate}
\end{theorem}
\begin{proof}
	The proofs of these two results are in the same spirit. Therefore, we only show the proof of the first result.
	
	Let us construct a probability measure $\tilde {\mathbb{ P}}^{\gamma_\ell}_t$ on the measurable space $\left(\tilde \Omega^{\gamma_\ell}_t, \tilde{\mathcal F}^{\gamma_\ell}_t\right)$ by $\frac{\dd \tilde {\mathbb{ P}}^{\gamma_\ell}_t}{\dd \tilde {\mathbb{ Q}}^{\gamma_\ell}_t}=\tilde Z^{\gamma_\ell}_c(t)$.
	Note that $Z^{\gamma_\ell}_c(t)$ is $\mathbb P^{\gamma_\ell}_c$-almost surely positive, so $\tilde Z^{\gamma_\ell}_c(t)$ is $\tilde {\mathbb Q}^{\gamma_\ell}_t$-almost surely positive by \ref{eq. assumption 1.2 of the big frame work}, which implies the measures $\tilde {\mathbb{ P}}^{\gamma_\ell}_t$ and $\tilde {\mathbb{ Q}}^{\gamma_\ell}_t$ to be equivalent.
	
	Let us denote $\epsilon_N=\hat f^{N,\gamma_\ell}_{\phi,c,t}\left(Y^{N,\gamma_{\ell}}_{c,0:t}\right)-\hat f^{\gamma_\ell}_{\phi,c,t}\left(Y^{N,\gamma_{\ell}}_{c,0:t}\right)$ and $\tilde \epsilon_N=\hat f^{N,\gamma_\ell}_{\phi,c,t}\left( \tilde Y^{\gamma_\ell}_{c,0:t}\right)-\hat f^{\gamma_\ell}_{\phi,c,t}\left( \tilde Y^{\gamma_\ell}_{c,0:t}\right)$.
	Then, for any $\delta>0$, we can calculate that
	\begin{align*}
	\lim_{N\to \infty}\mathbb{ P} \left( \|\epsilon_{N} \| > \delta \right)
	&=  \lim_{N\to \infty}\mathbb E_{\tilde {\mathbb Q}^{\gamma_\ell}_t} \left[\left(\tilde Z^{N,\gamma_\ell}_c(t)-\tilde Z^{\gamma_\ell}_c(t)\right)\mathbbold 1_{\{ \|\tilde \epsilon_{N}\|\geq \delta \}}\right] 
	+ \lim_{N\to \infty} \mathbb E_{\tilde {\mathbb Q}^{\gamma_\ell}_t} \left[ \tilde Z^{\gamma_\ell}_c(t) \mathbbold 1_{\{ \|\tilde \epsilon_{N}\|\geq \delta \}}\right]\\
	&\leq \lim_{N\to \infty} \mathbb E_{\tilde {\mathbb Q}^{\gamma_\ell}_t} \left[ \left|\left(\tilde Z^{N,\gamma_\ell}_c(t)-\tilde Z^{\gamma_\ell}_t(t)\right) \right| \right] 
	+  {\tilde {\mathbb P}^{\gamma_\ell}}_t \left( \|\tilde \epsilon_{N}\|\geq \delta \right)
	\end{align*}
	where the equality follows from \ref{eq. assumption 1.1 of the big frame work}.
	Finally, the last line in the above formula tends to zero by \ref{eq. assumption 1.3 of the big frame work}, \ref{eq. assumption 1.4 of the big frame work}, and the equivalence between $\tilde {\mathbb{ P}}^{\gamma_\ell}_t$ and $\tilde {\mathbb{ Q}}^{\gamma_\ell}_t$, which proves the result.
\end{proof}

\subsection{Constructing auxiliary random variables}{\label{subsection new random variables}}
This subsection contributes to constructing random variables and probability spaces that satisfy \ref{eq. assumption 1.1 of the big frame work}, \ref{eq. assumption 1.2 of the big frame work}, \ref{eq. assumption 2.1 of the big frame work} and \ref{eq. assumption 2.2 of the big frame work}.
We first apply Skorokhod's representation theorem to construct the required random variables.
\begin{lemma}{\label{lemma applying Skorokhod's representation theorem}}
	~
	\begin{enumerate}
		\item 	If  \eqref{eq. assumption initial conditons} and \eqref{eq. assumption infinite explosion time gamma1} hold, then for any $t>0$,
		there exist $\mathbb R^{r}_{>0}\times \mathbb D_{\mathbb R^{n}}[0,t]$-valued random variables $\left(\tilde {\mathcal K}^{N, \gamma_1}, \tilde X^{N,\gamma_1}(\cdot) \right)$ and $\left( \tilde {\mathcal K}^{\gamma_1}, \tilde X^{\gamma_1}(\cdot) \right)$ defined on a common probability space $\left(\tilde \Omega_{t,1}^{\gamma_1}, \tilde{\mathcal F}_{t,1}^{\gamma_1}, \tilde{\mathbb Q}_{t,1}^{\gamma_1} \right)$, such that $\left(\tilde {\mathcal K}^{N, \gamma_1}, \tilde X^{N,\gamma_1}(\cdot) \right)$ and $\left( \tilde {\mathcal K}^{\gamma_1}, \tilde X^{\gamma_1}(\cdot) \right)$  have the same laws as $\left( {\mathcal K},  X^{N,\gamma_1}(\cdot) \right)$ and $\left( {\mathcal K}, X^{\gamma_1}(\cdot) \right)$, respectively, and 
		$$\lim_{N\to\infty}\left(\tilde {\mathcal K}^{N, \gamma_1}, \tilde X^{N,\gamma_1}(\cdot) \right)=\left( \tilde {\mathcal K}^{\gamma_1}, \tilde X^{\gamma_1}(\cdot) \right)
		\qquad
		\tilde{\mathbb Q}_{t,1}^{\gamma_1}\text{-almost surely.}
		$$ 
		\item If conditions \eqref{eq. assumption initial conditons}, \eqref{eq. assumption for the second time scale 2}, \eqref{eq. derivative of V1}, \eqref{eq. assumption egordicity}, \eqref{eq. assumption technical 1}, \eqref{eq. assumption technical 2}, and \eqref{eq. assumption infinite explosion time gamma2} are satisfied, 
		then for any $t>0$, 
		there exist $\mathbb R^{r}_{>0}\times \mathbb D_{\mathbb R^{n}}[0,t]$-valued random variables $\left(\tilde {\mathcal K}^{N, \gamma_2}, \tilde X^{N,\gamma_2}(\cdot) \right)$ and $\left( \tilde {\mathcal K}^{\gamma_2}, \tilde X^{\gamma_2}(\cdot) \right)$ defined on a common probability space $\left(\tilde \Omega_{t,1}^{\gamma_2}, \tilde{\mathcal F}_{t,1}^{\gamma_2}, \tilde{\mathbb Q}_{t,1}^{\gamma_2} \right)$, such that $\left(\tilde {\mathcal K}^{N, \gamma_2}, \tilde X^{N,\gamma_2}(\cdot) \right)$ and $\left( \tilde {\mathcal K}^{\gamma_2}, \tilde X^{\gamma_2}(\cdot) \right)$  have the same laws as $\left( {\mathcal K},  X^{N,\gamma_2}(\cdot) \right)$ and $\left( {\mathcal K}, X^{\gamma_2}(\cdot) \right)$, respectively, and 
		$$\lim_{N\to\infty}\left(\tilde {\mathcal K}^{N, \gamma_2}, \Pi_2\tilde X^{N,\gamma_2}(\cdot) \right)=\left( \tilde {\mathcal K}^{\gamma_2}, \tilde X^{\gamma_2}(\cdot) \right)
        \qquad \tilde{\mathbb Q}_{t,1}^{\gamma_2} \text{-almost surely.}
        $$		
	\end{enumerate}
\end{lemma}
\begin{proof}
	By \Cref{prop Kang kurtz gamma1} and \Cref{prop kang kurtz gamma2}, we have $(\mathcal K,X^{N,\gamma_{1}}(\cdot)) \Rightarrow(\mathcal K,X^{\gamma_{1}}(\cdot)) $ and $(\mathcal K,\Pi_2X^{N,\gamma_{2}}(\cdot)) \Rightarrow(\mathcal K,X^{\gamma_{2}}(\cdot)) $.
	Therefore, by Skorokhod's representation theorem \cite[Theorem 1.8 in Chapter 3]{ethier1986markov}, the result is proven.
\end{proof}

In continuous-time observation scenarios, we further term $\tilde Y^{\gamma_\ell}_c(\cdot)$ ($\ell=1,2$) as $m$-vectors of independent standard Brownian motions on a certain filtered probability space, denoted as $\left(\tilde \Omega^{\gamma_\ell}_{c,2}, \tilde{\mathcal F}^{\gamma_\ell}_{c,2}, \{\tilde{\mathcal F}^{\gamma_\ell}_{t,c,2}\}_t, \tilde{ \mathbb Q}^{\gamma_\ell}_{c,2} \right)$.
Furthermore, for any $t>0$, we define product probability spaces $\left(\tilde \Omega^{\gamma_\ell}_t, \tilde{\mathcal F}^{\gamma_\ell}_t, \tilde{ \mathbb Q}^{\gamma_\ell}_t \right)
\triangleq
\left(\tilde \Omega^{\gamma_\ell}_{t,1}\times \tilde \Omega^{\gamma_\ell}_{c,2}, \tilde{\mathcal F}^{\gamma_\ell}_{t,1}\times \tilde{\mathcal F}^{\gamma_\ell}_{t,c,2}, \tilde{ \mathbb Q}^{\gamma_\ell}_{t,1}
\times \tilde{ \mathbb Q}^{\gamma_\ell}_{c,2} \right)
$
and Girsanov's random variables
\begin{align*}
\tilde Z^{N, \gamma_\ell}_c(t) &\triangleq \exp\left(
\int_0^t h^{\top}\left(\tilde X^{N,\gamma_\ell}(s)\right) \dd \tilde Y^{\gamma_\ell}_c - \frac{1}{2} \int_0^t \| h(\tilde X^{N,\gamma_\ell}(s))  \|^2 \dd s
\right) 
&& \ell =1,2
\\
\tilde Z^{\gamma_\ell}_c(t) &\triangleq \exp\left(
\int_0^t h^{\top}(\tilde X^{\gamma_\ell}(s)) \dd \tilde Y^{\gamma_\ell}_c - \frac{1}{2} \int_0^t \| h(\tilde X^{\gamma_\ell}(s))  \|^2 \dd s
\right) 
&& \ell =1,2
\end{align*}
on these product probability spaces.
Then, we can show that the above-defined random variables satisfy \ref{eq. assumption 1.1 of the big frame work} and \ref{eq. assumption 1.2 of the big frame work}.

\begin{lemma}{\label{lem A.1.1 A.1.2}}[Verifying \ref{eq. assumption 1.1 of the big frame work} and \ref{eq. assumption 1.2 of the big frame work}]
	~
	\begin{enumerate}
		\item If  \eqref{eq. assumption initial conditons} and \eqref{eq. assumption infinite explosion time gamma1} hold, then for any $t>0$, the above-defined probability space $\left(\tilde \Omega^{\gamma_1}_t, \tilde{\mathcal F}^{\gamma_1}_t, \tilde{ \mathbb Q}^{\gamma_1}_t \right)$ and random variables $\tilde {\mathcal K}^{N, \gamma_1}$, $\tilde {\mathcal K}^{\gamma_1}$, $\tilde X^{N,\gamma_1}(\cdot)$, $\tilde X^{\gamma_1}(\cdot)$, $\tilde Y^{\gamma_1}_c(\cdot)$, $\tilde Z^{N,\gamma_1}_c(t)$, and $\tilde Z^{\gamma_1}_c(t)$ satisfy \ref{eq. assumption 1.1 of the big frame work} and \ref{eq. assumption 1.2 of the big frame work} for $\ell=1$.
		\item If \eqref{eq. assumption initial conditons}, \eqref{eq. assumption for the second time scale 2}, \eqref{eq. derivative of V1}, \eqref{eq. assumption egordicity}, \eqref{eq. assumption technical 1}, \eqref{eq. assumption technical 2}, and \eqref{eq. assumption infinite explosion time gamma2} hold, then for any $t>0$, the above-defined probability space $\left(\tilde \Omega^{\gamma_2}_t, \tilde{\mathcal F}^{\gamma_2}_t, \tilde{ \mathbb Q}^{\gamma_2}_t \right)$ and random variables $\tilde {\mathcal K}^{N, \gamma_2}$, $\tilde {\mathcal K}^{\gamma_2}$, $\tilde X^{N,\gamma_2}(\cdot)$, $\tilde X^{\gamma_2}(\cdot)$, $\tilde Y^{\gamma_2}_c(\cdot)$, $\tilde Z^{N,\gamma_2}_c(t)$, and $\tilde Z^{\gamma_2}_c(t)$ satisfy \ref{eq. assumption 1.1 of the big frame work} and \ref{eq. assumption 1.2 of the big frame work} for $\ell=2$.
	\end{enumerate}
\end{lemma}
\begin{proof}
	This lemma follows immediately from \Cref{lemma applying Skorokhod's representation theorem} and definitions of these random variables.
\end{proof}

Similarly, in discrete time observations scenarios, we term $\{\tilde Y^{\gamma_\ell}_d(t_i)\}_{i\in\mathbb N_{>0}}$ ($\ell=1,2$) as mutually independent $m$-variate Gaussian random variables whose covariance matrices equal to the identical matrix.
We denote the filtered probability space where $\{\tilde Y^{\gamma_\ell}_d(t_i)\}_{i\in\mathbb N_{>0}}$ is located by $\left(\tilde \Omega^{\gamma_\ell}_{d}, \tilde{\mathcal F}^{\gamma_\ell}_{d,2}, \{\tilde{\mathcal F}^{\gamma_\ell}_{t_i,2}\}_i, \tilde{ \mathbb Q}^{\gamma_\ell}_{d,2} \right)$.
Moreover, for any $i\in\mathbb N_{>0}$, we define probability spaces $\left(\tilde \Omega^{\gamma_\ell}_{t_i}, \tilde{\mathcal F}^{\gamma_\ell}_{t_i}, \tilde{ \mathbb Q}^{\gamma_\ell}_{t_i} \right)
\triangleq
\left(\tilde \Omega^{\gamma_\ell}_{t_i,1}\times \tilde \Omega^{\gamma_\ell}_{d,2}, \tilde{\mathcal F}^{\gamma_\ell}_{t_i,1}\times \tilde{\mathcal F}^{\gamma_\ell}_{t_i,d,2}, \tilde{ \mathbb Q}^{\gamma_\ell}_{t_i,1}
\times \tilde{ \mathbb Q}^{\gamma_\ell}_{d,2} \right)
$
and random variables
\begin{equation*}
\tilde Z^{N, \gamma_\ell}_d(t_i)= \prod_{j=1}^{i} H\left(\tilde X^{N, \gamma_\ell}(t_j),\tilde Y^{\gamma_\ell}(t_j)\right) \quad \text{~and~} \quad 
\tilde Z^{\gamma_\ell}_d(t_i)= \prod_{j=1}^{i} H\left(\tilde X^{\gamma_\ell}(t_j),\tilde Y^{\gamma_\ell}(t_j)\right)
\end{equation*}
for $\ell=1,2$.
Then, we can also show that \ref{eq. assumption 2.1 of the big frame work} and \ref{eq. assumption 2.2 of the big frame work} are satisfied under some mild conditions.
\begin{lemma}{\label{lem A.2.1 A.2.2}}[Verifying \ref{eq. assumption 2.1 of the big frame work} and \ref{eq. assumption 2.2 of the big frame work}]
~
\begin{enumerate}
	\item 	If  \eqref{eq. assumption initial conditons} and \eqref{eq. assumption infinite explosion time gamma1} hold, then for any $i\in\mathbb N_{>0}$, the above-defined probability space $\left(\tilde \Omega^{\gamma_1}_{t_i}, \tilde{\mathcal F}^{\gamma_1}_{t_i}, \tilde{ \mathbb Q}^{\gamma_1}_{t_i} \right)$ and random variables $\tilde {\mathcal K}^{N, \gamma_1}$, $\tilde {\mathcal K}^{\gamma_1}$, $\tilde X^{N,\gamma_1}(\cdot)$, $\tilde X^{\gamma_1}(\cdot)$, $\tilde Y^{\gamma_1}_{d}(\cdot)$, $\tilde Z^{N,\gamma_1}_d(t_i)$, $\tilde Z^{\gamma_1}_d(t_i)$ satisfy \ref{eq. assumption 2.1 of the big frame work} and \ref{eq. assumption 2.2 of the big frame work} for $\ell=1$.
	\item If \eqref{eq. assumption initial conditons}, \eqref{eq. assumption for the second time scale 2}, \eqref{eq. derivative of V1}, \eqref{eq. assumption egordicity}, \eqref{eq. assumption technical 1}, \eqref{eq. assumption technical 2}, and \eqref{eq. assumption infinite explosion time gamma2} hold, then for any $i\in\mathbb N_{>0}$, the above-defined probability space $\left(\tilde \Omega^{\gamma_2}_{t_i}, \tilde{\mathcal F}^{\gamma_2}_{t_i}, \tilde{ \mathbb Q}^{\gamma_2}_{t_i} \right)$ and random variables $\tilde {\mathcal K}^{N, \gamma_2}$, $\tilde {\mathcal K}^{\gamma_2}$, $\tilde X^{N,\gamma_2}(\cdot)$, $\tilde X^{\gamma_2}(\cdot)$, $\tilde Y^{\gamma_2}_{d}(\cdot)$, $\tilde Z^{N,\gamma_2}_d(t_i)$, $\tilde Z^{\gamma_2}_d(t_i)$ satisfy \ref{eq. assumption 2.1 of the big frame work} and \ref{eq. assumption 2.2 of the big frame work} for $\ell=2$.
\end{enumerate}
\end{lemma}
\begin{proof}
	This lemma follows immediately from \Cref{lemma applying Skorokhod's representation theorem} and definitions of these random variables.
\end{proof}

\subsection{The convergence of the constructed auxiliary random variables}

This subsection contributes to showing that the random variables constructed above satisfy \ref{eq. assumption 1.3 of the big frame work} \ref{eq. assumption 1.4 of the big frame work}, \ref{eq. assumption 2.3 of the big frame work}, and \ref{eq. assumption 2.4 of the big frame work} and, therefore, prove \Cref{thm convergence of artificial filters}.
We first prove \ref{eq. assumption 1.4 of the big frame work} and \ref{eq. assumption 2.4 of the big frame work}.

\begin{lemma}{\label{lemma A.1.4}}[Verifying \ref{eq. assumption 1.4 of the big frame work} and \ref{eq. assumption 2.4 of the big frame work}]
	~
	\begin{enumerate}
		\item For $\ell=1$ (the first time scale), if  \eqref{eq. assumption initial conditons} and \eqref{eq. assumption infinite explosion time gamma1} hold, then the random variables and probability spaces defined in \ref{subsection new random variables} satisfy \ref{eq. assumption 1.4 of the big frame work} for any $t>0$ and \ref{eq. assumption 2.4 of the big frame work} for any $i\in\mathbb N_{>0}$..
		\item  For $\ell=2$ (the second time scale), if the conditions of the second part of \Cref{thm convergence of artificial filters} hold, then the random variables and probability spaces defined in \ref{subsection new random variables} satisfy \ref{eq. assumption 1.4 of the big frame work} for any $t>0$ and \ref{eq. assumption 2.4 of the big frame work} for any $i\in\mathbb N_{>0}$.	
	\end{enumerate}
\end{lemma}
\begin{proof}
	The proofs for \ref{eq. assumption 1.4 of the big frame work} and \ref{eq. assumption 2.4 of the big frame work} are identical; therefore, we only show the proof for \ref{eq. assumption 1.4 of the big frame work} and leave the other for readers.
	
	By Cauchy-Schwarz inequality, we have that 
	\begin{align*}
	\mathbb E_{\tilde{\mathbb Q}^{\gamma_\ell}_t} \left[\left|\tilde Z^{N,\gamma_\ell }_c(t)- 	\tilde Z^{\gamma_\ell }_c(t) \right|\right] 
	\leq \sqrt{ \mathbb E_{\tilde{\mathbb Q}^{\gamma_\ell}_t} \left[\left(\tilde Z^{N,\gamma_\ell}_c(t)/\tilde Z^{\gamma_\ell}_c(t) - 	1\right)^{2}\right]\mathbb E_{\tilde{\mathbb Q}^{\gamma_\ell}_t}\left[\left(\tilde Z^{\gamma_\ell}_c(t)\right)^{2} \right]}. 
	\end{align*}
	for $\ell=1,2$. 
	Note that $\mathbb E_{\tilde{\mathbb Q}^{\gamma_\ell}_t}\left[\left(\tilde Z^{\gamma_\ell}_c(t)\right)^{2}\right]$ is bounded due to the boundedness of $h(\cdot)$.
	Therefore, to prove the result, we only need to show $\mathbb E_{\tilde{\mathbb Q}^{\gamma_\ell}_t} \left[\left(\tilde Z^{N,\gamma_\ell}_c(t)/\tilde Z^{\gamma_\ell}_c(t) - 	1\right)^{2}\right]$ to converge to zero.
	By the definition of $\tilde Z^{N,\gamma_\ell}_c(t)$ and $\tilde Z^{\gamma_\ell}_c(t)$, we can arrived at
	\begin{align*}
		 &\mathbb E_{\tilde{\mathbb Q}^{\gamma_\ell}_t} \left[\left(\tilde Z^{N,\gamma_\ell}_c(t)/\tilde Z^{\gamma_\ell}_c(t) - 	1\right)^{2}\right] \\
		&=
		\mathbb E_{\tilde{\mathbb Q}^{\gamma_\ell}_t} \left[
		\mathbb E_{\tilde{\mathbb Q}^{\gamma_\ell}_t}
		\left[
		\left(\tilde Z^{N,\gamma_\ell}_c(t)/\tilde Z^{\gamma_\ell}_c(t) - 	1\right)^{2}
		\left| \tilde{\mathcal F}^{\gamma_\ell}_{t,1} \right.
		\right]
		\right] \\
		&= \sum_{j=0}^{2} 
		\left(
		\begin{matrix}
		 j \\ 2
		\end{matrix}
		\right)
		(-1)^{j} \mathbb E_{\tilde{\mathbb Q}^{\gamma_\ell}_t} \Bigg[ \exp \left(
		     \frac{j}{2} \int_0^t \| h(\tilde X^{N,\gamma_\ell}(s))-h(\tilde X^{\gamma_\ell}(s))\|^2 \dd s 
		      + \frac{j}{2} \int_0^t \| h(\tilde X^{\gamma_\ell}(s))\|^2 - \| h(\tilde X^{N,\gamma_\ell}(s))\|^2 \dd s
		\right)\Bigg]
	\end{align*}
	where $		\left(
	\begin{matrix}
	j \\ 2
	\end{matrix}
	\right)$ are binomial coefficients, the first equality follows from the law of total expectation, and the second equality follows from the fact that $\tilde Y^{\gamma_\ell}(\cdot)$ are $m$-vectors of independent standard Brownian motions and independent of $\tilde X^{N,\gamma_\ell}(\cdot)$ and $\tilde X^{\gamma_\ell}(\cdot)$.
    Finally, the right hand side of the last equality goes to zero by the boundedness and Lipschitz continuity of $h(\cdot)$ and the almost sure convergence of $\tilde X^{N,\gamma_\ell}(\cdot)$ to $\tilde X^{\gamma_\ell}(\cdot)$ (see \Cref{lemma applying Skorokhod's representation theorem}). (For $\ell=2$, we also need the condition $h(x)=h(\Pi_2 x)$ for all $x\in\mathbb R^{n}$.)
    Consequently, $\mathbb E_{\tilde{\mathbb Q}^{\gamma_\ell}_t} \left[\left(\tilde Z^{N,\gamma_\ell}_c(t)/\tilde Z^{\gamma_\ell}_c(t) - 	1\right)^{2}\right]\to 0$ and the lemma is shown.
\end{proof}

Then, we work on \ref{eq. assumption 1.3 of the big frame work} and \ref{eq. assumption 2.3 of the big frame work} using the fact that $A_NB_N\to AB$ in probability if $A_N\to A$ and $B_N \to B$ in probability.
Specifically, we use this fact to prove \ref{eq. assumption 1.3 of the big frame work} by taking
\begin{equation*}
	A_N=\hat g^{N,\gamma_\ell}_{\phi, c,t}(\tilde Y^{\gamma_{\ell}}_{c,0:t}),
	\quad
	A=\hat g^{\gamma_\ell}_{\phi, c,t}(\tilde Y^{\gamma_{\ell}}_{c,0:t}),
	\quad
	B_N=1/\hat g^{N,\gamma_\ell}_{1, c,t}(\tilde Y^{\gamma_{\ell}}_{c,0:t}),
	\quad
	B=1/\hat g^{\gamma_\ell}_{1, c,t}(\tilde Y^{\gamma_{\ell}}_{c,0:t})
\end{equation*}
 (see \eqref{eq. f function Kallianpur-Striebel formula full model continuous time} and \eqref{eq. f function Kallianpur-Striebel formula reduced model continuous time}). 
To prove \ref{eq. assumption 2.3 of the big frame work}, we take
\begin{equation*}
	A_N=\hat g^{N,\gamma_\ell}_{\phi, d,t_i}(\tilde Y^{\gamma_{\ell}}_{d,1:i}),
	~~
	A=\hat g^{\gamma_\ell}_{\phi, d,t_i}(\tilde Y^{\gamma_{\ell}}_{d,1:i}),
	~~
	B_N=1/\hat g^{N,\gamma_\ell}_{1, d,t_i}(\tilde Y^{\gamma_{\ell}}_{d,1:i}),
	~~
	B=1/\hat g^{\gamma_\ell}_{1, d,t_i}(\tilde Y^{\gamma_{\ell}}_{d,1:i})
\end{equation*}
(see \eqref{eq. f function Kallianpur-Striebel formula full model discrete time} and \eqref{eq. f function Kallianpur-Striebel formula reduced model discrete time}).
In the following propositions, we show $A_N\to A$ in probability, which, together with the continuous mapping theorem, also suggests $B_N\to B$ in probability.

\begin{proposition}{\label{proposition An-A continuous time}}
	~
	\begin{enumerate}
		\item At the first time scale, if \eqref{eq. assumption initial conditons} and \eqref{eq. assumption infinite explosion time gamma1} hold, then for any bounded continuous function $\phi$, there hold
		\begin{equation*}
		\hat g^{N,\gamma_1}_{\phi, c,t}(\tilde Y^{\gamma_{1}}_{c,0:t}) \stackrel{\tilde {\mathbb Q}^{\gamma_1}_t}{\longrightarrow}  \hat g^{\gamma_1}_{\phi, c,t}(\tilde Y^{\gamma_{1}}_{c,0:t})
		~~ \text{as ~$N\to\infty$}, ~~ \forall t>0, 
		\quad \text{and}\quad
		  \hat g^{N,\gamma_1}_{\phi, d,t_i}(\tilde Y^{\gamma_{1}}_{d,1:i}) \stackrel{\tilde {\mathbb Q}^{\gamma_1}_{t_i}}{\longrightarrow}  \hat g^{\gamma_1}_{\phi, d,t_i}(\tilde Y^{\gamma_{1}}_{d,1:i})
		 ~~ \text{as ~$N\to\infty$}, ~~ \forall i\in\mathbb N_{>0}.
		\end{equation*}
		Moreover by the continuous mapping theorem, there hold
		\begin{equation*}
		\frac{1}{\hat g^{N,\gamma_1}_{1, c,t}(\tilde Y^{\gamma_{1}}_{c,0:t})} \stackrel{\tilde {\mathbb Q}^{\gamma_1}_t}{\longrightarrow}  \frac{1}{\hat g^{\gamma_1}_{1, c,t}(\tilde Y^{\gamma_{1}}_{c,0:t})}
		~~ \text{as ~$N\to\infty$}, ~~ \forall t>0, 
		\quad
		\text{and}
		\quad 
		\frac{1}{\hat g^{N,\gamma_1}_{1, d,t_i}(\tilde Y^{\gamma_{1}}_{d,1:i})} \stackrel{\tilde {\mathbb Q}^{\gamma_1}_{t_i}}{\longrightarrow}  \frac{1}{\hat g^{\gamma_1}_{1, d,t_i}(\tilde Y^{\gamma_{1}}_{d,1:i})}
		~~ \text{as ~$N\to\infty$}, ~~ \forall i\in\mathbb N_{>0}.
		\end{equation*}
		\item At the second time scale, assume that the conditions of the second part of \Cref{thm convergence of artificial filters} hold. Then, for any bounded continuous function $\phi$ such that $\phi(\kappa,x)=\phi(\kappa, \Pi_2 x)$ $\forall(\kappa,x)\in \mathbb R^{r} \times \mathbb R^{n}$, there hold
		\begin{equation*}
		\hat g^{N,\gamma_2}_{\phi, c,t}(\tilde Y^{\gamma_{2}}_{c,0:t}) \stackrel{\tilde {\mathbb Q}^{\gamma_2}_t}{\longrightarrow}  \hat g^{\gamma_2}_{\phi, c,t}(\tilde Y^{\gamma_{2}}_{c,0:t})
	   ~~ \text{as ~$N\to\infty$}, ~~ \forall t>0, 
       \quad
		\text{and}
		\quad 
		\frac{1}{\hat g^{N,\gamma_1}_{1, d,t_i}(\tilde Y^{\gamma_{1}}_{d,1:i})} \stackrel{\tilde {\mathbb Q}^{\gamma_1}_{t_i}}{\longrightarrow}  \frac{1}{\hat g^{\gamma_1}_{1, d,t_i}(\tilde Y^{\gamma_{1}}_{d,1:i})}
		~~ \text{as ~$N\to\infty$}, ~~ \forall i\in\mathbb N_{>0}.
		\end{equation*}
		Moreover by the continuous mapping theorem, there hold
		\begin{equation*}
		\frac{1}{\hat g^{N,\gamma_2}_{1, c,t}(\tilde Y^{\gamma_{2}}_{c,0:t})} \stackrel{\tilde {\mathbb Q}^{\gamma_2}_t}{\longrightarrow}  \frac{1}{\hat g^{\gamma_2}_{1, c,t}(\tilde Y^{\gamma_{2}}_{c,0:t})}
		~~ \text{as ~$N\to\infty$}, ~~ \forall t>0, 
        \quad
		\text{and}
		\quad 
		\frac{1}{\hat g^{N,\gamma_2}_{1, d,t_i}(\tilde Y^{\gamma_{2}}_{d,1:i})} \stackrel{\tilde {\mathbb Q}^{\gamma_2}_{t_i}}{\longrightarrow}  \frac{1}{\hat g^{\gamma_2}_{1, d,t_i}(\tilde Y^{\gamma_{2}}_{d,1:i})}
		~~ \text{as ~$N\to\infty$}, ~~ \forall i\in\mathbb N_{>0}.
		\end{equation*}
	\end{enumerate}
\end{proposition}
\begin{proof}
	The proofs for the results in the continuous-time and discrete-time scenarios are the same. Therefore, we only present the proof for the result in the continuous-time case and leave the other for readers.
	
	By \Cref{eq. definition of g full model continuous time }, \eqref{eq. definition of g reduced mode continuous time}, and \Cref{lem A.1.1 A.1.2} (i.e., \ref{eq. assumption 1.1 of the big frame work} and \ref{eq. assumption 1.2 of the big frame work}), we have that 
	\begin{align*}
	\hat g^{N,\gamma_\ell}_{\phi, c,t}(\tilde Y^{\gamma_{\ell}}_{c,0:t}) &=
	\mathbb E_{\tilde {\mathbb Q}^{\gamma_\ell}_t} \left[ \tilde Z^{N,\gamma_{\ell}}_{c}(t) \phi\left(\tilde {\mathcal K}^{N,\gamma_\ell}, \tilde X^{N,\gamma_{\ell}}(t)\right) \left| \tilde{\mathcal Y}^{\gamma_{\gamma_{\ell}}}_{c,t} \right.\right]
	\qquad
	\tilde {\mathbb Q}^{\gamma_\ell}_t\text{-a.s.}
	\\
	\hat g^{\gamma_\ell}_{\phi, c,t}(\tilde Y^{\gamma_\ell}_{c,0:t})
	&=
	\mathbb E_{\tilde {\mathbb Q}^{\gamma_\ell}_t} \left[ \tilde Z^{\gamma_\ell}_c(t) \phi \left(\tilde {\mathcal K}^{\gamma_{\ell}}, \tilde X^{\gamma_\ell}(t)\right) \left| \tilde {\mathcal Y}^{\gamma_\ell}_{c,t} \right. \right]
	\qquad 	
	\tilde {\mathbb Q}^{\gamma_\ell}_t\text{-a.s.}
	\end{align*}
    for any $t>0$, $ \ell=1,2$, and any continuous function $\phi$ given in the lemma.
    In the above formulas, $\tilde {\mathcal Y}^{\gamma_\ell}_{c,t}$ is the filtration generated by the process $\{\tilde { Y}^{\gamma_\ell}_{c}(t)\}_{t>0}$.
	Therefore, by Jensen's inequality and the triangle inequality, we can further arrive at 
	\begin{align*}
		&\mathbb E_{\tilde {\mathbb Q}^{\gamma_\ell}_t}
		\left[\left|
		\hat g^{N,\gamma_\ell}_{\phi, c,t}(\tilde Y^{\gamma}_{c,0:t})-\hat g^{\gamma_\ell}_{\phi, c,t}(\tilde Y^{\gamma_\ell}_{c,0:t})
		\right|
		\right]\\
		&\leq 
		\mathbb E_{\tilde {\mathbb Q}^{\gamma_\ell}_t}
		\left[\left|
		    \left(\tilde Z^{\gamma_\ell}_c(t)-\tilde Z^{N,\gamma_{\ell}}_{c}(t) \right) 
		    \phi (\tilde {\mathcal K}^{\gamma_{\ell}}, \tilde X^{N,\gamma_\ell}(t)) 
		 \right|
		 \right]
		  + 
		 \mathbb E_{\tilde {\mathbb Q}^{\gamma_\ell}_t}
		 \left[\left|
		 \tilde Z^{\gamma_{\ell}}_{c}(t)
		 \left(
		\phi\left(\tilde {\mathcal K}^{N,\gamma_\ell}, \tilde X^{N,\gamma_{\ell}}(t)\right) -\phi \left(\tilde {\mathcal K}^{\gamma_{\ell}}, \tilde X^{\gamma_\ell}(t)\right) 
		 \right)
		 \right|
		 \right]\\
		 &\leq 
		\|\phi \|_{\infty} \mathbb  E_{\tilde {\mathbb Q}^{\gamma_\ell}_t}
		\left[\left|
		\left(\tilde Z^{\gamma_\ell}_c(t)-\tilde Z^{N,\gamma_{\ell}}_{c}(t) \right) 
		\right|
		\right]
		 +
		 \mathbb E_{\tilde {\mathbb P}^{\gamma_\ell}_t}
		 \left[\left|
		 \phi\left(\tilde {\mathcal K}^{N,\gamma_\ell}, \tilde X^{N,\gamma_{\ell}}(t)\right) -\phi \left(\tilde {\mathcal K}^{\gamma_{\ell}}, \tilde X^{\gamma_\ell}(t)\right) 
		 \right|
		 \right]
	\end{align*}
	where $\|\phi \|_{\infty}$ is the supremum norm of $\phi$, and the measure $\tilde {\mathbb{ P}}^{\gamma_\ell}_t$ is defined by  $\frac{\dd \tilde {\mathbb{ P}}^{\gamma_\ell}_t}{\dd \tilde {\mathbb{ Q}}^{\gamma_\ell}_t}=\tilde Z^{\gamma_\ell}_c(t)$.
	Note that $\tilde {\mathbb{ P}}^{\gamma_\ell}_t$ and $\tilde {\mathbb Q}^{\gamma_\ell}_t$ are equivalent due to the almost sure positiveness of $\tilde Z^{\gamma_\ell}_c(t)$.
	Finally, terms on the right hand side of the second inequality converge to zero by \Cref{lemma A.1.4}, \Cref{lemma applying Skorokhod's representation theorem}, and the dominant convergence theorem. 
\end{proof}

By \Cref{proposition An-A continuous time}, we can verify conditions \ref{eq. assumption 1.3 of the big frame work}
and \ref{eq. assumption 2.3 of the big frame work}.

\begin{lemma}\label{lem A.1.3}[Verifying \ref{eq. assumption 1.3 of the big frame work} and \ref{eq. assumption 2.3 of the big frame work}]
	~
	\begin{itemize}
		\item For $\ell=1$ (the first time scale), assume that \eqref{eq. assumption initial conditons} and \eqref{eq. assumption infinite explosion time gamma1} hold. Then, for any $t>0$ and any bounded continuous function $\phi$, the random variables and probability space defined in \ref{subsection new random variables} satisfy \ref{eq. assumption 1.3 of the big frame work} and \ref{eq. assumption 2.3 of the big frame work}.
		
		\item For $\ell=2$ (the second time scale), assume that the conditions of the second part of \Cref{thm convergence of artificial filters} hold. Then, for any $t>0$ and any bounded continuous function $\phi$ such that $\phi(\kappa,x)=\phi(\kappa, \Pi_2 x)$ $\forall(\kappa,x)\in \mathbb R^{r} \times \mathbb R^{n}$, the random variables and probability space defined in \ref{subsection new random variables} satisfy \ref{eq. assumption 1.3 of the big frame work} and \ref{eq. assumption 2.3 of the big frame work}.
	
	\end{itemize}
\end{lemma}
\begin{proof}
	By \eqref{eq. f function Kallianpur-Striebel formula full model continuous time} and \eqref{eq. f function Kallianpur-Striebel formula reduced model continuous time} and \Cref{lem A.1.1 A.1.2}, we have that 
	\begin{align*}
	\hat f^{N,\gamma_{\ell}}_{\phi, c,t}(\tilde Y^{N,\gamma_{\ell}}_{c,0:t}) 
	&=
	\frac{\hat g^{N,\gamma_{\ell}}_{\phi, c,t}(\tilde Y^{\gamma_{\ell}}_{c,0:t}) }{\hat g^{N,\gamma_{\ell}}_{1, c,t}(\tilde Y^{\gamma_\ell}_{c,0:t}) } 
	&& \text{for all bounded measurable function $\phi$ and $\ell=1,2$},
	\\
	\hat f^{\gamma_\ell}_{\phi, c,t}(\tilde Y^{\gamma_\ell}_{c,0:t})
	&=
	\frac{\hat g^{\gamma_\ell}_{\phi, c,t}(\tilde Y^{\gamma_\ell}_{c,0:t})}{\hat g^{\gamma_\ell}_{1, c,t}(\tilde Y^{\gamma_\ell}_{c,0:t})}
	&& \text{for all bounded measurable function $\phi$ and $\ell=1,2$}, 
	\end{align*}
	$\mathbb Q^{\gamma_\ell}_{t}$-a.s..
    Therefore, \ref{eq. assumption 1.3 of the big frame work} follows immediately from \Cref{proposition An-A continuous time}.
    Similarly, we can also show \ref{eq. assumption 2.3 of the big frame work}, which proves the lemma.
\end{proof}

\subsection{The proof of the target theorem}

With these preparations, we can finally prove \Cref{thm convergence of artificial filters}.

\begin{proof}[The proof of \Cref{thm convergence of artificial filters}]
	The result follows immediately from \Cref{theorem framework first time scale}, \Cref{lem A.1.1 A.1.2}, \Cref{lem A.2.1 A.2.2}, \Cref{lemma A.1.4}, and \Cref{lem A.1.3}.
\end{proof}

\section{The proof of the continuous-time part of  \Cref{thm convergence of particle filters}}{\label{section convergence of particle filters}}This section contributes to showing the convergence of the proposed continuous-time SIR particle filters, i.e., the continuous-time part of \Cref{thm convergence of particle filters}.
For diffusion systems, the convergence of continuous-time SIR particle filters is shown in \cite[Chapter 9]{bain2008fundamentals} by using backward Zakai's equation.
In this section, we use a similar proof scheme to extend the result to piecewise deterministic Markov processes.
Though we still follow the main idea of backward stochastic differential equations, we present the proof in the forward fashion so that notations can be simplified, and the discussion about the existence of the solution of the backward stochastic differential equation can be avoided.
Throughout this section, we assume that there hold the non-explosive condition \eqref{eq. assumption infinite explosion time gamma1} for the first time scale and \eqref{eq. assumption infinite explosion time gamma2} for the second time scale.
Also, the scaling factor $N$ is fixed in this section, so we drop the argument $N$ for most variables that have a dependence on $N$ to simplify the notation.

In general, the proof follows from the law of large numbers.
We first show the filter generated by these particles can accurately approximate the true filter under some mild conditions (see \Cref{thm help}), which require 
\begin{itemize}[itemindent=0em]
	\item the particle filter to accurately approximate the true filter at the initial time \ref{cc1},
	\item the error between the particle filter and the true one to grow mildly over time \ref{cc2},
	\item the resampling to perturb the particles mildly \ref{cc3}. 
\end{itemize}
Fortunately, these conditions can be easily verified by the law of large numbers and the nature of the residual resampling method, which finally proves the result (see the proof of \Cref{thm main results} at the end of this section). 

Notably, when analyzing the propagation of the particle filter's error over time  (i.e., \ref{cc2}), we need to use particles obtained at an earlier time to predict the particle filter's performance at a later time.
The particle filter at the later time contains some information about the observation that is unavailable at the earlier time. 
As a result, we need to extend the capability of the particle filter so that it can generate an estimate of a random variable that is measurable with respect to a $\sigma$-algebra at a later time. 
This fact leads to some complicated notations, which we introduce as follows.

For particle filters $\bar \pi^{N,\gamma_\ell}_{M,c,t}(\phi)$ $(\ell=1,2)$ (see \Cref{alg continuous time regularized particle filters} and \Cref{def particle filters}), we term variables $\left(\kappa_j^{\gamma_\ell} (t),  x^{\gamma_\ell}_j(t), z^{\gamma_{\ell}}_j(t), w^{\gamma_{\ell}}_j(t)\right)$ (for $t>0$) as particles (the first two variables), Girsanov-like variables, and weights obtained in the sampling step, $\left(\bar \kappa_j^{\gamma_\ell} (i\delta),  \bar x^{\gamma_\ell}_j(i\delta)\right)$ as resampled particles for $i\in\mathbb N_{>0}$ or initial particles for $i=0$,
$\bar {\mathcal F}^{\gamma_\ell}_{c,t}$ ($i\in\mathbb N_{>0}$) as the $\sigma$-algebra generated by random variables $\{\left(\kappa_j^{\gamma_\ell} (t),  x^{\gamma_\ell}_j(t) \right)\}_{j=1,\dots, M}$, 
and $\bar {\mathcal F}^{\gamma_\ell}_{c,(i\delta)^+}$ (for $i\in\mathbb N_{>0}\cup\{0\}$) as the $\sigma$-algebra generated by all random variables $\left\{\left(\bar \kappa_j^{\gamma_\ell} (i\delta),  \bar x^{\gamma_\ell}_j(i\delta)\right)\right\}_{j=1,\dots,M}$.
Also, for any measurable function $\psi_{t,\tau}: \mathbb R^r\times \mathbb R^n \times \mathbb D_{\mathbb R^m}[t,\tau] \to \mathbb R$ where $\tau>t$, we denote
\begin{align*}
	\bar \pi^{N,\gamma_\ell}_{M,c,t,\tau}(\psi_{t,\tau}) 
	&\triangleq 
	\sum_{j=1}^{M} w^{\gamma_{\ell}}_j(t) \psi_{t,\tau}\left( \kappa_j^{\gamma_\ell} (t),  x^{\gamma_\ell}_j(t), Y^{N,\gamma_{\ell}}_{c,t:\tau}\right), 
	&&
	\bar p^{N,\gamma_\ell}_{M,c,t,\tau}(\psi_{t,\tau}) 
	\triangleq 
	\sum_{j=1}^{M} z^{\gamma_{\ell}}_j(t) \psi_{t,\tau}\left( \kappa_j^{\gamma_\ell} (t),  x^{\gamma_\ell}_j(t), Y^{N,\gamma_{\ell}}_{c,t:\tau}\right), \\
	\hat p^{N,\gamma_\ell}_{M,c,i\delta,\tau}(\psi_{i\delta,\tau})
	&\triangleq 
	\sum_{j=1}^{M} \frac{1}{M}\psi_{t,\tau}\left( \kappa_j^{\gamma_\ell} (i\delta),  x^{\gamma_\ell}_j(i\delta), Y^{N,\gamma_{\ell}}_{c,i\delta:\tau}\right),
	~~  i\in\mathbb N_{>0} \cup \{0\}.
\end{align*}
where $Y^{N,\gamma_{\ell}}_{c,t:\tau}$ is the trajectory of $Y^{N,\gamma_{\ell}}_{c}$ from time t to $\tau$.
The quantities $\bar \pi^{N,\gamma_\ell}_{M,c,t,\tau}(\psi_{t,\tau}) $ and $\hat p^{N,\gamma_\ell}_{M,c,i\delta,\tau}(\psi_{i\delta,\tau})$ utilize particles to generate estimates of conditional expectations and are, respectively, $\bar {\mathcal F}^{\gamma_\ell}_{c,t}\bigvee \mathcal Y^{N,\gamma}_{c,\tau}$ and $\bar {\mathcal F}^{\gamma_\ell}_{c,(i\delta)^+}\bigvee \mathcal Y^{N,\gamma}_{c,\tau}$ measurable. 
In this section, we only consider those $\psi_{t,\tau}$ such that 
\begin{equation}{\label{eq. condition psi}}
	b_{\psi_{t,\tau}} \triangleq 
	\sup_{(\kappa,x)\in \mathbb R^{r}\times \mathbb R^n}
    \mathbb E_{\mathbb P} \left[ \left|
    \psi_{t,\tau}\left( \kappa,  x, Y^{N,\gamma_{\ell}}_{c,t:\tau}\right) 
    \right|^2
    \right] < +\infty.
\end{equation}
Note that when $\psi_{t,\tau}(\kappa,x,y)=\phi(\kappa,x)$, the variable $\bar \pi^{N,\gamma_\ell}_{M,c,t,\tau}(\psi_{t,\tau})$ results in the SIR particle filter $\bar \pi^{N,\gamma_\ell}_{M,c,t}(\phi)$, and, obviously, such a $\psi_{t,\tau}$ satisfies \eqref{eq. condition psi} if $\phi$ is uniformly bounded.

For $\ell\in\{1,2\}$, we construct random variables $(\hat {\mathcal K}, \hat{X}^{\gamma_\ell}(\cdot))$ in the natural probability space such that they have the same law as $(\mathcal K, X^{\gamma_\ell}(\cdot))$ and are independent of $Y^{N,\gamma_{\ell}}_c(\cdot)$.
Recall that processes $Y^{\gamma_{\ell}}_c(\cdot)$ and $Y^{\gamma_{N,\ell}}_c(\cdot)$ (under $\mathbb P^{\gamma_{\ell}}_c$ and $\mathbb P^{N,\gamma_{\ell}}_c$, respectively) are m-vectors of independent standard Brownian motions, and, moreover, the former is independent of the underlying system $(\mathcal K, X^{\gamma_\ell}(\cdot))$
(\ref{Sec change of measure methods}).
Then, by \eqref{eq. definition of g reduced mode continuous time}, \eqref{eq. f function Kallianpur-Striebel formula reduced model continuous time} and the equivalence between $\mathbb P$ and $\mathbb P^{N,\gamma_{\ell}}_c$, there holds
\begin{equation}{\label{eq. f hat proof in the particle filter}}
\hat f^{\gamma_{\ell}}_{\phi,c,t}\left(Y^{N,\gamma_{\ell}}_{c,0:t} \right)
={\mathbb E_{\mathbb P} \left[ \hat Z^{\gamma_\ell}_c(t) \phi (\hat{\mathcal K}, \hat{X}^{\gamma_\ell}(t)) \left| \mathcal Y^{N,\gamma_\ell}_{c,t} \right. \right]}
~\Bigg/~
{\mathbb E_{\mathbb P} \left[ \hat Z^{\gamma_\ell}_c(t)  \left| \mathcal Y^{N,\gamma_\ell}_{c,t} \right. \right]} 
\qquad \mathbb P\text{-a.s.} \mathbb,~ \mathbb P^{N,\gamma_{\ell}}_c\text{-a.s.}. 
\end{equation}
where
\begin{align*}
	\hat Z^{\gamma_\ell}_c(t) &\triangleq \exp\left(
	\int_0^t h^{\top}(X^{\gamma_\ell}(s)) \dd Y^{N, \gamma_\ell}_c(s) - \frac{1}{2} \int_0^t \| h(X^{\gamma_\ell}(s))  \|^2 \dd s
	\right)
\end{align*}
Furthermore, by the boundedness of $h(\cdot)$, we can show the moments of $\hat Z^{\gamma_\ell}_c(t_2)/ \hat Z^{\gamma_\ell}_c(t_1)$ ($0\leq t_1\leq t_2$) to be uniformly bounded in the sense that
\begin{equation}{\label{eq. boundedness of moments of Z hat}}
	\mathbb E_{\mathbb P} \left[\left(\hat Z^{\gamma_\ell}_c(t_2)/ \hat Z^{\gamma_\ell}_c(t_1)\right)^{q}
	\left|
	\hat {\mathcal K} = \kappa, \hat X^{\gamma_{\ell}}(t_1)=x
	\right.
	\right]
	\leq e^{{2mq^2\|h\|^2_{\infty}(t_2-t_1)}} 
\end{equation}
for any integer $q$ and any $(\kappa, x) \in\mathbb R^r \times \mathbb R^n$
where $\|\cdot\|_{\infty}$ is the supremum norm.

Also, for any measurable function $\psi_{t,\tau}: \mathbb R^r\times \mathbb R^n \times \mathbb D_{\mathbb R^m}[t,\tau] \to \mathbb R$, we term
\begin{align}
	\tilde p^{N,\gamma_\ell}_{c,t,\tau}(\psi_{t,\tau}) &\triangleq  {\mathbb E_{\mathbb P} \left[ \hat Z^{\gamma_\ell}_c(t) \psi_{t,\tau} (\hat{\mathcal K}, \hat X^{\gamma_\ell}(t), Y^{N,\gamma_{\ell}}_{c,t,\tau}) \left| \mathcal Y^{N,\gamma_\ell}_{c,\tau} \right. \right]}
	\label{eq. p tilde psi}
	\\
	\tilde \pi^{N,\gamma_\ell}_{c,t,\tau}(\psi_{t,\tau}) &\triangleq \frac{\tilde p^{N,\gamma_\ell}_{M,c,t,\tau}(\psi_{t,\tau})}
	{\tilde p^{N,\gamma_\ell}_{M,c,t,\tau}(1)}. {\label{eq. pi tilde psi}}
\end{align}
Moreover, for any measurable $\psi_{t,\tau}$ and any $i\in\mathbb N_{>0}$ such that $i\delta< t$, we define function $\bar \psi^{t}_{i\delta,\tau}: \mathbb R^r\times \mathbb R^n \times \mathbb D_{\mathbb R^m}[i\delta,\tau] \to \mathbb R$ as 
\begin{align}
	&\bar \psi^{t}_{i\delta,\tau}\left(\kappa, x, y\right)= \mathbb E_{\mathbb P} \Bigg[ \frac{\hat Z^{\gamma_\ell}_c(t)}{\hat Z^{\gamma_\ell}_c(i \delta)} \psi_{t,\tau} \left(\hat{\mathcal K},  \hat X^{\gamma_\ell}(t), Y^{N,\gamma_{\ell}}_{c,t:\tau}\right) \label{eq bar psi}
	\left| \hat {\mathcal K} = \kappa, \hat X^{\gamma_{\ell}}(i \delta)=x, Y^{N,\gamma_\ell}_{c,i \delta : \tau}=y\right. \Bigg] 
\end{align}
which, by Jensen's inequality, the law of total expectation, and \eqref{eq. boundedness of moments of Z hat}, satisfies 
\begin{equation}{\label{eq. boundedness bar psi}}
	b_{\bar \psi^{t}_{i\delta,\tau}} \triangleq 
	\sup_{(\kappa,x)\in \mathbb R^{r}\times \mathbb R^n}
	\mathbb E_{\mathbb P} \left[ \left|
	\bar \psi^{t}_{i\delta,\tau}\left( \kappa,  x, Y^{N,\gamma_{\ell}}_{c,i\delta :\tau}\right) 
	\right|^2
	\right] 
	\leq
	e^{{8m\|h\|^2_{\infty}(t-i\delta)}} b_{\psi_{t,\tau}}
\end{equation}
if $\psi_{t,\tau}$ satisfies \eqref{eq. condition psi}.
Moreover, using the law of total expectation and Markov property of $\hat X^{\gamma_{\ell}}(\cdot)$ (inherited from the Markov property of $X^{\gamma_{\ell}}(\cdot)$), one also has
\begin{equation}{\label{eq. tilde p recursive relation}}
		\tilde p^{N,\gamma_\ell}_{c,t,\tau}(\psi_{t,\tau})
		=\tilde p^{N,\gamma_\ell}_{c,i\delta,\tau}\left(\bar \psi^{t}_{i\delta,\tau}\right)
\end{equation}
where $i\delta<t$.
Notably, in cases where $\psi_{t,\tau}(\kappa,x,y)=\phi(\kappa,x)$ and $\phi$ is bounded, $\tilde \pi^{N,\gamma_\ell}_{c,t,\tau}(\psi_{t,\tau}) $ leads to $\hat f^{\gamma_{\ell}}_{\phi,c,t}\left(Y^{N,\gamma_{\ell}}_{c,0:t} \right)$ (see \eqref{eq. f hat proof in the particle filter} and \eqref{eq. pi tilde psi}).
Consequently, our task can be transformed into showing the convergence of $\bar \pi^{N,\gamma_\ell}_{M,c,t,\tau}(\psi_{t,\tau})$ to
$\tilde \pi^{N,\gamma_\ell}_{c,t,\tau}(\psi_{t,\tau})$ as $M\to\infty$.

We first present a theorem adapted from \cite{crisan2001particle}.

\begin{theorem}[adapted from \cite{crisan2001particle}]\label{thm help}
	We assume that $X^{\gamma_{\ell}}$ ( $\ell\in\{1,2\}$) is almost surely non-explosive.
	Then, for any $t>0$ and any bounded measurable function $\phi$, there holds the relation
	\begin{equation*}
		\mathbb E_{\mathbb P}\left[ \left|\bar \pi^{N,\gamma_\ell}_{M,c,t}(\phi)-
		\hat f^{\gamma_{\ell}}_{\phi,c,t}\left(Y^{N,\gamma_{\ell}}_{c,0:t} \right)\right| \right]
		\leq \frac{ \tilde c_t}{\sqrt M} \| \phi \|_{\infty}
	\end{equation*}
   where $\tilde c_t$ is a time dependent number, if the following conditions are satisfied.
   \begin{enumerate}[label=(C.C.\arabic*), itemindent=1em]
   	\item For any $\tau>0$ and any measurable function $\psi_{0,\tau}: \mathbb R^r\times \mathbb R^n \times \mathbb D_{\mathbb R^m}[0,\tau] \to \mathbb R$ satisfying \eqref{eq. condition psi}, there holds 
   	$$
   	\mathbb E_{\mathbb P} \left[\left|
   	\hat p^{N,\gamma_\ell}_{M,c,0,\tau}(\psi_{0,\tau})- 
   	 \psi_{t,\tau}\left( \mathcal K,  X^{\gamma_1}(0), Y^{N,\gamma_{\ell}}_{c,0:\tau}\right) 
   	\right|^2
   	\right]
   	\leq 
   	\frac{b_{\psi_{0,\tau}}}{{M}}
   	$$ \label{cc1}
   	\item 
   	For any $\tau>0$ and any measurable function $\psi_{t,\tau}: \mathbb R^r\times \mathbb R^n \times \mathbb D_{\mathbb R^m}[t,\tau] \to \mathbb R$ satisfying \eqref{eq. condition psi}, there holds
   	$$
   	\mathbb E_{\mathbb P} \left[\left|
   	\bar p^{N,\gamma_\ell}_{M,c,t,\tau}(\psi_{t,\tau}) 
   	-
   	\hat p^{N,\gamma_\ell}_{M,c,\lfloor t/\delta \rfloor\delta,\tau}\left(\bar \psi_{\lfloor t/\delta \rfloor\delta,\tau}^t\right)
   	\right|^2
   	\right]
   	\leq   
    \frac{ c_2 }{M} b_{\psi_{i\delta,\tau}}
   	$$
   	where $\bar \psi_{\lfloor t/\delta \rfloor\delta,\tau}$ is defined by \eqref{eq bar psi} and $c_2$ is a constant. \label{cc2}
   	\item For any $i\in\mathbb N_{>0}$ and any measurable function $\psi_{i\delta,\tau}: \mathbb R^r\times \mathbb R^n \times \mathbb D_{\mathbb R^m}[i\delta,\tau] \to \mathbb R$ satisfying \eqref{eq. condition psi}, there holds 
   	$$ \mathbb E_{\mathbb P} \left[\left| 
   	\bar p^{N,\gamma_\ell}_{M,c,i\delta,\tau}(\psi_{i\delta,\tau})
   	-
   	\hat p^{N,\gamma_\ell}_{M,c,i\delta,\tau}(\psi_{i\delta,\tau})
   	\right|^2\right] \leq \frac{c_3 }{M } b_{\psi_{i\delta,\tau}}
   	$$
   	where $c_3$ is a constant.  \label{cc3}
   \end{enumerate}
\end{theorem}

\begin{proof}	
	The proof is almost the same as the one of \cite[Theorem 2.3.1]{crisan2001particle}; therefore, we only single out the proof framework and leave the details for readers. 
	
	In principle, the proof is shown by mathematical induction. Specifically, the condition \ref{cc1} guarantees the particle filter to accurately approximate the true filter at the initial time.
	The condition \ref{cc2} suggests the error between the particle filter and the true filter to grow mildly over time.
	The condition \ref{cc3} suggests that resampling perturbs the particles mildly so that the particle filter's accuracy will not be influenced too much. 
	Notably, compared with \cite[Theorem 2.3.1]{crisan2001particle}, this theorem considers the boundedness of $\psi_{t,\tau}$ in the sense of \eqref{eq. condition psi} rather than the supremum norm. 
	Consequently, all the discussions concerning boundedness in the supremum norm in the proof of \cite[Theorem 2.3.1]{crisan2001particle} are replaced by the argument of \eqref{eq. condition psi}.
\end{proof}

Finally, by verifying all the conditions in the above theorem, we can prove the continuous-time part of \Cref{thm convergence of particle filters}.

\begin{proof}[The proof of the continuous-time part of \Cref{thm convergence of particle filters}]
	To prove the result, we need to verify all the conditions in \Cref{thm help}.
	The condition \ref{cc1} follows immediately from the central limit theorem. 
	The condition \ref{cc2} follows from the fact that 
	\begin{align*}
		& \mathbb E_{\mathbb P} \left[\left(
		\bar p^{N,\gamma_\ell}_{M,c,t,\tau}(\psi_{t,\tau}) 
		-
		\hat p^{N,\gamma_\ell}_{M,c,\lfloor t/\delta \rfloor\delta,\tau}\left(\bar \psi_{\lfloor t/\delta \rfloor\delta,\tau}^t\right)
		\right)^2
		\left| \bar {\mathcal F}^{\gamma_\ell}_{c,(\lfloor t/\delta \rfloor\delta)^{+}} \right. 
		\right]\\
		&= \frac{1}{M^2}  \mathbb E_{\mathbb P} \left[ \sum_{j=1}^{M}
	    \left(z^{\gamma_{\ell}}_j(t) \psi_{t,\tau}\left( \kappa_j^{\gamma_\ell} (t),  x^{\gamma_\ell}_j(t), Y^{N,\gamma_{\ell}}_{c,t:\tau}\right)\right)^2
		\left| \bar {\mathcal F}^{\gamma_\ell}_{c,(\lfloor t/\delta \rfloor\delta)^{+}} \right. 
		\right]
		-
		\frac{1}{M} \mathbb E_{\mathbb P} \left[
		\left(\hat p^{N,\gamma_\ell}_{M,c,\lfloor t/\delta \rfloor\delta,\tau}\left(\bar \psi_{\lfloor t/\delta \rfloor\delta,\tau}^t\right)\right)^2
		\left| \bar {\mathcal F}^{\gamma_\ell}_{c,(\lfloor t/\delta \rfloor\delta)^+} \right. 
		\right] \\
		&\leq \frac{	\exp\left({{8m\|h\|^2_{\infty}(t-\lfloor t/\delta \rfloor\delta)}}\right)}{M} b_{\psi_{t,\tau}}
	\end{align*}
	where the equality is due to the independence of the motion of the particles, and the inequality follows from \eqref{eq. boundedness bar psi}.
	Finally, if we use the multinomial resampling, then we can have the relation 
	\begin{align*}
		\mathbb E_{\mathbb P} \left[\left(
		\bar p^{N,\gamma_\ell}_{M,c,i\delta,\tau}(\psi_{i\delta,\tau})
		-
		\hat p^{N,\gamma_\ell}_{M,c,i\delta,\tau}(\psi_{i\delta,\tau})
		\right)^2 
		\left| \bar {\mathcal F}^{\gamma_\ell}_{c,i\delta} \vee \mathcal Y^{N,\gamma_\ell}_{c,i\delta} \right. \right]  
		& \leq  \frac{1}{M^2}
		\mathbb E_{\mathbb P} \left[ \sum_{j=1}^{M} \left(\psi_{i\delta,\tau}(\kappa^{\gamma_\ell}_j(i\delta),x^{\gamma_{\ell}}_j(i\delta),Y^{N,\gamma_{\ell}}_{c,t,\tau})\right)^2  
		\left| \bar {\mathcal F}^{\gamma_\ell}_{c,i\delta} \right.
		\right]
		 \leq \frac{1}{M} b_{\psi_{t,\tau}}
	\end{align*}
	due to the fact that all the resampled particles are independently and identically distributed under the multinomial resampling. 
	Moreover, by \cite[Exerice 9.1]{bain2008fundamentals}, the above relation also works for the residual resampling, and, therefore, the condition \ref{cc3} holds with $c_3=1$.
\end{proof}

\section*{CRediT authorship contribution statement}
\textbf{Zhou Fang}: Conceptualization, Methodology, Software, Formal analysis, Writing - Original Draft, Writing - Review \& Editing.
\textbf{Ankit Gupta}: Conceptualization, Writing - Review \& Editing, Supervision, Funding acquisition.
\textbf{Mustafa Khammash}: Conceptualization, Writing - Review \& Editing, Supervision, Funding acquisition.

\section*{Acknowledgments}
We acknowledge funding from the Swiss National Science Foundation under grant 182653.

\bibliographystyle{model1-num-names}
\bibliography{references}

\end{document}